\documentclass[compsoc,conference,a4paper,10pt,times]{IEEEtran}
\usepackage{xspace}
\usepackage{xparse}
\usepackage{amsfonts,amsmath}
\usepackage{stackrel}
\usepackage{mathtools}
\usepackage{amsthm}
\usepackage[english]{babel}
\usepackage{amssymb}
\usepackage{txfonts}
\usepackage{enumitem}
\usepackage{xcolor}
\usepackage{array}
\usepackage{tikz}
\usetikzlibrary{arrows, positioning}
\usepackage{pgfplots, pgfplotstable}
\pgfplotsset{compat=1.9}
\usepgfplotslibrary{groupplots}
\usetikzlibrary{patterns}
\usepackage{pstricks}
\usepackage{subcaption}
\usepackage{wrapfig}
\usepackage[utf8]{inputenc}
\usepackage{booktabs}
\usepackage{expl3}
\usepackage{collcell}
\usepackage[hyphens]{url}
\usepackage{hyperref}
\usepackage{graphicx}
\usepackage{multirow}
\usepackage{rotating}
\usepackage{framed}
\usepackage{verbatim}
\usepackage[square,sort,comma,numbers]{natbib}
\usepackage{mdframed}
\usepackage{algorithm}
\usepackage{etoolbox}
\usepackage{titlesec}
\titlespacing*{\subsection}{0pt}{0.5\baselineskip}{0.5\baselineskip}
\mdfsetup{frametitlealignment=\center}
\usepackage[utf8]{inputenc}
\usepackage{xstring}
\usepackage{fp}
\usepackage{expl3}
\usepackage{collcell}
\usepackage{colortbl} 
\usepackage{siunitx}
\usepackage{pifont}
\newcommand{\yesIndicator}{\ding{51}\xspace}
\newcommand{\noIndicator}{\ding{55}\xspace}

\usepackage{pgfplotstable}
\usepackage{booktabs}
\usepackage[margin=1in]{geometry}
\pgfplotsset{compat=1.18}

\newcounter{note}[section]
\renewcommand{\thenote}{\thesection.\arabic{note}}
\newcommand{\fixmecolor}{red}
\newcommand{\notecolor}{blue}
\newcommand{\mkr}[1]{\refstepcounter{note}{\bf
\textcolor{\fixmecolor}{$\ll$MKR~\thenote: {\sf #1}$\gg$}}}

\newcommand{\xz}[1]{\refstepcounter{note}{\bf
\textcolor{\notecolor}{$\ll$XZ~\thenote: {\sf #1}$\gg$}}}

\newcommand{\assign}{\ensuremath{\gets}\xspace}
\newcommand{\addText}[1]{\textcolor{black}{#1}}

\NewDocumentCommand{\GCcost}{g o}%
{%
  \IfNoValueTF{#2}%
              {\ensuremath{\algNotation{GC}_{#1}}\xspace}%
              {\ensuremath{\algNotation{GC}_{#1}(#2)}\xspace}
}
\NewDocumentCommand{\PSIcost}{o}%
{%
  \IfNoValueTF{#1}%
              {\ensuremath{\algNotation{PSI}}\xspace}
              {\ensuremath{\algNotation{PSI}(#1)}\xspace}
}
\NewDocumentCommand{\FPSIcost}{o}%
{%
  \IfNoValueTF{#1}%
              {\ensuremath{\algNotation{FPSI}}\xspace}
              {\ensuremath{\algNotation{FPSI}(#1)}\xspace}
}
\NewDocumentCommand{\OLEcost}{o}%
{%
  \IfNoValueTF{#1}%
              {\ensuremath{\algNotation{OLE}}\xspace}
              {\ensuremath{\algNotation{OLE}(#1)}\xspace}
}
\NewDocumentCommand{\embedcost}{o}%
{%
  \IfNoValueTF{#1}%
              {\ensuremath{\algNotation{emb}}\xspace}
              {\ensuremath{\algNotation{emb}(#1)}\xspace}
}
\newcommand{\inputSize}{\ensuremath{\ell}\xspace}

\newcommand{\circuit}{\ensuremath{\mathsf{C}}\xspace}
\newcommand{\circuitAlt}{\ensuremath{\circuit'}\xspace}
\newcommand{\garblerInput}{\ensuremath{a}\xspace}
\newcommand{\garblerInputAlt}{\ensuremath{\garblerInput'}\xspace}
\newcommand{\evaluatorInput}{\ensuremath{b}\xspace}
\newcommand{\gate}{\ensuremath{\mathsf{g}}\xspace}
\newcommand{\gateAlt}{\ensuremath{\gate'}\xspace}
\NewDocumentCommand{\wireLabel}{o}
{\IfNoValueTF{#1}
	{\ensuremath{v}\xspace}
	{\ensuremath{v_{#1}}\xspace}
}
\NewDocumentCommand{\wireLabelAlt}{o}
{\IfNoValueTF{#1}
	{\ensuremath{v'}\xspace}
	{\ensuremath{v'_{#1}}\xspace}
}
\newcommand{\truthTableIn}{\ensuremath{w}\xspace}
\newcommand{\truthTableOut}{\ensuremath{u}\xspace}
\newcommand{\wireLabelMask}{\ensuremath{m}\xspace}

\newcommand{\secref}[1]{\mbox{Sec.~\ref{#1}}\xspace}

\newcommand{\figref}[1]{\mbox{Fig.~\ref{#1}}\xspace}

\newcommand{\tblref}[1]{\mbox{Table~\ref{#1}}\xspace}

\newcommand{\appref}[1]{\mbox{App.~\ref{#1}}\xspace}
\newcommand{\eqnref}[1]{\mbox{(\ref{#1})}\xspace}
\newcommand{\eqnsref}[2]{\mbox{(\ref{#1})--(\ref{#2})}\xspace}
\newcommand{\propref}[1]{\mbox{Prop.~\ref{#1}}\xspace}

\newcommand{\dfnref}[1]{\mbox{Definition~\ref{#1}}\xspace}

\newcommand{\lemmaref}[1]{\mbox{Lemma~\ref{#1}}\xspace}
\newcommand{\thmref}[1]{\mbox{Theorem~\ref{#1}}\xspace}
\newcommand{\stepref}[1]{\mbox{Step~\ref{#1}}\xspace}
\newcommand{\stepsref}[2]{\mbox{Steps~\ref{#1}--\ref{#2}}\xspace}

\newcommand{\myparagraph}[1]{\smallskip\noindent\textbf{#1}:\xspace}

\newtheorem{defn}{Definition}
\newtheorem{theorem}{Theorem}[section]
\newtheorem*{theorem*}{Theorem}
\newtheorem{lemma}[theorem]{Lemma}
\newtheorem{prop}[theorem]{Proposition}
\newtheorem*{prop*}{Proposition}
\newenvironment{sproof}{%
\proof}{\endproof}
\newenvironment{proofsketch}{\par\noindent\textsl{Proof (sketch):}\kern1ex}%
{\hfil\hskip2em\penalty250\parfillskip=0pt\finalhyphendemerits=0$\qed$%
	\par\smallskip\par}

\newcommand{\symmUR}{\ensuremath{\mathsf{SUR}}\xspace}
\newcommand{\symmURDist}[2]{\ensuremath{\vectorNorm{\symmUR}{{#1}-{#2}}}\xspace}
\newcommand{\symmURDefn}{\ensuremath{\vectorNorm{\symmUR}{\cdot}}\xspace}

\DeclareMathOperator{\symmDiff}{\triangle}

\NewDocumentCommand{\distance}{o}%
{%
	\IfNoValueTF{#1}%
	{\ensuremath{\fnNotation{d}}\xspace}
	{\ensuremath{\fnNotation{d}_{#1}}\xspace}
}

\NewDocumentCommand{\distanceNotation}{o o o}
{
\IfNoValueTF{#2}
{\ensuremath{\left\Vert . \right\Vert_{\mathsf{#1}}}\xspace}
{
\IfNoValueTF{#3}
{\ensuremath{\left\Vert #2 \right\Vert_{#1}}\xspace}
{\ensuremath{\left\Vert #2 - #3 \right\Vert_{#1}}\xspace}
}
}

\NewDocumentCommand{\hammingDist}{ g g }{\distanceNotation[H][#1][#2]}

\NewDocumentCommand{\embed}{o}%
{%
  \IfNoValueTF{#1}%
              {\ensuremath{\fnNotation{f}}\xspace}
              {\ensuremath{\fnNotation{f}_{#1}}\xspace}
}

\newcommand{\embedESP}{\ensuremath{\embed[\embedKey]}\xspace}

\NewDocumentCommand{\setToVecMap}{o}
{\IfNoValueTF{#1}
{\ensuremath{\mathsf{SetToVec}}\xspace}
{\ensuremath{\mathsf{SetToVec}_{#1}}\xspace}
}
\newcommand{\setOfEvalPts}{\setNotation{X}\xspace}

\NewDocumentCommand{\ESPNodeVec}{ g o
}{\ensuremath{\textsc{esp}({#1}).\mathsf{nvec}\IfNoValueF{#2}{_{#2}}}\xspace}
\NewDocumentCommand{\ESPCharVec}{ g o
}{\ensuremath{\textsc{esp}({#1}).\mathsf{cvec}\IfNoValueF{#2}{_{#2}}}\xspace}
\NewDocumentCommand{\ESPCharSet}{ g o
}{\ensuremath{\textsc{esp}({#1}).\mathsf{cset}\IfNoValueF{#2}{\mathopen{}({#2}
)\mathclose{}}}\xspace}

\NewDocumentCommand{\fileBlock}{o}{\ensuremath{\cInput\IfNoValueF{#1}{_{#1}}}\xspace}
\NewDocumentCommand{\fileBlockAlt}{o}{\ensuremath{\cInputAlt\IfNoValueF{#1}{_{
#1}}}\xspace}
\NewDocumentCommand{\fileBlockVecAlt}{o}
{\IfNoValueTF{#1}
	{\ensuremath{\vecNotation{\fileBlockAlt}}\xspace}
	{\ensuremath{\fileBlockAlt_{#1}}\xspace}
}
\newcommand{\maxFileLen}{\ensuremath{\{0,1\}^{\ast}}\xspace}

\NewDocumentCommand{\cInputVec}{o}
{\IfNoValueTF{#1}
	{\ensuremath{\vecNotation{\cInput}}\xspace}
	{\ensuremath{\vecNotation{\cInput}[#1]}\xspace}
}

\NewDocumentCommand{\cInputVecAlt}{o}
{\IfNoValueTF{#1}
	{\ensuremath{\vecNotation{\cInput}'}\xspace}
	{\ensuremath{\vecNotation{\cInput}'[#1]}\xspace}
}

\NewDocumentCommand{\simHash}{o}
{
	\IfNoValueTF{#1}
	{\ensuremath{\mathsf{SimHash}}\xspace}
	{\ensuremath{\mathsf{SimHash}_{#1}}\xspace}
}
\NewDocumentCommand{\TLSH}{o}
{
	\IfNoValueTF{#1}
	{\ensuremath{\mathsf{LSH}}\xspace}
	{\ensuremath{\mathsf{LSH}_{#1}}\xspace}
}

\newcommand{\numberOfPointsSim}{\ensuremath{\theta}\xspace}
\newcommand{\numberOfPointsSimDefn}{\ensuremath{2\degreeOfKey{1} -
\frac{\threshold}{2} + 2}\xspace}
\newcommand{\metricSpaceESP}{\ensuremath{\finiteField^{\numberOfPointsSim}
}\xspace}

\newcommand{\metricSpaceGen}{\ensuremath{\finiteField^{\metricSpaceDimGen}
}\xspace}
\newcommand{\metricSpaceDimGen}{\ensuremath{\theta}\xspace}

\makeatletter
\newcommand{\vast}{\bBigg@{4}}
\newcommand{\Vast}{\bBigg@{5}}
\makeatother

\newcommand{\oleFunc}{\ensuremath{\mathcal{F}_{\mathit{ole}}}\xspace}
\newcommand{\otFunc}{\ensuremath{\mathcal{F}_{\mathit{GC}}}\xspace}

\newcommand{\npaFunc}{\ensuremath{\mathcal{F}_{\mathit{npa}}}\xspace}

\NewDocumentCommand{\hash}{o}%
{
  \IfNoValueTF{#1}%
	{\ensuremath{\mathsf{H}}\xspace}
	{\ensuremath{\mathsf{H}_{#1}}\xspace}
}

\NewDocumentCommand{\Enc}{o}%
{
  \IfNoValueTF{#1}%
	{\ensuremath{\constNotation{AuthEnc}}\xspace}
	{\ensuremath{\constNotation{AuthEnc}_{#1}}\xspace}
}
\NewDocumentCommand{\PRP}{o}%
{%
	\IfNoValueTF{#1}%
	{\ensuremath{\mathsf{P}}\xspace}
	{\ensuremath{\mathsf{P}_{#1}}\xspace}
}

\NewDocumentCommand{\PRF}{o o}%
{%
	\IfNoValueTF{#1}%
	{\ensuremath{\mathsf{OPRF}\IfNoValueF{#2}{_{#2}}}\xspace}%
{\ensuremath{\ifblank{#1}{\mathsf{OPRF}\IfNoValueF{#2}{_{#2}}}{\mathsf{PRF
}\IfNoValueF{#2}{_{#2}}}}\xspace}%
}

\NewDocumentCommand{\aPAKEOPRF}{o}
{
\IfNoValueTF{#1}
{\ensuremath{\mathsf{PRF'}}\xspace}
{\ensuremath{\mathsf{PRF'_{#1}}}\xspace}
}

\NewDocumentCommand{\genericFnFamily}{ o
}{\ensuremath{\fnNotation{G}\IfNoValueF{#1}{_{#1}}}\xspace}
\newcommand{\genericFnKeyspace}{\ensuremath{\setNotation{K}}\xspace}
\newcommand{\genericFnKey}{\ensuremath{k}\xspace}
\newcommand{\genericFnDomain}{\ensuremath{\setNotation{D}}\xspace}
\newcommand{\genericFnDomainElmt}{\ensuremath{d}\xspace}
\newcommand{\genericFnRange}{\ensuremath{\setNotation{R}}\xspace}

\newcommand{\rootIdx}{\ensuremath{j}\xspace}
\newcommand{\ptIdx}{\ensuremath{k}\xspace}
\newcommand{\setIdx}{\ensuremath{i}\xspace}

\NewDocumentCommand{\nats}{o o}%
{%
	\IfNoValueTF{#1}%
	{\ensuremath{\mathbb{N}}\xspace}%
	{\IfNoValueTF{#2}%
		{\ensuremath{[{#1}]}\xspace}%
		{\ensuremath{[{#1},{#2}]}\xspace}%
	}%
}

\NewDocumentCommand{\cInput}{o}{%
\IfNoValueTF{#1}%
{\ensuremath{\valNotation{w}}\xspace}%
{\ensuremath{\valNotation{w}_{#1}}\xspace}%
}

\NewDocumentCommand{\cInputAlt}{o}%
{%
  \IfNoValueTF{#1}%
              {\ensuremath{\valNotation{w}'}\xspace}%
              {\ensuremath{\valNotation{w}'_{#1}}\xspace}%
}

\newcommand{\cInputEmbedding}{\ensuremath{\vecNotation{z}}\xspace}

\newcommand{\authToken}{\ensuremath{\vecNotation{y}}\xspace}
\newcommand{\authTokenAlt}{\ensuremath{\vecNotation{y}'}\xspace}

\NewDocumentCommand{\PRPKey}{o}
{
\IfNoValueTF{#1}
{\ensuremath{\vecNotation{{k}_{\PRP}}}\xspace}
{\ensuremath{\vecNotation{{k#1}_{\PRP}}}\xspace}
}

\newcommand{\PRPKeySpace}{\ensuremath{\setNotation{K}_{\PRP}}\xspace}
\NewDocumentCommand{\PRFKeySpace}{o}{\ensuremath{\setNotation{K}_{\PRF[{#1}]}
}\xspace}
\NewDocumentCommand{\PRFKey}{o}{\ensuremath{\valNotation{k}_{\PRF[{#1}]}
}\xspace}
\NewDocumentCommand{\PRFKeyAlt}{o}{\ensuremath{\valNotation{k}'_{\PRF[{#1}]}
}\xspace}

\newcommand{\embedKeySpace}{\ensuremath{\setNotation{K}_{\embed}}\xspace}
\newcommand{\embedKey}{\ensuremath{\valNotation{k}_{\embed}}\xspace}
\newcommand{\embedKeyAlt}{\ensuremath{\valNotation{k}_{\embed}'}\xspace}

\newcommand{\degreeOfKey}[1]{
\ifstrequal{#1}{1}
{\ensuremath{\delta}\xspace}
{\ensuremath{2\delta}\xspace}
}

\NewDocumentCommand{\PRPKeyVec}{o o}
{
	\ifstrequal{#1}{1}
	{
		\IfNoValueTF{#2}
		{\ensuremath{\vecNotation{\kappa}_{1}}\xspace}
		{\ifblank{#2}{\ensuremath{\kappa_{1}}\xspace}{\ensuremath{\kappa_{1,#2}}\xspace} }
	}
	{
		\IfNoValueTF{#2}
		{\ensuremath{\vecNotation{\kappa}_{2}}\xspace}
		{\ifblank{#2}{\ensuremath{\kappa_2}\xspace}{\ensuremath{\kappa_{2,#2}}\xspace} }
	}
}

\NewDocumentCommand{\PRPKeyEmbeddingVec}{o o}
{
	\ifstrequal{#1}{1}
	{
		\IfNoValueTF{#2}
		{\ensuremath{\vecNotation{{\kappa}_{1}}}\xspace}
		{\ifblank{#2}{\ensuremath{\kappa_{1}}\xspace}{\ensuremath{\kappa_{1,#2}}\xspace} }
	}
	{
		\IfNoValueTF{#2}
		{\ensuremath{\vecNotation{{\kappa}_{2}}}\xspace}
		{\ifblank{#2}{\ensuremath{\kappa_2}\xspace}{\ensuremath{\kappa_{2,#2}}\xspace} }
	}
}

\NewDocumentCommand{\PRPKeyHashVec}{o o}
{
	\ifstrequal{#1}{1}
	{
		\IfNoValueTF{#2}
		{\ensuremath{\vecNotation{{\kappa}_{1}}'}\xspace}
		{\ifblank{#2}{\ensuremath{\kappa_{1}'}\xspace}{\ensuremath{\kappa_{1,#2}'}\xspace} }
	}
	{
		\IfNoValueTF{#2}
		{\ensuremath{\vecNotation{{\kappa}_{2}}'}\xspace}
		{\ifblank{#2}{\ensuremath{\kappa_2'}\xspace}{\ensuremath{\kappa_{2,#2}'}\xspace} }
	}
}

\newcommand{\bigO}[1]{\ensuremath{\mathsf{O}(#1)}\xspace}

\newcommand{\reals}{\ensuremath{\mathbb{R}}\xspace}
\newcommand{\nonnegativeReals}{\ensuremath{\reals_{\ge 0}}\xspace}

\newcommand{\setSize}[1]{\ensuremath{\left|{#1}\right|}\xspace}

\newcommand{\valNotation}[1]{\ensuremath{\expandafter\MakeLowercase{#1}}\xspace}
\newcommand{\setNotation}[1]{\ensuremath{\expandafter\mathcal{#1}}\xspace}
\newcommand{\rvNotation}[1]{\ensuremath{\expandafter\varmathbb{#1}}\xspace}
\newcommand{\fnNotation}[1]{\ensuremath{\expandafter\MakeLowercase{#1}}\xspace}
\newcommand{\algNotation}[1]{\ensuremath{\expandafter\mathtt{\MakeUppercase{#1}
}}\xspace}
\newcommand{\constNotation}[1]{\ensuremath{\expandafter\mathsf{#1}}\xspace}
\newcommand{\polyNotation}[1]{\ensuremath{\expandafter\MakeLowercase{#1}
}\xspace}
\newcommand{\vecNotation}[1]{\ensuremath{\expandafter\vec{#1}}\xspace}
\newcommand{\getsr}{\ensuremath{\stackrel{\scriptsize\$}{\leftarrow}}\xspace}
\newcommand{\defeq}{\ensuremath{\mathbin{\stackrel{\scriptsize\mbox{def}}{=}}
}\xspace}
\newcommand{\tru}{\ensuremath{\mathsf{TRUE}}\xspace}
\newcommand{\fals}{\ensuremath{\mathsf{FALSE}}\xspace}
\newcommand{\equals}{\stackrel{?}{=}}
\newcommand{\hadamardProd}{\ensuremath{\times}\xspace}

\newcommand{\bit}{\ensuremath{b}\xspace}
\newcommand{\bitAlt}{\ensuremath{\bit'}\xspace}
\newcommand{\bitRange}{\ensuremath{\{0, 1\}}\xspace}

\newcommand{\OLEAliceInput}{\ensuremath{x}\xspace}
\newcommand{\OLEBobInputOne}{\ensuremath{u}\xspace}
\newcommand{\OLEBobInputTwo}{\ensuremath{v}\xspace}

\NewDocumentCommand{\genericVecComp}{o o}
{
	\IfNoValueTF{#2}
	{\ensuremath{v_{#1}}\xspace}
	{\ensuremath{v_{#1, #2}}\xspace}
}

\NewDocumentCommand{\genericVecAltComp}{o o}
{
	\IfNoValueTF{#2}
	{\ensuremath{v'_{#1}}\xspace}
	{\ensuremath{v'_{#1, #2}}\xspace}
}

\newcommand{\linearComb}{\ensuremath{+}\xspace}

\newcommand{\prpOutputEmbedding}{\ensuremath{\vecNotation{p_1}}\xspace}
\newcommand{\prpOutputEmbeddingVec}{\ensuremath{\vecNotation{p_1}}\xspace}
\newcommand{\prpOutputHash}{\ensuremath{\vecNotation{p_2}}\xspace}
\newcommand{\prpOutputEmbeddingAlt}{\ensuremath{\vecNotation{p'_1}}\xspace}
\newcommand{\prpOutputHashAlt}{\ensuremath{\vecNotation{p'_2}}\xspace}
\newcommand{\prpOutputHashVec}{\ensuremath{\vecNotation{p_2}}\xspace}

\newcommand{\prfOutput}{\ensuremath{\vecNotation{\gamma}}\xspace}
\newcommand{\prfOutputAlt}{\ensuremath{\prfOutput'}\xspace}

\newcommand{\prpOutputEmbeddingComp}[1]{\ensuremath{\valNotation{p_{1,#1}}
}\xspace}

\newcommand{\prpOutputHashComp}[1]{\ensuremath{\valNotation{p_{2,#1}}}\xspace}

\newcommand{\cICInputOneComp}[1]{\ensuremath{\valNotation{z_{1,#1}}
}\xspace}
\newcommand{\cICInputTwoComp}[1]{\ensuremath{\valNotation{z_{2,#1}}}\xspace}

\newcommand{\sMaskComp}[1]{\ensuremath{\valNotation{\sMask_{#1}}}\xspace}

\newcommand{\hashStar}{\ensuremath{\valNotation{h^{*}}}\xspace}

\newcommand{\mapFn}{\ensuremath{\fnNotation{M}}\xspace}
\NewDocumentCommand{\sInputPW}{o}{\ensuremath{\lanElmt\IfNoValueF{#1}{_{#1}}
}\xspace}
\newcommand{\vectorNorm}[2]{\ensuremath{\left\Vert {#2}
\right\Vert_{#1}}\xspace}

\NewDocumentCommand{\pw}{o}{\ensuremath{\cInput\IfNoValueF{#1}{_{#1}}}\xspace}
\NewDocumentCommand{\pwAlt}{o}{\ensuremath{\cInputAlt\IfNoValueF{#1}{_{#1}}
}\xspace}

\newcommand{\blockedElmts}{\ensuremath{\universe_{\mathsf{blk}}}\xspace}

\newcommand{\policyLan}{\ensuremath{\setNotation{L}}\xspace}

\newcommand{\blockList}{\policyLan}
\newcommand{\blockListSize}{\ensuremath{n}}
\NewDocumentCommand{\lanElmt}{o}
{\IfNoValueTF{#1}
	{\ensuremath{\vecNotation{l}}\xspace}
	{\ensuremath{\vecNotation{{l}_{#1}}}\xspace}
}
\NewDocumentCommand{\lanElmtU}{o}
{\IfNoValueTF{#1}
	{\ensuremath{\valNotation{l}}\xspace}
	{\ensuremath{\valNotation{l}_{#1}}\xspace}
}

\NewDocumentCommand{\lanElmtComp}{g
g}{\ensuremath{\valNotation{l}_{#1\IfValueT{#2}{,#2}}}\xspace}

\NewDocumentCommand{\lElmtVecComp}{g
g}{\ensuremath{\vecNotation{{l}_{#1\IfValueT{#2}{,#2}}}}\xspace}

\NewDocumentCommand{\lanElmtVec}{g}{\ensuremath{\vecNotation{\lanElmt\IfValueF{
#1}{_{#1}}}\ }\xspace}

\newcommand{\threshold}{\ensuremath{T}\xspace}

\newcommand{\universe}{\ensuremath{\mathcal{U}}\xspace}

\NewDocumentCommand{\commitVec}{o}
{\IfNoValueTF{#1}
	{\ensuremath{\vecNotation{C}}\xspace}
	{\ensuremath{\vecNotation{c}_{#1}}\xspace}
}

\newcommand{\finiteField}{\ensuremath{\mathbb{F}}\xspace}
\newcommand{\fieldOrder}{\ensuremath{q}\xspace}
\newcommand{\polyField}{\ensuremath{\mathbb{F}[x]}\xspace}

\NewDocumentCommand{\genericSetVar}{o}%
{%
	\IfNoValueTF{#1}%
	{\ensuremath{\setNotation{S}}\xspace}
	{\ensuremath{\setNotation{S}_{#1}}\xspace}
}

\NewDocumentCommand{\genericVec}{o}%
{%
	\IfNoValueTF{#1}%
	{\ensuremath{\vecNotation{v}}\xspace}
	{\ensuremath{\vecNotation{{v}_{#1}}}\xspace}
}

\NewDocumentCommand{\genericVecAlt}{o}%
{%
	\IfNoValueTF{#1}%
	{\ensuremath{\vecNotation{v}'}\xspace}
	{\ensuremath{\vecNotation{v}_{#1}'}\xspace}
}

\newcommand{\genericSetElmt}{\ensuremath{s}\xspace}

\newcommand{\genericFieldElmt}{\ensuremath{x}}
\newcommand{\genericFieldElmtAlt}{\ensuremath{y}}

\newcommand{\numberOfPoints}{\ensuremath{2\delta + 1}\xspace}
\newcommand{\numberOfMasks}{\ensuremath{\degreeOfKey{1}}\xspace}
\newcommand{\distInMaskedSets}{\ensuremath{t}\xspace}
\newcommand{\distInUniverse}{\ensuremath{t}\xspace}
\newcommand{\tmpVar}{\ensuremath{t}\xspace}

\NewDocumentCommand{\tmpVarVec}{o o}
{
\IfNoValueTF{#2}
{\ensuremath{\vecNotation{t_{#1}}}\xspace}
{\ensuremath{t_{#1, #2}}\xspace}
}

\NewDocumentCommand{\tmpVarAltVec}{o o}
{
\IfNoValueTF{#2}
{\ensuremath{\vecNotation{t'_{#1}}}\xspace}
{\ensuremath{t'_{#1, #2}}\xspace}
}

\NewDocumentCommand{\bitMask}{o}
{
	\IfNoValueTF{#1}{\ensuremath{M}\xspace}{\ensuremath{M_{#1}}\xspace}
}

\newcommand{\vecMap}{\ensuremath{M_{\setOfEvalPts}}\xspace}

\newcommand{\diffInComps}{\ensuremath{d}\xspace}

\newcommand{\cMask}{\ensuremath{\vecNotation{m}}\xspace}
\newcommand{\cMaskAlt}{\ensuremath{\cMask'}\xspace}
\newcommand{\sMask}{\ensuremath{\vecNotation{s}}\xspace}
\newcommand{\sMaskAlt}{\ensuremath{\sMask'}\xspace}
\NewDocumentCommand{\cICInput}{o}%
{%
	\IfNoValueTF{#1}%
	{\ensuremath{\vecNotation{z}}\xspace}
	{\ensuremath{\vecNotation{z_{#1}}}\xspace}
}
\NewDocumentCommand{\cICInputAlt}{o}%
{%
	\IfNoValueTF{#1}%
	{\ensuremath{\valNotation{z'}}\xspace}
	{\ensuremath{\valNotation{z'_{#1}}}\xspace}
}

\newcommand{\nonce}{\ensuremath{nonce}\xspace}
\newcommand{\nonceAlt}{\ensuremath{\nonce'}\xspace}

\NewDocumentCommand{\client}{o}{%
  \IfNoValueTF{#1}%
    {\ensuremath{\mathsf{C}}\xspace}
    {%
      \IfEqCase{#1}{%
        {1}{\ensuremath{\mathsf{C}_{\mathsf{snd}}}\xspace}%
        {2}{\ensuremath{\mathsf{C}_{\mathsf{rcv}}}\xspace}%
      }[\ensuremath{\mathsf{C}_{#1}}\xspace]
    }%
}

\newcommand{\server}{\ensuremath{\mathsf{S}}\xspace}

\newcommand{\degreeOfPoly}{\ensuremath{\mathsf{Deg}}\xspace}
\newcommand{\uniqueRoots}[1]{\ensuremath{\setNotation{R}_{#1}}\xspace}
\newcommand{\polyFromVec}[1]{\ensuremath{\polyNotation{p_{#1}}(x)}\xspace}

\NewDocumentCommand{\polyFromEmbedding}{o o}
{
\IfNoValueTF{#2}
{\ensuremath{\polyNotation{e_{#1}}(x)} \xspace}
{\ensuremath{\polyNotation{e_{#1}}(#2)} \xspace}
}

\NewDocumentCommand{\clientsBlindingPolyVOLENoSub}{o}%
{%
	\IfNoValueTF{#1}%
	{\ensuremath{\polyNotation{R'}(x)}\xspace}
	{\ensuremath{\polyNotation{R'}(#1)}\xspace}
}

\NewDocumentCommand{\clientsBlindingPolyVOLE}{o o}%
{%
	\IfNoValueTF{#2}%
		{\IfNoValueTF{#1}%
			{\ensuremath{\polyNotation{R'}(x)}\xspace}
			{\ensuremath{\polyNotation{R'}_{#1}(x)}\xspace}
		}
	{
		{\ensuremath{\polyNotation{R'}_{#1}(#2)}\xspace}
	}
}

\NewDocumentCommand{\genericRand}{o}
{\IfNoValueTF{#1}
	{\ensuremath{r}\xspace}
	{\ensuremath{r_{#1}}\xspace}
}

\NewDocumentCommand{\PRPKeyPolyHash}{o o}
{
	\ifstrequal{#1}{1}
	{
		\IfNoValueTF{#2}
		{\ensuremath{\polyNotation{A'}(x)}\xspace}
		{\ensuremath{\polyNotation{A'}(#2)}\xspace}
	}
	{
		\IfNoValueTF{#2}
		{\ensuremath{\polyNotation{B'}(x)}\xspace}
		{\ensuremath{\polyNotation{B'}(#2)}\xspace}
	}
}

\NewDocumentCommand{\clientsRandomPoly}{o o}
{
	\IfNoValueTF{#2}
	{\ensuremath{\polyNotation{R}_{#1}(x)}\xspace}
	{\ensuremath{\polyNotation{R}_{#1}(#2)}\xspace}
}

\NewDocumentCommand{\simClientsRandomPoly}{o o}
{
	\IfNoValueTF{#2}
	{\ensuremath{\overline{\polyNotation{R}}_{#1}(x)}\xspace}
	{\ensuremath{\overline{\polyNotation{R}}_{#1}(#2)}\xspace}
}

\NewDocumentCommand{\simServersRandomPoly}{o o}
{
	\IfNoValueTF{#2}
	{\ensuremath{\overline{\polyNotation{R}}_{#1}'(x)}\xspace}
	{\ensuremath{\overline{\polyNotation{R}}_{#1}(#2)}\xspace}
}

\NewDocumentCommand{\serversRandomPoly}{o o}
{
	\IfNoValueTF{#2}
	{\ensuremath{\polyNotation{R}_{#1}'(x)}\xspace}
	{\ensuremath{\polyNotation{R}_{#1}'(#2)}\xspace}
}

\NewDocumentCommand{\interpolatingPoly}{o}
{
	\IfNoValueTF{#1}
	{\ensuremath{\polyNotation{G}(x)}\xspace}
	{\ensuremath{\polyNotation{G}(#1)}\xspace}
}

\NewDocumentCommand{\evalPt}{o}
{
	\IfNoValueTF{#1}
	{\ensuremath{\genericFieldElmt}\xspace}
	{\ensuremath{\genericFieldElmt_{#1}}\xspace}
}

\NewDocumentCommand{\genericPoly}{o}%
{%
	\IfNoValueTF{#1}%
	{\ensuremath{\polyNotation{p}(x)}\xspace}
	{\ensuremath{\polyNotation{p}_{#1}(x)}\xspace}
}

\NewDocumentCommand{\genericPolyEval}{o o}%
{%
	\IfNoValueTF{#2}%
	{\ensuremath{\polyNotation{p}(#1)}\xspace}
	{\ensuremath{\polyNotation{p}_{#2}(#1)}\xspace}
}

\NewDocumentCommand{\genericPolyAlt}{o}%
{%
	\IfNoValueTF{#1}%
	{\ensuremath{\polyNotation{p'}(x)}\xspace}
	{\ensuremath{\polyNotation{p'}_{#1}(x)}\xspace}
}

\NewDocumentCommand{\genericPolyAltEval}{o o}%
{%
	\IfNoValueTF{#2}%
	{\ensuremath{\polyNotation{p'}(#1)}\xspace}
	{\ensuremath{\polyNotation{p'}_{#2}(#1)}\xspace}
}

\newcommand{\hashComp}[1]{\ensuremath{h_{#1}}\xspace}

\NewDocumentCommand{\randomPoly}{o}%
{%
	\IfNoValueTF{#1}%
	{\ensuremath{\polyNotation{r}(x)}\xspace}
	{\ensuremath{\polyNotation{r}_{#1}(x)}\xspace}
}

\NewDocumentCommand{\randomPolyComb}{o}%
{%
	\IfNoValueTF{#1}%
	{\ensuremath{\polyNotation{r_c}(x)}\xspace}
	{\ensuremath{\polyNotation{r_c}(#1)}\xspace}
}

\NewDocumentCommand{\randomPolyEval}{o o}%
{%
	\IfNoValueTF{#2}%
	{\ensuremath{\polyNotation{r}(#1)}\xspace}
	{\ensuremath{\polyNotation{r}_{#2}(#1)}\xspace}
}

\NewDocumentCommand{\randomVec}{o}{\ensuremath{\vecNotation{r\IfValueT{#1}{_{#1
}}}}\xspace}

\newcommand{\serversRandomPolyEvalVecDefn}[1]{\ensuremath
{\serversRandomPolyEvalVec[#1] \gets \langle
\serversRandomPoly[\setIdx]\rangle_{\evalPt \in \setOfEvalPts}}\xspace}
\newcommand{\clientsRandomPolyEvalVecDefn}[1]{\ensuremath
{\clientsRandomPolyEvalVec[#1] \gets \langle
\clientsRandomPoly[\setIdx]\rangle_{\evalPt \in \setOfEvalPts}}\xspace}
\newcommand{\randomPolyCombEvalVec}{\vecNotation{r_{c}}\xspace}

\NewDocumentCommand{\clientsESPPolyEvalVec}{o}
{\IfNoValueTF{#1}
	{\ensuremath{\vecNotation{W}}\xspace}
	{\ensuremath{w_{#1}}\xspace}
}

\NewDocumentCommand{\clientsSimPolyEvalVec}{o}
{\IfNoValueTF{#1}
	{\ensuremath{\vecNotation{W}}\xspace}
	{\ensuremath{w_{#1}}\xspace}
}

\NewDocumentCommand{\clientsInputPolyEvalVec}{o}
{\IfNoValueTF{#1}
	{\ensuremath{\vecNotation{p}}\xspace}
	{\ensuremath{p_{#1}}\xspace}
}

\NewDocumentCommand{\clientsRandomPolyEvalVec}{o o}
{\IfNoValueTF{#2}
	{\ensuremath{\vecNotation{r}_{#1}}\xspace}
	{\ensuremath{r_{#1, #2}}\xspace}
}

\NewDocumentCommand{\serversRandomPolyEvalVec}{o o}
{\IfNoValueTF{#2}
	{\ensuremath{\vecNotation{r}_{#1}'}\xspace}
	{\ensuremath{r_{#1, #2}'}\xspace}
}

\NewDocumentCommand{\serversRandomPolyEval}{o o}
{\IfNoValueTF{#2}
	{\ensuremath{\vecNotation{r_{#1}}}\xspace}
	{\ensuremath{\vecNotation{r_{#1, #2}}}\xspace}
}
\NewDocumentCommand{\serversRandomPolyEvalAlt}{o o}
{\IfNoValueTF{#2}
	{\ensuremath{\vecNotation{r_{#1}}'}\xspace}
	{\ensuremath{\vecNotation{r_{#1, #2}}'}\xspace}
}

\NewDocumentCommand{\serversInputPolyEvalVec}{o o}
{\IfNoValueTF{#2}
	{\ensuremath{\vecNotation{Q}_{#1}}\xspace}
	{\ensuremath{q_{#1, #2}}\xspace}
}

\NewDocumentCommand{\simRandomPoly}{o}%
{%
	\IfNoValueTF{#1}%
	{\ensuremath{\overline{\polyNotation{r}}(x)}\xspace}
	{\ensuremath{\overline{\polyNotation{r}}_{#1}(x)}\xspace}
}

\NewDocumentCommand{\simPolyFromEmbedding}{o o}
{
\IfNoValueTF{#2}
{\ensuremath{\overline{\polyNotation{e}}_{#1}}(x) \xspace}
{\ensuremath{\overline{\polyNotation{e}}_{#1}}(#2) \xspace}
}

\NewDocumentCommand{\simRandomPolyEval}{o o}
{
\IfNoValueTF{#2}{\ensuremath{\overline{\polyNotation{r}}(#1)}\xspace}
{\ensuremath{\overline{\polyNotation{r}}_{#2}(#1)}\xspace}
}

\NewDocumentCommand{\simRandomPolyAlt}{o}
{
	\IfNoValueTF{#1}
	{\ensuremath{\overline{\polyNotation{r'}}(x)}\xspace}
	{\ensuremath{\overline{\polyNotation{r'}}_{#1}(x)}\xspace}
}

\NewDocumentCommand{\simRandomPolyAltEval}{o o}
{
\IfNoValueTF{#2}{\ensuremath{\overline{\polyNotation{r'}}(#1)}\xspace}
{\ensuremath{\overline{\polyNotation{r'}}_{#2}(#1)}\xspace}
}

\newcommand{\combVector}{\ensuremath{\vecNotation{c}}\xspace}
\newcommand{\combPoly}{\ensuremath{\polyNotation{c}(x)}\xspace}

\newcommand{\prob}[1]{\ensuremath{\mathbb{P}\mathopen{}\left({#1}\right
)\mathclose{}}\xspace}
\newcommand{\cprob}[3]{\ensuremath{\mathbb{P}\mathopen{}{#1(}{#2} \; {#1|} \;
{#3}{#1)}\mathclose{}}\xspace}

\newcommand{\falseRejectRate}[3]{\ensuremath{\mathsf{FRR}^{#1}_{{#2},{#3}}
}\xspace}
\newcommand{\falseAcceptRate}[3]{\ensuremath{\mathsf{FAR}^{#1}_{{#2},{#3}}
}\xspace}
\newcommand{\trueRejectRate}[3]{\ensuremath{\mathsf{TRR}^{#1}_{{#2},{#3}}
}\xspace}
\newcommand{\trueAcceptRate}[3]{\ensuremath{\mathsf{TAR}^{#1}_{{#2},{#3}}
}\xspace}

\newcommand{\Adv}[1]{\ensuremath{\algNotation{A}_{#1}}\xspace}
\newcommand{\AdvS}[1]{\ensuremath{\algNotation{S}_{#1}}\xspace}
\newcommand{\genericAdv}{\ensuremath{\mathcal{A}}\xspace}

\NewDocumentCommand{\metricSpace}{o}%
{%
	\IfNoValueTF{#1}%
	{\ensuremath{\finiteField^{\numberOfPointsSim}}\xspace}%
	{\ensuremath{\setNotation{M}_{#1}}\xspace}
}

\newcommand{\PRPsecdef}{\textsf{\scriptsize opu}\xspace}

\newcommand{\WPRsecdef}{\textsf{\scriptsize wpr}\xspace}
\newcommand{\CRsecdef}{\textsf{\scriptsize cr}\xspace}
\newcommand{\PRFINDsecdef}{\textsf{\scriptsize wpr}\xspace}
\newcommand{\Ssecdef}{\textsf{\scriptsize S}\xspace}
\newcommand{\DCOPRFsecdef}{\scriptsize\textsf{\exptAcro}\xspace}

\NewDocumentCommand{\ExptCC}{o}%
{%
  \IfNoValueTF{#1}%
{\ensuremath{\mathbf{Expt}^{\mbox{\scriptsize\textsf{\exptAcro}}}}\xspace}%
{\ensuremath{\mathbf{Expt}^{\mbox{\scriptsize\textsf{\exptAcro}}}_{#1}}\xspace}%
}

\NewDocumentCommand{\ExptS}{o}%
{%
  \IfNoValueTF{#1}%
{\ensuremath{\mathbf{Expt}^{\mbox{\scriptsize\textsf{\algNotation{S}}}}}\xspace
}%
{\ensuremath{\mathbf{Expt}^{\mbox{\scriptsize\textsf{\algNotation{S}}}}_{#1}
}\xspace}%
}

\NewDocumentCommand{\blocked}{g g g}%
{%
	\IfNoValueTF{#2}
{\ensuremath{\fnNotation{\mathit{blocked}}^{{#1}}}\xspace}
{
\IfNoValueTF{#3}%
	{\ensuremath{\fnNotation{\mathit{blocked}}^{{#1},{#2}}}\xspace}
	{\ensuremath{\fnNotation{\mathit{blocked}}_{#1}^{{#2},{#3}}}\xspace}
}
}

\newcommand{\ExptWPR}[1]{\ensuremath{\mathbf{Expt}^{\mbox{\textsf{\WPRsecdef}}}
_{#1}}\xspace}
\newcommand{\ExptCR}[1]{\ensuremath{\mathbf{Expt}^{\mbox{\textsf{\CRsecdef}}}
_{#1}}\xspace}
\NewDocumentCommand{\ExptPRP}{o}%
{%
  \IfNoValueTF{#1}%
{\ensuremath{\mathbf{Expt}^{\mbox{\textsf{\PRPsecdef}}}_{\PRP}}\xspace}
{\ensuremath{\mathbf{Expt}^{\mbox{\textsf{\PRPsecdef}-{#1}}}_{\PRP}}\xspace}
}
\NewDocumentCommand{\ExptPRFIND}{o}%
{%
  \IfNoValueTF{#1}%
{\ensuremath{\mathbf{Expt}^{\mbox{\textsf{\PRFINDsecdef}}}_{\PRF[1]}}\xspace}
{\ensuremath{\mathbf{Expt}^{\mbox{\textsf{\PRFINDsecdef}-{#1}}}_{\PRF[1]}}\xspace}
}

\newcommand{\Advstate}[1]{\ensuremath{\phi_{#1}}\xspace}

\newcommand{\Advantage}[3]{\ensuremath{\mathsf{Adv}_{#2}^{\mbox{\textsf{#1}}
}\mathopen{}\left({#3}\right)\mathclose{}}\xspace}
\newcommand{\exptAcro}{CC}
\newcommand{\timeBound}{\ensuremath{t}\xspace}
\newcommand{\timeBoundAlt}{\ensuremath{t'}\xspace}

\newcommand{\forwardOracleQueries}{\ensuremath{q}\xspace}

\newcommand{\hashOracleQueries}{\ensuremath{\forwardOracleQueries_{\mathsf{H}}
}\xspace}

\newcommand{\simOp}[1]{\ensuremath{\overline{#1}}\xspace}

\newcommand{\prfOutputIC}{\ensuremath{\hat{\prfOutput}}\xspace}

\newcommand{\primitiveName}{\textit{Half-Moon Cookie}\xspace}
\newcommand{\PrimitiveName}{\textit{Half-Moon Cookie}\xspace}
\newcommand{\primitiveNameShort}{\textit{Half Moon}\xspace}
\newcommand{\PrimitiveNameShort}{\textit{Half Moon}\xspace}
\newcommand{\implicitCheckPhase}{implicit check\xspace}
\newcommand{\explicitCheckPhase}{explicit check\xspace}

\newcommand{\ImplicitCheckPhase}{Implicit check\xspace}
\newcommand{\ExplicitCheckPhase}{Explicit check\xspace}
\newcommand{\ImplicitCheckPhaseHeading}{Implicit Check\xspace}
\newcommand{\ExplicitCheckPhaseHeading}{Explicit Check\xspace}

\newcommand{\serverState}{\ensuremath{\phi}\xspace}

\newcommand{\genericProtocol}{\ensuremath{\Pi}\xspace}
\newcommand{\protocolEmbedESP}{\ensuremath{\genericProtocol_
{\mathit{EM}}}\xspace}
\newcommand{\protocolCommitESP}{\ensuremath{\genericProtocol_
{\mathit{TC}}}\xspace}
\newcommand{\protocolValidateESP}{\ensuremath{\genericProtocol_
{\mathit{IC}}}\xspace}

\newcommand{\protocolEmbed}{\protocolEmbedESP}
\newcommand{\protocolCommit}{\protocolCommitESP}
\newcommand{\protocolValidate}{\protocolValidateESP}

\NewDocumentCommand{\protocolOPRF}{o}
{
	\IfNoValueTF{#1}
	{\ensuremath{\genericProtocol_{\mathit{OPRF}}}\xspace}
	{\ensuremath{\genericProtocol_{\mathit{OPRF}#1}}\xspace}
}

\newcommand{\embedAndMap}{\ensuremath{\mathcal{F}_{\mathsf{EM}}}\xspace}
\newcommand{\testAndCommit}{\ensuremath{\mathcal{F}_{\mathsf{TC}}}\xspace}
\newcommand{\authentication}{\ensuremath{\mathcal{F}_{\mathsf{IC}}}\xspace}

\newcommand{\genericMetricSpaceDist}{\ensuremath{D}\xspace}

\NewDocumentCommand{\genericPRPKey}{g}{\ensuremath{\langle
\PRPKeyVec[1]\IfNoValueF{#1}{_{#1}},
\PRPKeyVec[2]\IfNoValueF{#1}{_{#1}}\rangle}\xspace}
	
\ExplSyntaxOn
\cs_new:Npn \intensity #1
   {\fp_eval:n {\MinIntensity + (1.0-\MinIntensity) * (((\MaxNumber-{#1})/(\MaxNumber-\MinNumber)))}
   }
\ExplSyntaxOff
\FPeval{\MinIntensity}{0.5}   
\FPeval{\MinNumber}{0.0}
\FPeval{\MaxNumber}{0.0}
\newcommand{\ApplyGradientX}[1]{\cellcolor[gray]{\intensity{#1}}{\makebox[\width][r]{#1}}}
\newcolumntype{X}{>{\collectcell\ApplyGradientX}r<{\endcollectcell}}

\begin{document}
\title{\PrimitiveName: Private, Similarity-Based Blocklisting\\ with
  TOCTOU-Attack Resilience}
\author{
Xinyuan Zhang \\ 
Duke University
\and
Anrin Chakraborti \\ 
University of Illinois at Chicago
\and
Michael K. Reiter \\ 
Duke University}
\maketitle

\iffalse
\begin{abstract}
Blocklisting is a widely deployed technique for preventing the 
distribution and execution of known malicious content. However, 
conventional blocklisting infrastructures either require clients to reveal 
their queries to the server or require the blocklist itself to be public.
In this work, we introduce a private blocklisting framework, \primitiveName, 
that leverages both a blacklist and a whitelist maintained by the 
policy-checking server. Our design embeds each client input into a metric 
space and adds a commitment for this input to a whitelist only if its 
embedding does not fall within a cluster of known malicious elements in 
the blacklist. The structural separation between the embedding step and 
the blocklist membership check enables each to be implemented efficiently 
using specialized techniques. We then cryptographically “stitch” these 
phases together to preserve soundness and privacy. The resulting whitelist 
commitments are both binding and hiding, allowing the server to support 
fast re-evaluation of previously seen inputs while simultaneously providing 
resilience against time-of-check-to-time-of-use (TOCTOU) attacks. We 
demonstrate how \primitiveName can be instantiated for similarity-based 
malware detection, enabling effective identification of malicious binaries 
without revealing client inputs or disclosing the underlying blocklist.
\end{abstract}
\else
\begin{abstract}
Blocklisting is a common technique for preventing the use of known
malicious content. However, conventional blocklisting infrastructures
require either the blocklist to be public or clients to reveal their
queries to the blocklist server.  We introduce a private blocklisting
framework, \primitiveName, by which a client can check an item against
a proprietary blocklist held by a server, to determine whether the
item is close to any blocklist element in a metric space.  Critically,
our design separates the embedding step from the blocklist check,
enabling independent choice of methods to compute them privately and
efficiently.
Still, this computation might be too costly to perform on the
critical path of using the item, and so our design also supports a
very efficient check that an item previously passed the blocklist
check.  In doing so, we support applications where one client can
perform the blocklist check on the item before sending it, and
recipients can more efficiently confirm the previous result before
using the item, thereby avoiding TOCTOU attacks.  We show how
\primitiveName can be instantiated for similarity-based malware
detection, to block malicious executables without revealing client
inputs or disclosing the blocklist.
\end{abstract}
\fi



\section{Introduction}
\label{sec:intro}

Blocklisting is a core primitive of computer security that is relied
upon daily to interfere with proliferation or use of 
malware~\cite{URLHaus, MalwareBazaar}, unsolicited or fraudulent 
email~\cite{DNSBW}, revoked certificates~\cite{rfc5280},
malicious domains~\cite{SpamhausDBL}, malicious IP 
subnets~\cite{TeamCymru, SpamhausBL, SpamhausDROP}, child 
sexual abuse material 
(CSAM)~\cite{Microsoft:PhotoDNA, NCMEC:Google, NCMEC:report, IWF:CSAM}, 
and other malicious artifacts~\cite{Easylist}.
When the blocklist and/or the artifacts checked against it are private
or proprietary, verifying that an artifact is not on a blocklist might
require using a cryptographic protocol to hide the blocklist from the
client with the artifact and the artifact from the blocklist server.
Performing such a blocklist check immediately prior to artifact use
can be impractical in some cases, however, particularly if the
blocklist check requires a rich cryptographic functionality (e.g., an
\textit{approximate} set 
intersection~\cite{Garimella2022:structurePSI, Chakraborti2023:DAPSI, 
Garimella2024:structurePSI, Cho2024:approxPSI, Blass25:fuzzyPSI} 
using the client's artifact and the server's blocklist) that incurs 
nontrivial cost. Moving this blocklist check to an earlier, more 
manageable time can lead to time-of-check vs.\ time-of-use (TOCTOU)
vulnerabilities~\cite{cwe367}, if not paired with a confirmation that a
previous check is still valid right before the artifact's use.

In this paper we propose a framework by which a client can perform a
blocklist check---even an approximate one---at a server, while
ensuring privacy of the artifact from the server and privacy of the
blocklist from the client.  If this \textit{\explicitCheckPhase}
passes, then the server adds a hiding and binding token for the
artifact to an allowlist, which this or another client can later
confirm has not been revoked, via an \textit{\implicitCheckPhase} with
the server.  Like the \explicitCheckPhase, the \implicitCheckPhase
preserves the privacy of the artifact (and the blocklist).  However,
it is much faster and so more practical to perform immediately before
the artifact's use, rendering TOCTOU attacks less of a threat.  Due to
our design that allowlists an artifact after confirming it is not
similar to anything on a blocklist, we refer to our design as
``\primitiveName,'' or ``\primitiveNameShort,'' for short.\footnote{A
half-moon, or black-and-white, cookie is a cookie with vanilla
frosting on one half and chocolate frosting on the
other~\cite{bwcookie}.}

We illustrate the utility of \primitiveNameShort with a concrete
application to blocklisting of executable content (e.g., binaries,
scripts, or bytecode).  When receiving executable content, it is
prudent for recipients to confirm that the content has been checked
against a blocklist.  To this end, a variety of content hashing
techniques have been developed (e.g.,~\cite{Kornblum06:ssdeep,
  Roussev10:sdhash, Oliver13:TLSH}), so that executable content that
hashes to values close to blocklist elements are flagged as possibly
malicious.  However, checking proprietary executable content against a
proprietary blocklist mandates that this check be performed in a
privacy-preserving way---i.e., to hide the content from the blocklist
server and the blocklist from the client---which in turn requires a
rich cryptographic functionality to perform this check.  In the common
case that the content sender is outnumbered by the content receivers
and can perform this blocklist check off the critical path of
distributing the content, it is far more efficient for the sender to
perform this check.  \PrimitiveNameShort enables the sender to perform
this check once, using an \explicitCheckPhase, and then for each
receiver to more efficiently perform an \implicitCheckPhase, just
before consuming the content.

The approximate blocklist check required in the application above, and
that we support in \primitiveNameShort, determines whether the hash of
the sender's artifact (in the case above, executable content) lies
within a prespecified distance to the hash of some blocklist element,
i.e., in a metric space defined over the hash function range.  A core
challenge that we need to solve, then, is to guarantee that the sender
embeds (i.e., hashes) its artifact correctly, since without such
enforcement, a malicious sender could simply check a hash value
different from the hash of its artifact.  The most direct approach to
do so would be to perform the embedding and check together in a single
immutable computation, for which a monolithic garbled circuit is a
natural candidate (e.g.,~\cite{Katz2016:Input2PC}).  As we will show
in \secref{sec:eval}, however, this approach does not scale well in
the number of blocklist elements (or their clusters).  The
\explicitCheckPhase of \primitiveNameShort thus relies on a garbled
circuit only for the embedding phase and improves the circuit cost by
making it independent of the size of the blocklist.  Finally, the
\implicitCheckPhase does not leverage a garbled circuit at all, since
the receiver is trusted to embed the received artifact correctly,
presuming it desires to know whether this artifact was, in fact,
previously checked against the blocklist.

Even with making the garbled circuit independent of the size of the
blocklist, embedding an executable by hashing it within a garbled
circuit is itself costly, since executables tend to be large.  We
optimize the circuit cost further by decoupling it from the input file
size, as well, by leveraging techniques from \citet{CRGC}, while
maintaining full security. In \secref{sec:app:embed-and-map:reusable}, 
we present how to
split the embedding phase into reusable and non-reusable garbled
circuits to save the computation and communication cost, without
leaking the garbler's input. Without such techniques, the latency of a
monolithic garbled circuit is $231\times$ the cost of
\primitiveNameShort's embedding, while the network traffic is
$43{,}505\times$ that of \primitiveNameShort for an average size
email attachment (\SI{194}{\kilo\byte}).

To summarize, our contributions are as follows.
\begin{itemize}[nosep,leftmargin=1em,labelwidth=*,align=left]
\item We formally define the \primitiveNameShort
  primitive, and provide a general framework to design
  \primitiveNameShort protocols that support an
  \textit{\explicitCheckPhase} and a more efficient
  \textit{\implicitCheckPhase} that confirms that the provided
  artifact previously passed an \explicitCheckPhase.

\item We provide efficient \primitiveNameShort implementations for
  blocklists that can be embedded into a metric space based on Hamming
  distance.

\item We demonstrate how \primitiveNameShort can be applied to
  similarity-based malware defense. We show details for how to
  instantiate the application and empirically analyze the system
  performance.
\end{itemize}

\section{Related Work}
\label{sec:related}

Below we discuss protocols with similar design goals and tools that
are relevant to our task.

\myparagraph{Outsourced malware defense} 
A driving application for \primitiveNameShort is to enable a sender of
an executable to perform an \explicitCheckPhase with a policy server
so that recipients can more efficiently confirm the result via an
\implicitCheckPhase, enabling recipients to avoid TOCTOU attacks. Our
framework is best suited to malware detection by fuzzy
hashing~\cite{Roussev10:sdhash, Li15:FHReview} and provides the
content privacy necessitated by outsourcing this scanning to a
third-party service with a proprietary blocklist.
CloudAv~\cite{Oberheide2008:CloudAV} was the first work to outsource
malware scanning to a cloud that runs the anti-virus engines.
SplitScreen~\cite{Cha2010:SplitScreen} and RScam~\cite{Sun2015:RScam}
enhance client and malware vendor privacy by sending only compact
representations of the files/malware signatures over the network, but
provide no cryptographic protections to the data.  Perhaps closest to
our envisioned design, \citet{Sun17:PriMal} proposed a
privacy-preserving cloud-based malware detection scheme.  However, it
works by exactly matching client input files with malware signatures,
necessitating a very large blocklist---and resulting high protocol
cost---to achieve effective coverage.

More distantly related to our work are solutions to perform middlebox
functionalities over encrypted traffic and encrypted rules
(e.g.,~\cite{Sherry2015:BlindBox, Lan2016:Embark,
  Melis2016:Outsourced, Fan17:spabox}) and privacy-preserving spam
filtering that combines machine learning classifiers and MPC,
functional encryption, or homomorphic encryption~\cite{Pathak11:spam,
  Wang20:spam, Lee23:spam} to privately filter communication.  None of
these designs, however, offer a faster, \implicitCheckPhase to confirm
that an artifact previously passed a slower, \explicitCheckPhase.


\myparagraph{Private blocklist matching}
A line of research has focused on enabling a client to query whether
an item appears on a blocklist while preserving the confidentiality of
the client's input. This paradigm has been adapted to applications
such as malware detection~\cite{FHEmalware}, spam
filtering~\cite{Kim23:PCSF}, and credential breach
checks~\cite{Li19:C3}. \citet{Kogan21:Checklist} provide a generic
private blocklist lookup construction with two-server PIR. Private
Hash Matching~\cite{Kulshrestha2021} is a service designed for
end-to-end encrypted messaging, where a server matches a client's
media content against a set of images (e.g., child sexual abuse
material) while hiding the server's hash set or nonmatching content
from the client.  But none of the above schemes generate a
hiding and binding token after the match to allow for faster
verifications later on the same content (as in our
\implicitCheckPhase).

Some OPRF variations bear similarity to our \primitiveNameShort
primitive, as well.  Programmable OPRFs
(OPPRF)~\cite{Chandran2022:OPPRF} and oblivious key-value stores
(OKVS)~\cite{Garimella2021:OKVS} both ``program'' an OPRF so that the
server gets a PRF output for an input $\cInput$ iff $\cInput$ matches a
server list item, while the output is pseudorandom elsewhere.  All
solutions above require exact matching against the entire blocklist,
which might be too costly for large blocklists.  \citet{approxSchema}
consider approximate matching over database records but rely on a
trusted third party to facilitate the check.

\myparagraph{Enforcing policies on protocol inputs}
Several other protocols address the general problem of ensuring that
protocol inputs satisfy some predicate. Some
(e.g.,~\cite{Camenisch2009:CertSet, Blanton2015:genom,
  Zhang2017:Input}) outsource the predicate check to a trusted third
party, but in our applications, finding a trusted third party is
problematic. Others (e.g.,~\cite{Chang2021:stat, Prio, Zombie, Reef,
  Luo2023:rexMatching}) prove in zero knowledge certain properties of
an input. However, the proof complexity increases as a function of the
property that needs to be checked. Besides, zero-knowledge
policy-enforcement designs work only when the matching pattern and
policy circuit are public, while \primitiveNameShort aims at keeping
the server’s blocklist private.

Secure computation tools, e.g., garbled circuits, provide a viable
alternative for enforcing policies on protocol inputs. For example,
\citet{Katz2016:Input2PC} proposed a scheme where a garbled circuit is
preceded by a predicate-checking circuit to ensure that the input
follows certain properties.  A similar construction is proposed by
\citet{Baum2016:GC}. However, a garbled circuit may not be the most
efficient primitive for performing a certain task privately
\cite{Garg2018:pkGC}, especially when dealing with sets of items.
\citet{Agrawal2021:covert} proposed an MPC scheme by which a party can
commit to a value and use this value subsequently in a secure
computation, but it supports only covert security, where cheating can
be detected rather than fully prevented. \primitiveNameShort provides
stronger security and privacy guarantees.

\myparagraph{PSI-based mechanisms}
A part of our design relies on approximately matching an item against
a set privately (e.g., using fuzzy PSI mechanisms). Structure-aware
PSIs~\cite{Garimella2022:structurePSI, Garimella2024:structurePSI}
rely on spatial hashing, which guarantees reasonable compute costs
only under certain assumptions about the data in one (or both) of the
input sets. Approximate PSI~\cite{Cho2024:approxPSI} is a similarly
motivated concept that also works only under certain distribution
assumptions. As it is unclear if these distribution assumptions hold
across a variety of applications, we cannot generally use a
structure-aware PSI or an approximate PSI as a building block for a
\primitiveNameShort. 

textit{Distance-aware} PSI~\cite{Chakraborti2023:DAPSI} and
\textit{Fuzzy} PSI~\cite{Blass25:fuzzyPSI} perform
fuzzy matching over two sets \textit{without any data distribution
  assumptions}. While this incurs a higher cost than both
structure-aware PSI and approximate PSI, distance-aware PSI and fuzzy PSI
can be applied more broadly because it is agnostic to the data
distribution. We will use Chakraborti, et al.~\cite{Chakraborti2023:DAPSI} as a building block, 
as we detail in \appref{app:background}.
Unlike Chakraborti et al.,~\cite{Chakraborti2023:DAPSI} 
Blass and Noubir~\cite{Blass25:fuzzyPSI} operates directly on binary vector inputs
and is therefore not directly compatible with \primitiveNameShort.
However, all above solutions
either assume a semi-honest client or require non-trivial
modifications to ensure a consistent client input between the fuzzy
PSI matching and the computation of a hiding and binding token to
store at \server, to enable a subsequent \implicitCheckPhase.
\primitiveNameShort does not trust the client during an
\explicitCheckPhase and, in particular, detects when the client
submits inconsistent inputs from embedding to subsequent checking. 

Beyond general application-agnostic designs,
\citet{Wang2015:approx_editDist} realize an approximate edit-distance
check on genome data. Similar to a \primitiveNameShort, they utilize
an embedding phase and a test phase to compare entries at the client
and the server. Also related is similarity-aware compromised
credential checking (C3)~\cite{Pal2022:C3}.  These solutions work in
semi-honest models only, however. In addition, \primitiveNameShort
removes the assumptions on the structure of input data by
\citet{Wang2015:approx_editDist} that are tailored to genome data and
may not be applicable to other settings.

In \tblref{tab:comp_protocols}, we summarize the costs and key
guarantees of \primitiveNameShort and several cryptographic
alternatives.

\begin{table*}[t]
\centering
{\small
\begin{tabular}{@{}l@{\hspace{1em}}c@{\hspace{1em}}c@{\hspace{1em}}c@{\hspace{1em}}c@{}}
\toprule
\multicolumn{1}{c}{\multirow{2}{*}{Protocol}} 
& \multirow{2}{*}{Computation} 
& \multirow{2}{*}{Communication} 
& \multirow{2}{*}{Client}
& \multirow{2}{*}{\shortstack{TOCTOU \\ resilient}} \\
& & & \\
\midrule
\primitiveNameShort \\
\hspace{1em}\explicitCheckPhase & $\GCcost{\embed}[\inputSize]+\FPSIcost[\setSize{\policyLan}]$ & $\GCcost{\embed}[\inputSize]+\FPSIcost[\setSize{\policyLan}]$ & Mal & \yesIndicator \\
\hspace{1em}\implicitCheckPhase & $\embedcost[\inputSize]+\OLEcost$ & $\OLEcost$ & SH & \yesIndicator \\
PSI & $\PSIcost[\setSize{\blockedElmts}]$ & $\PSIcost[\setSize{\blockedElmts}]$ & Mal & \yesIndicator \\
Fuzzy PSI & $\embedcost[\inputSize]+\FPSIcost[\setSize{\policyLan}]$ & $\FPSIcost[\setSize{\policyLan}]$ & SH & \yesIndicator \\
Garbled Circuit & $\GCcost{\embed}[\inputSize] + \GCcost{\blocked{}}[\setSize{\policyLan}]$ & $\GCcost{\embed}[\inputSize] + \GCcost{\blocked{}}[\setSize{\policyLan}]$ & Mal & \yesIndicator \\
Blind Signature & $\PSIcost[\setSize{\blockedElmts}]$ & $\PSIcost[\setSize{\blockedElmts}]$ & Mal & \noIndicator \\
\bottomrule
\end{tabular}
}
\caption{Protocol comparison. $\GCcost{\embed}$ and
    $\GCcost{\blocked{}}$ denote the costs of executing a garbled
    circuit computing \embed or $\blocked{}$, respectively.  \PSIcost
    and \FPSIcost denote the costs of plain PSI and fuzzy PSI,
    respectively. \OLEcost denotes the cost of performing an OLE, and
    \embedcost denotes the cost of running the embedding algorithm
    \embed outside the garbled circuit. \inputSize denotes the size of
    the client input \cInput.  ``Mal'' and ``SH'' indicate security
    against malicious and only semi-honest clients, respectively.
\label{tab:comp_protocols}}
\vspace{-5pt}
\end{table*}

\section{Use-Case Scenario: TOCTOU Attacks in Email Attachments and URLs}
\label{sec:usecase}

\subsection{Background}
The class of vulnerabilities known as time-of-check vs.\ time-of-use
(TOCTOU) describes inconsistencies that occur when a resource is
validated as safe for use but, when actually used later, would have
been discovered as unsafe if checked at that time~\cite{cwe367}.  The
inconsistency between the earlier check and the later use can occur
because the used item was changed since it was checked or the item was
not changed but was discovered to be malicious in the interim.  TOCTOU
vulnerabilities abound, and indeed the Amazon AWS outage in October
2025 has been at least partly attributed to
one~\cite{Belotti2025:AWS}.

In the scenario we consider here, we are concerned with policy
enforcement on content, such as an email attachment, received from an
untrusted source.  We are agnostic to whether the checked content is
sent by value (e.g., itself attached to an email) or by reference
(e.g., retrieved using a URL in the email).  Empirical analyses from
major security vendors have confirmed that TOCTOU-style email threats
are widespread (e.g.,~\cite{microsoft:filehost,
  Deepwatch:ransomware}).  In a representative TOCTOU exploitation
scenario, an organizational user clicks a link that was previously
marked as referencing benign content.  Between the initial inspection
and user interaction, the adversary replaces the linked content with
ransomware, enabling credential theft and account takeover that
cascade into broader operational and financial damage.

A common model of TOCTOU-resistant malware defense (e.g., see
Microsoft Defender~\cite{Microsoft:Defender}, Symantec Endpoint
Protection~\cite{Broadcom:SEP}, and Cisco Secure
Endpoint~\cite{Cisco:Endpoint}) involves the receiving client querying
a remote (typically, cloud-based) server that hosts a curated,
proprietary blocklist of known malicious artifact descriptors (e.g.,
distance-preserving hashes~\cite{Charikar2002:SimHash, 
Kornblum06:ssdeep, Roussev10:sdhash, Oliver13:TLSH, 
Li15:FHReview}), immediately prior to the
artifact's execution.  Because the blocklist is maintained by the
vendor, it can update the blocklist frequently from threat feeds,
without demanding significant client-side resources (e.g., for local
signature databases).  However, this approach is known to introduce
privacy risks by exposing benign content descriptors to the server,
and additionally adds query latency on the critical path of execution.

The design we consider here improves on this popular model of malware
defense in two ways.  First, our design improves the interaction with
the server to be privacy-preserving for both the checking client's
content and the server's blocklist.  However, the costs of doing so
would add to the receiving client's latency, and so the second
improvement we explore is to move the blocklist check to the sending
client, so that this expensive check need not be incurred on the
critical path of the receiving client using it.

Because the receiving client does not trust the sending client in our
scenario, it requires evidence that the artifact was checked.  A
natural approach would be for the service to (blindly) sign the
content after verifying that it is not blocklisted, and for the sender
to forward this signature to the receiver. However, if the receiver
only locally verifies this signature, then the TOCTOU vulnerability is
reintroduced: the blocklist may have changed after the signature was
issued. Content that was signed while policy-compliant may later
become blocked. Therefore, the receiver must still interact with the
service at the time of use to ensure that the artifact complies with
the \textit{latest} policy.

\subsection{\addText{Definitions}}
We are concerned with building a three-party primitive, 
\primitiveName (\primitiveNameShort in short),
where a sending client ($\client[1]$) checks whether its input $\cInput$
is present on a blocklist held by a server (\server).  If not, the
server will add a token to an allowlist so that a receiving client
($\client[2]$) can confirm that the value $\cInputAlt$ it receives
from $\client[1]$ is present on the allowlist.  At the same time,
$\cInput$ is hidden from \server, while \server's blocklist is hidden
from $\client[1]$ and $\client[2]$. We start with defining
necessary notations.

\subsubsection{Notations}
Let $\genericFnFamily: \genericFnKeyspace \times \genericFnDomain
\rightarrow \genericFnRange$ denote a function family with
\textit{keyspace} \genericFnKeyspace, \textit{domain}
\genericFnDomain, and \textit{range} \genericFnRange.  For
$\genericFnKey \in \genericFnKeyspace$, we denote by
$\genericFnFamily[\genericFnKey]: \genericFnDomain \rightarrow
\genericFnRange$ the function with
$\genericFnFamily[\genericFnKey](\genericFnDomainElmt) =
\genericFnFamily(\genericFnKey, \genericFnDomainElmt)$ for each
$\genericFnDomainElmt \in \genericFnDomain$.  Let \nonnegativeReals
denote the non-negative real numbers. 
For the description of the cryptographic primitives, we 
use $\finiteField$ to denote a finite field, and $\polyNotation{p} \in
\polyField$ to denote a polynomial whose coefficients are drawn from
$\finiteField$. A random polynomial $\genericPoly \getsr \polyField$
is a polynomial with coefficients uniformly sampled from
$\finiteField$. Given two vectors $\genericVec[1], \genericVec[2] \in 
\finiteField^{\numberOfPointsSim}$, $\genericVec[1] \hadamardProd \genericVec[2] \in 
\finiteField^{\numberOfPointsSim}$, and ``\hadamardProd'' is the Hadamard product. 

\subsubsection{\primitiveNameShort Definition}

\primitiveNameShort takes an input $\cInput \in \universe$
from $\client[1]$, where \universe is the input universe, and a list
of items $\blockedElmts \subset \universe$ from \server. The primitive
should produce a binding and hiding mapping on $\cInput$
to add to \server's allowlist \textit{only if}
$\cInput \not\in \blockedElmts$; if $\cInput \in
\blockedElmts$, the protocol should abort. $\client[1]$ then forwards
$\cInput$ to $\client[2]$. $\client[2]$ will interact with \server to
privately test membership of $\cInput$ on \server's allowlist.

That said, there are many scenarios where the blocklist \blockedElmts is 
too large to be specified directly as an input to a \primitiveNameShort 
protocol (as we will show in \secref{sec:eval}). Therefore, we focus on 
settings where the set \blockedElmts can be embedded into a metric 
space \metricSpace with distance metric $\distance: \metricSpace 
\times \metricSpace \rightarrow \nonnegativeReals$, where 
$\finiteField$ denotes a finite field and $\nonnegativeReals$
denotes the non-negative real numbers. 
The embedded 
set can be summarized by a much smaller set $\policyLan \subset 
\metricSpace$.  That is, in lieu of specifying \blockedElmts to the protocol 
directly, \server instead specifies a blocklist $\policyLan \subset \metricSpace$, 
$\setSize{\policyLan} \ll \setSize{\blockedElmts}$, and a threshold \threshold 
such that elements of \blockedElmts, once embedded in $\metricSpace$, 
should fall within distance \threshold of some elements of \policyLan. 
Mathematically, for an output \cInputEmbedding embedded from an element 
in \blockedElmts,
\begin{align*}
	\blocked{}{\policyLan}{\threshold}(\cInputEmbedding)
	& \defeq \left(\exists \lanElmt \in \policyLan: \distance\left(\cInputEmbedding, \lanElmt\right) \le \threshold\right).
\end{align*}

Naturally, such an embedding will not be perfect, and will introduce non-zero false-accept and false-reject rates. Here we consider a randomized embedding function family of the form $\embed: \embedKeySpace \times \universe \rightarrow \metricSpace$.
For instance, the keyspace \embedKeySpace might include the random
vectors used for projections in locality-sensitive hashing algorithms
(e.g.,~\cite{datar2004locality}).  We then define true- and false-reject
rates as follows:
{\small
\begin{align*}
	\trueRejectRate{\blockedElmts}{\policyLan}{\threshold}
	& \defeq \min_{\cInput \in \blockedElmts} 
	\cprob{\Big}{\blocked{\policyLan}{\threshold}(\cInputEmbedding)}{\embedKey \getsr 
	\embedKeySpace, \cInputEmbedding \gets \embed[\embedKey](\cInput)} \\
	\falseRejectRate{\blockedElmts}{\policyLan}{\threshold}
	& \defeq \max_{\cInput \in \universe\setminus\blockedElmts} \cprob{\Big}{\blocked{}{\policyLan}{\threshold}(\cInputEmbedding)}{\embedKey \getsr \embedKeySpace, \cInputEmbedding \gets \embed[\embedKey](\cInput)}
\end{align*}
}

In words, \trueRejectRate{\blockedElmts}{\policyLan}{\threshold} is a
lower bound on the probability of blocking an input \cInput that
should be blocked, and
\falseRejectRate{\blockedElmts}{\policyLan}{\threshold} is an upper
bound on the probability of blocking an input \cInput that should not
be blocked.  As usual, false- and true-accept rates can be defined by
$\falseAcceptRate{\blockedElmts}{\policyLan}{\threshold} \defeq 1 -
\trueRejectRate{\blockedElmts}{\policyLan}{\threshold}$ and
$\trueAcceptRate{\blockedElmts}{\policyLan}{\threshold} \defeq 1 -
\falseRejectRate{\blockedElmts}{\policyLan}{\threshold}$,
respectively.  Naturally, we strive to maximize
\trueRejectRate{\blockedElmts}{\policyLan}{\threshold} and
\trueAcceptRate{\blockedElmts}{\policyLan}{\threshold}.
The definition implies that inputs are not uniformly distributed, while 
keys are. This corresponds to the applications, where inputs are user-chosen.

\begin{defn}[\PrimitiveName]
  \label{dfn:primitive}
  A \primitiveName for blocklist \blockedElmts is a three-party protocol
  among a sending client ($\client[1]$), a receiving client ($\client[2]$),
  and a server (\server), consisting of two phases:
  \explicitCheckPhase and \implicitCheckPhase.
  \begin{itemize}[nosep,leftmargin=1em,labelwidth=*,align=left]
  \item {\bf \ExplicitCheckPhaseHeading:} $\client[1]$ inputs $\cInput \in
    \universe$ and \server inputs $\blockedElmts \subset \universe$,
    $\policyLan \subset \metricSpace$, $\threshold \in \nonnegativeReals$ 
    (with both \policyLan and \threshold computed from \blockedElmts).  
    This phase \textit{succeeds} with probability at least
    \trueAcceptRate{\blockedElmts}{\policyLan}{\threshold} if $\cInput
    \not\in\blockedElmts$ and with probability at most
    \falseAcceptRate{\blockedElmts}{\policyLan}{\threshold} otherwise.
    If this phase succeeds, then $\client[1]$ outputs
    1 and $\server$ outputs state \serverState.  Otherwise, \server aborts
    and \client[1] outputs 0.

   \item {\bf \ImplicitCheckPhaseHeading:} A client $\client \in
     \{\client[1], \client[2]\}$ inputs $\cInputAlt \in \universe$ and
     the server \server inputs \serverState. This phase
     \textit{succeeds} if \serverState contains an output from a
     successful \explicitCheckPhase using $\cInputAlt$.  If so,
     \client outputs 1.  Otherwise, \client outputs 0.  \server has no
     output.
   \end{itemize}
\end{defn}
We will define the guarantees that we want to provide from
these steps depending on the following threat models.


\subsubsection{Threat model}
\label{sec:framework:threat-model}

The threat models we consider for our applications in 
\secref{sec:app} are:
\begin{enumerate}[label=\textbf{T\arabic*},nosep,align=left,labelwidth=*,leftmargin=1.7em]
\item \label{model:client-mal} The server \server and receiving client \client[2] are honest, but $\client[1]$ is malicious, 
  i.e., $\client[1]$ deviates from the
  protocol to circumvent the check. More specifically, a malicious
  $\client[1]$ should be prevented from causing \server to store a
  token for a blocklisted input during \explicitCheckPhase, which is successfully verified during 
an honest implicit check initiated by $\client[2]$. 
The malicious $\client[1]$ may also invoke an implicit check or collude with
some $\client[2]$ to learn more information 
about the allowlist, including the secrets contributed by \server in its computation.,
tokens stored at \server created by itself or by other users.
  We require that the allowlist token is binding, meaning that 
  $\client[1]$ cannot generate an alternative input $\cInputAlt$ that 
  is consistent with an existing allowlist token.
  In addition, the protocol should not leak more information about 
  \blockedElmts to $\client[1]$ than the decision ($\cInput \in 
  \blockedElmts$ or not) implies. We note that \server needs to limit 
  the number of \explicitCheckPhase attempts. Otherwise, $\client[1]$ 
  can online guess the blocklist entries.

\item $\client[1]$ and $\client[2]$ are honest, but \server is
  honest-but-curious.  In this case, our concern is still leaking
  information about $\cInput$ to \server from its interaction with
  $\client[1]$ and $\client[2]$. That is, the allowlist token is hiding, and 
  $\server$’s ability to distinguish between different values 
  of $\cInput$ is negligible.
  \label{model:server-hbc}
\end{enumerate}

We do not consider a threat model in which \server is malicious,
simply because we do not find that reasonable in our applications, 
where \server is contracted to, say, provide malware defense for 
\client[1] and \client[2]. In such cases, \server is incentivized to 
perform policy checking correctly.  In \appref{app:proof_framework}, 
we discuss how to adapt \primitiveNameShort against a server maliciously 
trying to learn $\cInput$. In addition, $\client[2]$ is considered honest 
since it is incentivized to perform the check correctly.  We discuss
more in \appref{app:proof_framework} how \primitiveNameShort protects 
against colluding $\client[1]$ and $\client[2]$.

\section{Preliminaries}
\label{sec:framework}


%

\myparagraph{Oblivious Linear Evaluation ($\oleFunc$)} 
OLE is a two-party
cryptographic primitive between a sender (e.g., a client) and a receiver (e.g., a server).  
The receiver's input is $\OLEAliceInput \in
\finiteField$; the sender's inputs are $\OLEBobInputOne,
\OLEBobInputTwo \in \finiteField$. The ideal function $\oleFunc$
returns to the receiver $\OLEBobInputOne \OLEAliceInput +
\OLEBobInputTwo$ \textit{without revealing \OLEBobInputOne and
	\OLEBobInputTwo}.  No information is revealed to the sender. In our
protocols, we will use the maliciously secure OLE construction by
\citet{Ghosh2017:OLE}.

\myparagraph{Noisy Polynomial Addition ($\npaFunc$)}
\citet{Ghosh2019communication} introduced a {\it noisy polynomial addition}
functionality \npaFunc that takes as input pairs of vectors $\randomVec[1], \cInputVec \in 
\finiteField^{\numberOfPointsSim}$ from one 
party, say a client, and another pair of vectors $\randomVec[2], \cInputVecAlt \in 
\finiteField^{\numberOfPointsSim}$ from another party, say a server, and computes a 
linear combination of these vectors. 
Specifically, the functionality returns to either one or both of the parties the vector 
$\randomVec[2] \hadamardProd \cInputVec + \randomVec[1] \hadamardProd \cInputVecAlt$, 
where ``\hadamardProd'' is the Hadamard product. The functionality does not reveal 
$\randomVec[1], \cInputVec$ to the server and $\randomVec[2], \cInputVecAlt $ to the client (beyond 
what can be learned from the output).

\myparagraph{One-time pairwise unpredictable function} 
We will require a function that is {\it one-time pairwise unpredictable}. Intuitively, a keyed-function $\genericFnFamily: \genericFnKeyspace \times \genericFnDomain \rightarrow \genericFnRange$
is one-time pairwise unpredictable if given the output at a specific point 
under a randomly chosen key, an adversary cannot guess the output at another distinct point. 
The function is one-time in the sense that each new invocation of the function 
requires a 
new random key selection to preserve unpredictable outputs.

\begin{defn} 
	
	A keyed function $\genericFnFamily: \genericFnKeyspace \times \genericFnDomain 
	\rightarrow 
	\genericFnRange$,
        $\setSize{\genericFnKeyspace} \ge \setSize{\genericFnRange}$,
        is one-time pairwise unpredictable if under a random choice of a key 
	$\genericFnKey \in \genericFnKeyspace$, and for any two distinct pairs of inputs $\cInput, 
	\cInputAlt 
	\in 
	\genericFnDomain$, and outputs $\genericFieldElmt, \genericFieldElmtAlt \in 
	\genericFnRange$
	
	\begin{eqnarray}
	 & \prob{\genericFnFamily[\genericFnKey](\cInput) = \genericFieldElmt} = \frac{1}{\setSize{\genericFnRange}}\\
		& \cprob{}{\genericFnFamily[\genericFnKey](\cInput) = \genericFieldElmt} 
		{\genericFnFamily[\genericFnKey](\cInputAlt) = \genericFieldElmtAlt} =
	\frac{\setSize{\genericFnRange}}{\setSize{\genericFnKeyspace}}  \label{eq:unpred_conditional}
	\end{eqnarray}

\noindent
where the probability is over the choice of the key. 
\end{defn}

The first condition in the definition implies that under a random choice of the key the function acts like a random function. 
The second condition implies that given the evaluation at a specific point, the value at a different point can be predicted with probability at most $\frac{\setSize{\genericFnRange}}{\setSize{\genericFnKeyspace}}$. If $\setSize{\genericFnKeyspace} \gg \setSize{\genericFnRange}$, then this probability is small and implies that the 
value is hard to predict.

We now extend this notion to a {\it one-time pairwise unpredictable permutation} 
over vectors in a finite field $\metricSpace$ We define the properties required from such a permutation below. 

\begin{defn}
\label{defn:unpredictable_metric}
	Let $\finiteField$ be a finite field and 
	$\numberOfPointsSim \ge 1$. A keyed-function $\genericFnFamily:
	\genericFnKeyspace^{\numberOfPointsSim} \times \metricSpace \rightarrow \metricSpace$ is
	one-time pairwise unpredictable over \metricSpace if
	for any pair of input $\cInputVec,\cInputVecAlt
	\in \metricSpace$ for which $| \cInputVec \Delta \cInputVecAlt | = \diffInComps$ , and outputs 
	$\vecNotation{\genericFieldElmt},
	\vecNotation{\genericFieldElmtAlt} \in \metricSpace$, 
	
	\begin{eqnarray}
	& \prob{\genericFnFamily[\genericFnKey](\cInputVec) = \vecNotation{\genericFieldElmt}} = \frac{1}{\setSize{\finiteField}^{\numberOfPointsSim}} \\
	&\cprob{}{\genericFnFamily[\genericFnKey](\cInputVec) = \vecNotation{\genericFieldElmt}}{
		\genericFnFamily[\genericFnKey](\cInputVecAlt) = \vecNotation{\genericFieldElmtAlt}}
	=
	\left(\frac{\setSize{\finiteField}}{\setSize{\genericFnKeyspace}}\right)^{\diffInComps}
	\end{eqnarray}

\noindent
where the probability is over the choice of the key $\genericFnKey
\in \genericFnKeyspace^{\numberOfPointsSim}$. Here $| \cInputVec
\Delta \cInputVecAlt | = \diffInComps$ means that \cInputVec and
\cInputVecAlt differ in at least $\diffInComps$ components.  

	
\end{defn}

Intuitively, this definition implies that if two vectors disagree on \diffInComps components, then the evaluations at these components remain unpredictable. If there is a one-time pairwise unpredictable function 
$\genericFnFamily: \genericFnKeyspace \times \finiteField \rightarrow \finiteField$, 
then 
we can construct a one-time pairwise unpredictable permutation over the metric space 
$\metricSpace$ by using \numberOfPointsSim independent 
instances of 
$\genericFnFamily$, one for each dimension in $\metricSpace$.

\section{\primitiveNameShort Framework}
\label{sec:framework:building-blocks}

\primitiveNameShort is realized with a general cryptographic 
framework. In this section, we will describe the individual 
building blocks of this framework as ideal functionalities, and 
then show how to combine them. 
In \secref{sec:app}, we will describe how to realize 
these ideal functionalities. 

Recall that
\primitiveNameShort has two phases that share states, an
\explicitCheckPhase and an \implicitCheckPhase. To
accommodate both threat models \ref{model:client-mal} and 
\ref{model:server-hbc}, we will further split the
\explicitCheckPhase functionality into two separate ideal
functionalities. Specifically, the \explicitCheckPhase can be
viewed as a two-step process: i) \textit{Embed-and-Map} 
$\embedAndMap$, that embeds $\client[1]$'s input into the
metric space $\metricSpace$, and ii) \textit{Test-and-Commit} 
$\testAndCommit$, that performs the predicate check
$\blocked{}{\policyLan}{\threshold}(\cdot)$ on the embedded input
against the embedded blocklist $\policyLan$ provided by \server.
Lastly, during \implicitCheckPhase, $\client[2]$ and \server
participate in iii) \textit{\ImplicitCheckPhase} $\authentication$, 
that allows fast verification on a previously checked input. 
We present the details for the three ideal functionalities,
which are summarized in \figref{fig:frameworkIF}.

\begin{figure}[t]
\begin{oframed}
\noindent\small

\smallskip\noindent \textbf{Parameters:~} An embedding function
$\embed: \embedKeySpace \times \universe \rightarrow \metricSpace$; 
a non-programmable random oracle $\hash: \{0,1\}^{\ast} \rightarrow 
\metricSpace$; and a one-time pairwise unpredictable 
permutation family over the metric space $\PRP: {\PRPKeySpace}^{\numberOfPointsSim} \times 
\metricSpace\rightarrow \metricSpace$.

\medskip\noindent\underline{$\embedAndMap$: Embed-and-Map functionality}

\smallskip\noindent \textbf{Inputs:~} $\client[1]$ inputs $\cInput
\in \universe$ and $\cMask \in \metricSpace$. \server inputs $\embedKey \in
\embedKeySpace$, $\sMask \in \metricSpace$, and $\PRPKey[1], \PRPKey[2]
\in {\PRPKeySpace}^{\numberOfPointsSim}$ .

\smallskip\noindent \textbf{Output:~} $\embedAndMap$ outputs 
$\prpOutputEmbedding \assign
\PRP[\PRPKey[1]](\embed[\embedKey](\cInput))$ and
$\prpOutputHash \assign \PRP[\PRPKey[2]](\authToken)$ where
$\authToken \assign \sMask \hadamardProd \hash(\cInput,
\embed[\embedKey](\cInput), \embedKey, \cMask)$
to $\client[1]$; and \server has no output.

\medskip\noindent\underline{$\testAndCommit$: Test-and-Commit Functionality}

\smallskip\noindent \textbf{Inputs:~} $\client[1]$ inputs
$\prpOutputEmbedding$, $\prpOutputHash \in \metricSpace$.
\server inputs $\policyLan \subset \metricSpace$, 
$\threshold \in \nonnegativeReals$, and keys $\PRPKey[1]$, $\PRPKey[2]
\in {\PRPKeySpace}^{\numberOfPointsSim}$.

\smallskip\noindent \textbf{Output:~} If
$\neg\blocked{\policyLan}{\threshold}(\PRP[\PRPKey[1]]^{-1}
(\prpOutputEmbedding))$, then $\testAndCommit$ outputs 1
and $\nonce$ to $\client[1]$ and $\prfOutput \assign \PRP[\PRPKey[1]]^{-1}(\prpOutputEmbedding) 
\linearComb \PRP[\PRPKey[2]]^{-1}(\prpOutputHash)$ to $\server$.
Otherwise, $\testAndCommit$ outputs 0
to $\client[1]$; \server aborts and learns $\{\lanElmt \in \blockList: \distance\left
(\PRP[\PRPKey[1]]^{-1}
(\prpOutputEmbedding), \lanElmt\right) \le \threshold\}$.\footnotemark

\medskip\noindent\underline{$\authentication$: \ImplicitCheckPhaseHeading Functionality}

\smallskip\noindent \textbf{Inputs:~} $\client[2]$ inputs
$\cICInput[1], \cICInput[2] \in \metricSpace$.  
\server inputs $\sMask, \prfOutput \in \metricSpace$.

\smallskip\noindent \textbf{Output:~} $\authentication$ outputs 
$\prfOutputAlt \gets \cICInput[1] \linearComb \sMask \hadamardProd
\cICInput[2]$  to $\server$. If $\prfOutput \equals \prfOutputAlt$, 
$\client[2]$ receives 1; otherwise, $\client[2]$ receives 0.

\end{oframed}
\captionof{figure}{Ideal Functionalities used in \primitiveNameShort }
\label{fig:frameworkIF}
\end{figure}

\footnotetext{Note that the server can implicitly learn this
  information even if the functionality did not return it, as an abort
  would imply \cInput is close to one of the server's inputs.}

\subsection{Embed-and-Map} 

Given $\client[1]$'s input $\cInput$ and a random mask $\cMask$, and 
\server's keys $\embedKey$, $\PRPKey[1]$, $\PRPKey[2]$, and 
a random mask $\sMask$, this functionality first computes 
$\prpOutputEmbedding \assign \PRP[\PRPKey[1]](\embed
[\embedKey](\cInput))$ and $\prpOutputHash \assign \PRP
[\PRPKey[2]](\sMask \hadamardProd
\hash(\cInput,\embed[\embedKey](\cInput), \embedKey, \cMask))$ where
$\embed$ is an embedding function, $\PRP:
{\PRPKeySpace}^{\numberOfPointsSim} \times \metricSpace\rightarrow \metricSpace$ is a
one-time pairwise unpredictable permutation family, and $\hash$ is a
hash function modeled as a non-programmable random oracle.
$\prpOutputEmbedding$ and $\prpOutputHash$ are expected to be
inputs to the next functionality (i.e., Test-and-Commit).
$\prpOutputEmbedding$ will be used (after inversion)
for testing against the blocklist
for approximate matching, and $\prpOutputHash$ will be used for
generating a hiding and binding token as an allowlist
entry for the \implicitCheckPhase. We highlight that though we split
the \explicitCheckPhase into two phases for concrete efficiency,
we maintain the security by carefully ``stitching'' the two
phases. 
\addText{
$\prpOutputEmbedding$ and $\prpOutputHash$ are
non-invertible to $\client[1]$ without the keys $\PRPKey[1], \PRPKey[2]$
under the unpredictable permutation $\PRP$.
} 
As we will see, if $\client[1]$ is malicious and uses
$\prpOutputEmbeddingAlt \ne \prpOutputEmbedding$ and/or
$\prpOutputHashAlt \ne \prpOutputHash$, it will be unable to find (with overwhelming probability) 
a $\cInputAlt \ne \cInput$ and a $\cMaskAlt \ne \cMask$ that matches an allowlist
token that will be stored on \server.
%
%
Therefore, $\client[1]$ is
``forced'' to submit a consistent input to $\testAndCommit$ that is generated from
$\embedAndMap$. This addresses security against a
malicious $\client[1]$ (threat model \ref{model:client-mal}).
\addText{
\server keeps $\PRPKey[1], \PRPKey[2]$ from $\embedAndMap$ 
for the next functionality $\testAndCommit$.
}


\subsection{Test-and-Commit}

During Test-and-Commit, the functionality \testAndCommit
inverts $\prpOutputEmbedding$ and  checks whether $\PRP[\PRPKey[1]]^{-1}(\prpOutputEmbedding)$ (which should
be equal to $\embed[\embedKey](\cInput)$ if $\client[1]$ is honest) is blocklisted. 
If not, \testAndCommit
inverts $\prpOutputHash$ and outputs an allowlist token $\prfOutput \assign 
\PRP[\PRPKey[1]]^{-1}(\prpOutputEmbedding) \linearComb \PRP[\PRPKey[2]]^{-1}
(\prpOutputHash)$ to \server for updating the allowlist. 
If $\client[1]$ inputs the same $\prpOutputEmbedding$ and $\prpOutputHash$ 
that it received from $\embedAndMap$ to \testAndCommit, 
then $\prfOutput = \embed[\embedKey](\cInput) + \sMask \times \hash(\cInput,\embed[\embedKey](\cInput), \embedKey, \cMask)$.

\addText{
Subsequently, $\client[1]$ forwards $\cInput$ and $\cMask$ to $\client[2]$ 
to use in an \implicitCheckPhase, while \server adds a new entry to the allowlist. 
For a successful check, \server discards $\PRPKey[1], \PRPKey[2]$ stored
from $\embedAndMap$, and adds a token $\langle \nonce, \embedKey, \sMask, 
\prfOutput \rangle$ to its allowlist.
We will explain the use of $\nonce$ when describing $\authentication$.
}
\client[1] only learns whether the \explicitCheckPhase succeeded in setting up 
an allowlist token, \emph{but does not learn the token itself}. 
As we will prove, $\prfOutput$ hides $\cInput$ due to the 
random oracle; this addresses the security requirements for 
a semi-honest \server (threat model \ref{model:server-hbc}).


\subsection{Implicit Check} 
To achieve TOCTOU resilience, parties $\client[2]$, who receive $\client[1]$'s 
input $\cInput$ and \cMask, must ensure that the content $\cInput$ remains
policy compliant with the blocklist at use time.  However, re-executing an
\explicitCheckPhase introduces non-trivial overhead since the cost must be at least linear 
to the size of $\policyLan$. \emph{In many applications (e.g., software 
distribution, email attachments, and CDN delivery), there will be 
more content receivers than the content sender and it is necessary 
for such receivers to repeat the check for the same input $\cInput$.}
To mitigate this, \server will retain an allowlist with $\langle \nonce, \embedKey, 
\sMask, \prfOutput \rangle$ from each successful \explicitCheckPhase.
Then, $\client[2]$ can verify that a content $\cInput$ is not blocklisted only 
if a previous \explicitCheckPhase succeeded for that input, achieving verification 
with a computation cost independent of the blocklist size.

During \implicitCheckPhase, $\client[2]$ can be provided the 
embedding key \embedKey by \server in the clear, because each 
\explicitCheckPhase requires a new embedding key, and once $\client[1]$ 
has set up \prfOutput during the \explicitCheckPhase, the embedding 
key does not provide any additional information/advantage to a malicious 
$\client[1]$ in circumventing the policy check. 
\addText{\embedKey is not an input to the one-time unpredictable function, 
its exposure has no impact on that property.}
$\client[2]$ uses the forwarded
input pair $\cInput, \cMask$ from $\client[1]$ to re-compute 
token $\cICInput[1] \assign \embed[\embedKey](\cInput)$ and $\cICInput[2]
\assign \hash(\cInput,\embed[\embedKey] (\cInput), \embedKey, \cMask)$ 
locally. 
%
%
%
Then $\server$ and $\client[2]$ interactively and privately 
evaluate $\prfOutputAlt \gets \cICInput[1] + \sMask \hadamardProd \cICInput[2]$. 
Only \server receives the output $\prfOutputAlt$ and no further 
information is given to either $\client[2]$ or \server. 
\server compares $\prfOutputAlt$ 
against the $\prfOutput$ stored from a previous \explicitCheckPhase. $\client[2]$ learns
whether the \implicitCheckPhase has succeeded, implying that $\prfOutputAlt$ has matched some
$\prfOutput$ from an allowlist token $\langle \nonce, \embedKey, \sMask, \prfOutput \rangle$. 
As we will show, $\prfOutput$ is computationally binding, and so it is infeasible for a malicious $\client[1]$ 
to find some pair of inputs $\cInputAlt, \cMaskAlt$ that match an allowlist token generated from different 
inputs $\cInput, \cMask$ during an \explicitCheckPhase. 
\addText{
We acknowledge that this design inevitably reveals the correspondence 
between $\client[1]$ and $\client[2]$ to \server. However, we emphasize 
that this disclosure stems from the underlying system architecture rather 
than from a weakness of the protocol itself. In many widely deployed 
systems—particularly email infrastructures~\cite{rfc5321} and policy enforcement 
services~\cite{rfc6409}—the server must observe both the source and the destination 
of each request. Eliminating such metadata would require substantial 
complexity and performance overhead and are outside the scope of our 
system design.
}

\myparagraph{Allowlist}
\addText{
So far we have asserted that the \implicitCheckPhase 
outperforms the \explicitCheckPhase because the allowlist is 
much smaller than the blocklist; however, this assumption has 
not yet been formally justified. In the actual implementation, 
each explicit check adds one entry to the allowlist.
\server keeps an identifier for each allowlist entry because 
when $\client[2]$ initiates \authentication, \server will need to 
forward $\embedKey$ associated with the allowlist entry 
$\client[2]$ is trying to query. Hence, the \implicitCheckPhase 
is indexed by a nonce $\nonce$ and is implemented with a fast 
hashtable lookup. $\nonce$ is generated from the 
\explicitCheckPhase timestamp and is independent of the file 
content to avoid potential leakage over $\cInput$. 
The storage overhead grows linearly with the number of explicit 
checks. We argue that this cost is inherent to our security model: 
each allowlist entry must incorporate a per-explicit-check server 
secret to prevent clients from probing existing entries, as well as a 
client secret to prevent the server from mounting offline attacks 
to recover w. 
}

\addText{
As was discussed in \secref{sec:usecase}, \primitiveNameShort 
enables resistance to TOCTOU attacks. Specifically, the allowlist 
is cleared whenever the blocklist is updated, at which point 
$\client[2]$ is redirected to the \explicitCheckPhase so that the 
updated policy can be applied. In practice, public blocklists 
are refreshed on timescales ranging from 
minutes~\cite{SpamhausSBL, GoogleSafeBrowsing}, 
to hours~\cite{proofpointET}, to days~\cite{Shadowserver}. Private 
blocklists are typically even harder for an attacker to learn or adapt 
to, and therefore may require updates less frequently. Consequently, 
the amortized benefit of the \implicitCheckPhase, which avoids 
repeated \explicitCheckPhase, will generally outweigh the occasional 
cost incurred by blocklist updates.
}

\addText{We note that TOCTOU vulnerabilities are primarily 
problematic when the inconsistency window is sufficiently large 
for an adversary to reliably exploit it. Therefore, we design
\primitiveNameShort to reduce the time window in which 
a TOCTOU inconsistency can be exploited rather than eliminating
the it entirely. First, even under optimal conditions, executing 
a query requires a non-zero amount of time, which inherently 
leaves a brief opportunity for a TOCTOU attack. By shrinking 
this window to the latency of a single query, \primitiveNameShort 
forces an attacker to act within a very narrow timeframe, 
substantially lowering the probability of successful exploitation. 
Second, there exists an unavoidable delay for newly emerged 
threats to be reflected in $\server$’s state, creating a potential 
zero-day exposure. We argue that this residual delay between 
server-side updates and client-side checks is inherent to any 
distributed system and does not represent a weakness specific to 
\primitiveNameShort. 
}

\myparagraph{Putting it all together}
We present the framework utilizing these functionalities in
\figref{fig:frameworkSH}, i.e., assuming that \embedAndMap,
\testAndCommit, and \authentication are secure in the threat models in
\secref{sec:framework:threat-model}. In \figref{fig:framework_interaction}, we visualize the interactions 
among the ideal
functionalities, $\client[1]$, $\client[2]$, and $\server$.  
 We present a formal proof of 
security of the framework in \appref{app:proof_framework}, relative to the
properties of $\PRP$ and $\hash$
\addText{
, and how to securely instantiate each ideal
functionality in \appref{app:proofs_pake}. 
}
Below we provide a sketch of the proof. 

\begin{figure}[t]
\centering
\includegraphics[width=0.8\textwidth]{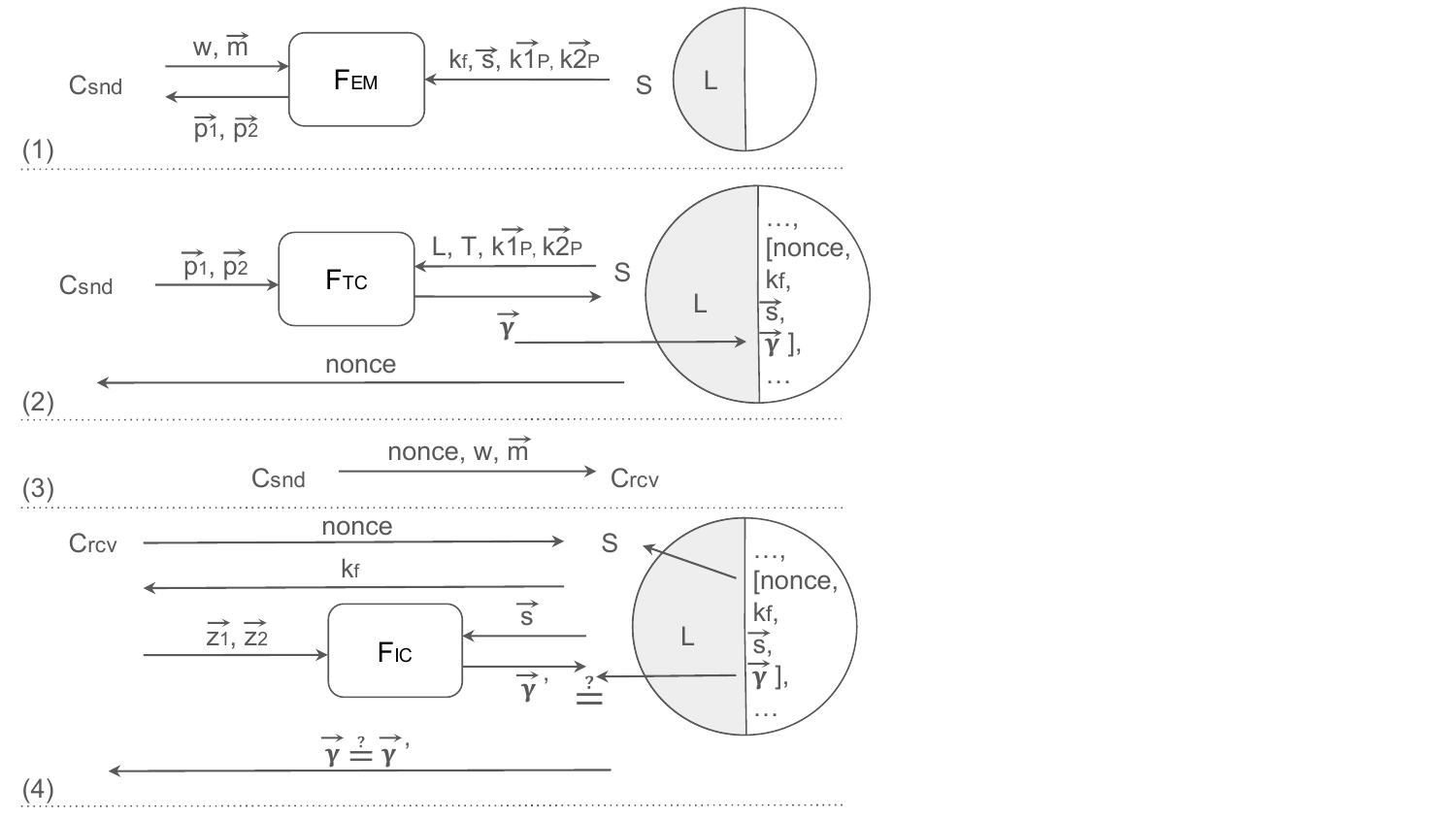}
\caption{Framework interaction\label{fig:framework_interaction}}
\end{figure}

\begin{figure}[t]
\begin{oframed}
{\small
\noindent\underline{\primitiveNameShort Framework}
		

\smallskip\noindent \textbf{Inputs:~} $\client[1]$ inputs $\cInput \in
\universe$.  \server inputs $\policyLan \subset \metricSpace$,
$\threshold \in \nonnegativeReals$, and $\PRPKey[1], \PRPKey[2]
\in {\PRPKeySpace}^{\numberOfPointsSim}$.

\smallskip\noindent
\textbf{\ExplicitCheckPhase:}
\begin{enumerate}[topsep=0.5em,leftmargin=1.75em,labelwidth=*,align=left]

\item $\client[1]$ selects $\cMask \getsr \metricSpace$. \server selects 
  $\sMask \getsr \metricSpace$ and
  $\embedKey \getsr \embedKeySpace$.

\item $\client[1]$ and \server call \embedAndMap. $\client[1]$ inputs
  $\cInput$ and $\cMask$. \server inputs \embedKey, $\sMask$, $\PRPKey[1]$, 
  $\PRPKey[2]$. $\client[1]$ learns $\prpOutputEmbedding \assign
  \PRP[\PRPKey[1]](\embed[\embedKey](\cInput))$ and $\prpOutputHash
  \assign \PRP[\PRPKey[2]](\authToken)$ where $\authToken \assign \sMask \hadamardProd
  \hash(\cInput, \embed[\embedKey](\cInput), \embedKey, \cMask)$. \server gets no output.
  \label{frameworkSH:embed}
			
\item $\client[1]$ sets
  $\prpOutputEmbeddingAlt \gets \prpOutputEmbedding$,
  and $\prpOutputHashAlt \assign \prpOutputHash$.
  \label{frameworkSH:glue}
  
\item $\client[1]$ and \server call \testAndCommit. $\client[1]$ inputs
  $\prpOutputEmbeddingAlt$ and $\prpOutputHashAlt$. \server inputs
  \policyLan, \threshold, $\PRPKey[1]$, and $\PRPKey[2]$.
  If
  $\neg\blocked{\policyLan}{\threshold}(\PRP[\PRPKey[1]]^{-1}(\prpOutputEmbeddingAlt))$
  then $\client[1]$ receives 1 and $\nonce$; \server receives $\prfOutput \assign 
  \PRP[\PRPKey[1]]^{-1}(\prpOutputEmbedding) 
  \linearComb \PRP[\PRPKey[2]]^{-1}(\prpOutputHash)$.
  Otherwise, $\client[1]$ receives 0; \server aborts and learns $\{\lanElmt \in \blockList: \distance\left
  (\PRP[\PRPKey[1]]^{-1}(\prpOutputEmbeddingAlt), \lanElmt\right) \le \threshold\}$.
  \label{frameworkSH:commit}

\end{enumerate}

\textbf{\ImplicitCheckPhase:}
\begin{enumerate}[topsep=0.5em,leftmargin=1.75em,labelwidth=*,align=left]
\setcounter{enumi}{4}

\item \client[1] sets $\nonceAlt \assign \nonce$, $\cInputAlt \assign \cInput$, 
  $\cMaskAlt \assign \cMask$
  and sends $\nonceAlt$, $\cInputAlt$ and $\cMaskAlt$ to $\client[2]$.

\item $\client[2]$ sends $\nonceAlt$ to \server. \server indexes the allowlist
  $\serverState$ with $\nonceAlt$ and obtains $\langle \embedKey, \sMask, \prfOutput \rangle$. 
  \server sets $\embedKeyAlt \assign \embedKey$, 
  $\sMaskAlt \assign \sMask$, and $\prfOutputAlt \assign \prfOutput$. \server sends
  \embedKeyAlt to $\client[2]$.

\item $\client[2]$ computes $\cICInput[1] \assign \embed[\embedKeyAlt](\cInputAlt)$
  and $\cICInput[2] \assign
  \hash(\cInputAlt, \embed[\embedKeyAlt](\cInputAlt), \embedKeyAlt, \cMaskAlt)$.
  \label{frameworkSH:client_input}

\item $\client[2]$ and \server call \authentication. $\client[2]$ inputs
  $\cICInput[1], \cICInput[2]$, and \server inputs $\sMaskAlt$. 
  $\server$ receives $\prfOutputIC \assign \cICInput[1] 
  \linearComb \sMaskAlt \hadamardProd
  \cICInput[2]$.
  \label{frameworkSH:authenticate}

\item \server checks whether $\prfOutputIC \equals \prfOutputAlt$. If so, 
  $\client[2]$ receives 1; otherwise, $\client[2]$ receives 0.

\end{enumerate}
}

\end{oframed}
\captionof{figure}{\primitiveNameShort framework}
\label{fig:frameworkSH}	
\end{figure}

\begin{theorem}
Assuming that there are secure protocolds realizing \embedAndMap, \testAndCommit and \authentication, 
the framework described in \figref{fig:frameworkSH} securely realizes \primitiveNameShort
against the adversaries described in \secref{sec:framework:threat-model}.
\end{theorem}

\begin{proofsketch}
We prove that the \primitiveNameShort framework is secure against threat models \ref{model:client-mal}
and \ref{model:server-hbc}. 

\smallskip\noindent\underline{Protection against threat model \ref{model:client-mal}:}
Consider an adversary that interacts with \server. 
During \explicitCheckPhase, the adversary may try to produce a file input $\cInput$ that should be blocked, but is not detected as such 
due to a false negative under the embedding function $\embed$. Since the embedding key $\embedKey$ 
is randomly selected and unknown to the adversary, we assume that it is hard for the adversary to produce such a false negative. 

Alternatively, the adversary may produce file inputs $\cInput \neq \cInputAlt$ and masks $\cMask \neq \cMaskAlt$, such that the \explicitCheckPhase succeeds with $\cInput, \cMask$ as inputs. However, the adversary 
sends $\cInputAlt, \cMaskAlt$ to an honest receiving client $\client[2]$ and the \implicitCheckPhase 
with $\cInputAlt, \cMaskAlt$ succeeds. This implies that the adversary has tricked \server into authenticating 
an input for an honest \client[2] that was not verified by a previous \explicitCheckPhase. 
We will show next that this happens with probability $\le |\finiteField|/|\PRPKeySpace|$. 

%
%
%

Suppose that the adversary provides inputs $\cInput, \cMask$ to $\embedAndMap$, and $\prpOutputEmbeddingAlt, \prpOutputHashAlt$ to \testAndCommit. Let $\prfOutput$ be the allowlist token 
stored at \server after the explicit check. Now consider that the adversary provides $\cInputAlt, \cMaskAlt$ to $\client[2]$. It must be that 
\begin{eqnarray}
\prfOutput = \embed[\embedKey](\cInputAlt) + \sMask \times \hash(\cInputAlt, \embed[\embedKey](\cInputAlt), \embedKey, \cMaskAlt) \label{eq:token_defn}
\end{eqnarray}
\noindent
Otherwise, the implicit check, where the honest $\client[2]$ provides inputs $\cInputAlt, \cMaskAlt$, and the server provides inputs $\embedKey, \prfOutput, \sMask$, will not verify.

Since, $\testAndCommit$ sets $\prfOutput = \PRP[\PRPKey[1]]^{-1}(\prpOutputEmbeddingAlt) + \PRP[\PRPKey[1]]^{-1}(\prpOutputHashAlt)$, we have
\[
\PRP[\PRPKey[1]]^{-1}(\prpOutputEmbeddingAlt) + \PRP[\PRPKey[1]]^{-1}(\prpOutputHashAlt) = \embed[\embedKey](\cInputAlt) + \sMask \times \hash(\cInputAlt, \embed[\embedKey](\cInputAlt), \embedKey, \cMaskAlt)
\]
\noindent
If $\prpOutputEmbeddingAlt \neq \prpOutputEmbedding$ or 
$\prpOutputHashAlt \neq \prpOutputHash$, then 
$\prfOutput$ is unpredictable without $\PRPKey[1], \PRPKey[2]$ 
(\dfnref{defn:unpredictable_metric}). In that case, the adversary 
produces $\cInputAlt, \cMaskAlt$ consistent with \eqref{eq:token_defn} 
only with probability $\le |\finiteField|/|\PRPKeySpace|$. So, it must be that 
$\embed[\embedKey](\cInputAlt) = \embed[\embedKey](\cInput)$ and 
$\hash(\cInputAlt, \embed[\embedKey](\cInputAlt), \embedKey, 
\cMaskAlt) = \hash(\cInput, \embed[\embedKey](\cInput), \embedKey, 
\cMask)$. Thus, an adversary that produces $\cInput \neq \cInputAlt$ 
such that $\hash(\cInputAlt, \embed[\embedKey](\cInputAlt), \embedKey, 
\cMaskAlt) = \hash(\cInput, \embed[\embedKey](\cInput), \embedKey, 
\cMask)$, also finds a collision in the random oracle 
$\hash$, which happens with negligible probability in $\bigO{\numberOfPointsSim \cdot \log |\finiteField|}$.  

\medskip\noindent\underline{Protection against threat model \ref{model:server-hbc}:}
The only information \server obtains from the protocol is whether
$\blocked{\blockedElmts}{\policyLan}{\threshold}(\embedESP(\cInput))$
and the value $\prfOutput$. 
As mentioned in \figref{fig:frameworkIF}, when 
$\blocked{\blockedElmts}{\policyLan}{\threshold}(\embedESP(\cInput))$,
there would be inevitable leakage to \server. So we ensure that the 
\server should learn strictly no information about $\cInput$ if 
$\neg\blocked{\blockedElmts}{\policyLan}{\threshold}(\embedESP(\cInput))$.
Since $\cMask$ is selected randomly and $\hash$ is a non-programmable 
random oracle, $\hash(\cInput, \embed[\embedKey](\cInput), \embedKey, \cMask)$ 
is uniformly distributed $\metricSpace$. 
Thus, $\prfOutput$ hides $\cInput$ from $\server$. 
\end{proofsketch}

We provide the full framework proof in \appref{app:proof_framework}.

\section{Protocol Description}
\label{sec:app}

In this section, we will describe how to efficiently realize the 
ideal functionalities leveraged in the \primitiveNameShort 
framework. \figref{fig:protocol_overview} provides a high-level 
visualization of the \explicitCheckPhase and 
\figref{fig:protocol_overview_ic} visualizes the \implicitCheckPhase.

\subsection{Protocol for Embed-and-Map}
\label{sec:app:embed-and-map}
Recall that \embedAndMap computes $\prpOutputEmbedding 
\assign \PRP[\PRPKey[1]](\embed[\embedKey](\cInput))$ and
$\prpOutputHash \assign \PRP[\PRPKey[2]](\authToken)$ where
$\authToken \assign \sMask 
\hadamardProd \hash(\cInput, \embed[\embedKey](\cInput), 
\embedKey, \cMask)$. Here, $\embed$ is a function mapping $\cInput$'s 
input file to a vector in $\metricSpace$, that generates close outputs 
for similar inputs. Ideally, we would like to use a fuzzy hashing 
scheme, commonly used for malware similarity 
tests~\cite{Kornblum06:ssdeep, Roussev10:sdhash, Oliver13:TLSH}. 
Unfortunately, none of these schemes produce outputs in 
$\metricSpace$. Instead, we will define $\embed$ as a composition 
of two functions as described below.

\subsubsection{Malware Hashing and Embedding Function}
\label{sec:app:embed-and-map:hash}
We begin by selecting hashing schemes that embed files into the 
Hamming space $\TLSH: \embedKeySpace \times \{0,1\}^{\ast} 
\rightarrow \{0,1\}^{\degreeOfKey{1}}$. We provide more details
on how the hashing scheme is selected in \secref{sec:eval}.
Then, we will define an injective map, $\vecMap$ from bit-vectors 
to $\metricSpace$. Finally, we have $\embed[\embedKey](\cInput) = 
\vecMap(\TLSH[\embedKey](\cInput))$.  

\begin{defn}[Hamming to \metricSpace map]
\label{defn:Hmap}
The mapping function $\vecMap:
\{0,1\}^{\degreeOfKey{1}} \rightarrow \finiteField^{\numberOfPointsSim}$ is parameterized by the set of arbitrary points  $\setOfEvalPts \assign \{\evalPt[1], \dots, \evalPt[\numberOfPointsSim]\} \subset \finiteField$ for $\numberOfPointsSim \ge \numberOfPoints$, and works as follows: 
\begin{enumerate}[nosep, leftmargin=1.6em,labelwidth=*,align=left]
	\item \label{step:hammingPSI:map} {\bf Map to set:} Given
	$\cInputVec \in \{0,1\}^{\degreeOfKey{1}}$, use an injective function $\mapFn:
	\{0,1\} \times \nats[\degreeOfKey{1}] \rightarrow \finiteField$ to create a
	set $\genericSetVar[\cInputVec] = \{\mapFn(\cInputVec[\rootIdx], \rootIdx)\}_{\rootIdx \in 
	  \nats[\degreeOfKey{1}]}$, where
        $\nats[\degreeOfKey{1}] = \{1, \ldots, \degreeOfKey{1}\}$ and
	$\cInputVec[\rootIdx]$ is the \rootIdx-th component of $\cInputVec$. 
	
	\item \label{step:hammingPSI:poly_create} {\bf Create a polynomial:}
	Compute a polynomial $\polyFromEmbedding[\cInputVec] \assign
	\prod\limits_{\genericSetElmt
		\in \genericSetVar[\cInputVec]} (\genericFieldElmt - \genericSetElmt)$, a polynomial with roots in $\genericSetVar[\cInputVec]$. 
	
	\item \label{step:hammingPSI:eval} {\bf Evaluate the polynomial:}
	Evaluate $\polyFromEmbedding[\cInputVec]$ at the points in \setOfEvalPts and output 
	a vector $\genericVec \assign \langle
	\polyFromEmbedding[\cInputVec][\evalPt[1]], \dots,
	\polyFromEmbedding[\cInputVec][\evalPt[\numberOfPointsSim]]\rangle$.

\end{enumerate}

\end{defn}

The function 
\vecMap is injective for any value $\numberOfPointsSim \ge \degreeOfKey{1} + 1$ because the 
polynomial $\polyFromEmbedding[\cInputVec]$ has $\degreeOfKey{1}$ roots, and such a 
polynomial can be uniquely defined by $\degreeOfKey{1} + 1$ points.
Also note that given two vectors $\cInputVec \in \{0,1\}^{\degreeOfKey{1}}$ 
and $\cInputVecAlt \in \{0,1\}^{\degreeOfKey{1}}$,
\begin{align}
	\hammingDist{\cInputVec}{\cInputVecAlt} \le \distInMaskedSets & \iff
	\setSize{\genericSetVar[\cInputVec] \cap \genericSetVar[\cInputVecAlt]} \ge
	(\degreeOfKey{1} - \distInMaskedSets) \label{eq:hamming_distance_diff}
\end{align}
\noindent
This follows directly from the definition of Hamming distance. 
It is immediately obvious from \eqref{eq:hamming_distance_diff}
that if $\cInputVec = \TLSH[\embedKey](\cInput)$ and $\cInputVecAlt = \TLSH[\embedKey](\cInputAlt)$, then    
$\hammingDist{\TLSH[\embedKey](\cInput)}{\TLSH[\embedKey](\cInputAlt)}$ can be calculated by computing  
$\setSize{\genericSetVar[\cInputVec] \cap \genericSetVar[\cInputVecAlt]}$. As we will see later, this is exactly 
how the distance check is performed in \testAndCommit.

\subsubsection{Embedding with Reusable Garbled Circuits}
\label{sec:app:embed-and-map:reusable}
Given the description of the embedding function $\embed$, 
a protocol realizing $\embedAndMap$ will interactively compute 
$\embed[\embedKey](\cInput) = \vecMap(\TLSH[\embedKey]
(\cInput))$. While doing so, it hides \server's inputs---the embedding 
key $\embedKey$, the keys for the permutation $\PRPKey[1], 
\PRPKey[2]$, and the mask $\sMask$---from the \client[1]. It also 
hides \client[1]'s inputs---the file $\cInput$ and the mask
$\cMask$---from \server. We realize this requirement using a garbled circuit. 
Specifically, upon being provided $\cInput \in \maxFileLen$ by 
$\client[1]$, and $\embedKey \in \embedKeySpace$ from \server, 
the garbled circuit computes $\embed[\embedKey](\cInput)$, and 
outputs $\embed[\embedKey](\cInput)$ to $\client[1]$.  However, a 
naive implementation will be impractical because the size of the 
garbled circuit will scale linearly in the size of the input file $\cInput$. 
Since we do not make any assumptions about the input size, the 
computation and communication will be prohibitive for large files. 

To address the scalability issues of the naive solution, we employ 
reusable garbled circuits. Unlike classical Yao 
garbling~\cite{Yao82:2PC, Yao86:GC}, which requires a fresh circuit 
for each execution, reusable garbled circuits allow the evaluator to 
run the same garbled circuit repeatedly with different inputs. In our 
case, the evaluator can reuse the same circuit to handle recurring 
operations over distinct segments of the input (e.g., iterative scans, 
loop-based processing, or sliding-window evaluations). These 
reusable circuits are then combined with non-reusable segments, 
thereby emulating executing a monolithic garbled circuit in a more 
efficient manner.

Many reusable garbled circuit 
schemes~\cite{Goldwasser13:RGC, Agrawal17:RGC} rely on expensive 
cryptographic primitives (e.g., FHE) that limit their practicality. In 
addition, the security guarantees provided by such schemes are stronger 
than our requirement in \secref{sec:framework:threat-model}. More 
specifically, we observe that our requirements are weaker along two axes: 

\begin{enumerate}
\item We require reusability only between the two parties:
  $\client[1]$ and \server.  In contrast, FHE-based reusable garbled
  circuits allow anyone that possesses the corresponding public key to
  evaluate.

\item The fuzzy hashing algorithm implemented by the garbled circuit 
need not be kept private within the garbled circuit. The security of 
embedding relies only on the input embedding key \embedKey of \server 
(the garbler) being kept secret from $\client[1]$ (the evaluator).
\end{enumerate}


Based on these observations, we employ optimizations grounded 
in the CRGC framework~\cite{CRGC}. Specifically, CRGC obfuscates a Boolean circuit 
$\circuit$ to $\circuitAlt$ and the garbler's input $\garblerInput$ to 
$\garblerInputAlt$, while maintaining the circuit functionality. In other 
words, it generates $\circuitAlt$ and $\garblerInput$ such that 
$\circuit(\garblerInput, \evaluatorInput) = \circuitAlt(\garblerInputAlt, 
\evaluatorInput)$ on any evaluator input $\evaluatorInput$. As a result, 
the evaluator can treat the garbled circuit as a ``black-box reusable 
program'', but cannot infer intermediate states across runs. This 
reusability, however, is highly restricted to the type of functionality
garbled. The security guarantees offered by CRGC is circuit-structure
dependent, and gates that do not fall into the categories for obfuscation 
may introduce leakage and expose correlations about input labels. 

We use the CRGC framework by carefully splitting the $\embed$
functionality into reusable and non-reusable sections according to
the specifications of CRGC. Ideally, operations with size linear to the input 
size will be encoded to a reusable garbled circuit so that $\client[1]$ can 
evaluate the same functionality repeatedly on its input locally. After this, 
the outputs from the reusable parts will be input to a non-reusable 
Yao's garbled circuit, which will be evaluated interactively between 
$\client[1]$ and \server through OT. More specifically, we rely on two main 
security guarantees offered by a CRGC circuit:

\begin{enumerate}
\item A balanced gate (i.e., an XOR or XNOR gate) in the reusable
  section obfuscates the garbler's input bit under the security
  assumptions of CRGC~\cite[Lemma~2]{CRGC}.

\item A gate in the non-reusable section does not leak garbler's input
  bit under the security assumptions of Yao's garbled circuit
  protocol~\cite[Lemma~4]{CRGC}.
\end{enumerate}

For the components of $\embed$ whose size scales linearly with the input, 
we first determine whether they can be compiled into CRGC form by checking 
whether they consist of the required balanced gates, or whether they can be 
augmented with additional balanced gates. In practice, this means that the 
$\server$ may insert extra XOR layers that preserve the original functionality 
while enabling CRGC compatibility. We then combine the resulting reusable 
CRGC segments with the remaining non-reusable parts. Under this transformation, 
$\client[1]$ is able to evaluate $\embed$ without learning the embedding key 
$\embedKey$ held by $\server$.

Since $\server$ receives no output from the reusable segments, the 
security under model~\ref{model:server-hbc} continues to rely on the standard 
guarantees of a conventional garbled circuit. For security model~\ref{model:client-mal},
$\client[1]$ can evaluate $\embed[\embedKey]$ on only a single input 
$\cInput$, as the non-reusable components are executed once per session.
In addition, if $\client[1]$ provides inconsistent inputs to different parts of the 
garbled circuit, $\client[2]$ will not be able to recompute $\prfOutput$ from 
$\cInput$. We argue that such a deviation by a malicious $\client[1]$ is equivalent 
to supplying an alternative input $\cInputAlt \neq \cInput$ to $\client[2]$, 
which is formally addressed in \appref{app:proof_framework}.

In the open-source CRGC 
implementation\footnote{\url{https://github.com/chart21/CRGC}}
an analysis tool enumerates all signal paths originating from the 
garbler’s input wires to provide a lower bound on potential leakage.
The library classifies a generator input as safe if and only if every such 
path contains at least two balanced gates and that the resulting circuit 
achieves the level of indistinguishability necessary for secure, 
reusable garbling. This tool supports our findings and confirms that our 
design prevents exposure of any bits of $\embedKey$. Additional background 
on how CRGC conceals the garbler’s input (i.e., $\server$’s $\embedKey$) 
is provided in \appref{app:CRGC}.

\subsubsection{$\hash$ as a non-programmable random oracle}
\label{sec:app:embed-and-map:npro}
In \appref{app:proof_framework}, we model $\hash$ as a random 
oracle for the purposes of analysis. However, when evaluating 
$\hash$ inside a garbled circuit, $\hash$ behaves as a 
non-programmable random oracle~\cite{Fischlin10:NPRO}, which 
can be represented by a circuit of size polynomial in the security parameter
and the input length. Our 
construction follows the standard methodology of treating hash 
functions as random oracles~\cite{Bellare1993:Random-oracles}, 
but we only rely on properties that hold in the non-programmable 
setting. In particular:
\begin{itemize}
\item The adversary may freely evaluate $\hash(\cdot)$ offline.

\item The simulator is not required to re-program or adapt the oracle output.
\end{itemize}
Therefore, $\hash$ can be safely instantiated with a fixed hash function 
embedded directly in the circuit, consistent with the non-programmable 
RO framework~\cite{Barnum23:NPRO}.

\subsubsection{Unpredictable permutation for \metricSpace}
\label{sec:app:embed-and-map:perm}
Recall that the output of \embedAndMap is an unpredictable permutation
applied to $\embed[\embedKey](\cInput)$ and $\authToken$, both vectors
in $\metricSpace$. We construct such an unpredictable permutation from
a random linear combination.\footnote{As it will be evident later,
using a linear combination means we can efficiently invert the
permutation using \npaFunc.} Specifically,

\begin{theorem}
	For a random selection of a key $\langle \PRPKeyVec[1], \PRPKeyVec[2]
	\rangle \in (\finiteField\setminus\{0\})^{\numberOfPointsSim} \times
	(\finiteField)^{\numberOfPointsSim}$, the function $\PRP[\langle
	\PRPKeyVec[1], \PRPKeyVec[2]\rangle](\genericVec) = \PRPKeyVec[1] \hadamardProd
	\genericVec  + \PRPKeyVec[2]$ is a one-time pairwise-unpredictable permutation
	over \metricSpaceESP.
\end{theorem}

The proof is in \appref{app:proof_prp}.

\myparagraph{Protocol}
The protocol \protocolEmbed realizing \embedAndMap is described in \appref{app:em}. The protocol computes using a reusable garbled circuit $\prpOutputEmbedding \assign
\PRP[\langle\PRPKeyEmbeddingVec[1], \PRPKeyEmbeddingVec[2]\rangle](\vecMap(\TLSH[\embedKey](\cInput))) = \PRPKeyEmbeddingVec[1] \hadamardProd \vecMap(\TLSH[\embedKey](\cInput)) + \PRPKeyEmbeddingVec[2]$, and $\prpOutputHash \assign  \PRP[\langle\PRPKeyHashVec[1], \PRPKeyHashVec[2]\rangle](\authToken) = \PRPKeyHashVec[1] \hadamardProd \authToken + \PRPKeyHashVec[2]$, where $\authToken \assign \sMask \hadamardProd \hash(\cInput,
\embed[\embedKey](\cInput), \embedKey, \cMask)$ and is also computed by the circuit. 

\subsection{Protocol for Test-and-Commit}
\label{sec:app:test-and-commit}
The functionality \testAndCommit privately compares $\client[1]$'s input 
$\embed[\embedKey](\cInput)$ to the embedded items in the blocklist, and 
produces the authentication token. Since we are using an embedding into 
the Hamming space, \testAndCommit must check that whether the Hamming 
distance between \client[1]'s input and any of the bit-vectors in the blocklist is 
greater than the threshold. To perform this check privately, we can apply an 
\emph{assumption-free}\footnote{An assumption-free fuzzy PSI scheme does 
not require any specific conditions on the data distribution. Since the output 
of malware hashes are not constrained to any specific distribution, we must 
leverage assumption-free schemes only.} fuzzy PSI scheme that compares 
vectors in the Hamming space, 
e.g., \cite{Chakraborti2023:DAPSI, Blass25:fuzzyPSI}. However, a naive 
application does not suffice because $\embed[\embedKey](\cInput)$ is not 
a bit-vector, but rather a vector in a field $\metricSpace$ generated using 
the map $\vecMap$. (We ignore the unpredictable permutation $\PRP$, 
because as we will see the permutation is inverted before performing the 
distance check.) 

We overcome this problem by leveraging the mapping function 
$\vecMap$ (\dfnref{defn:Hmap}), 
and defining a metric such that computing the distance between two vectors 
in $\metricSpace$ provides the Hamming distance between their preimages under $\vecMap$. Then, we modify an existing 
assumption-free fuzzy PSI protocol to compute Hamming distance using this 
metric. 

We define the distance metric {\it symmetric uncommon 
roots} (\symmUR) as follows. For polynomial 
$\genericPoly \in \polyField$, let $\uniqueRoots{\genericPoly}$ be 
its set of (unique) roots. Let $\genericSetVar[1] 
\symmDiff \genericSetVar[2]$ denote the symmetric difference 
between sets $\genericSetVar[1]$ and $\genericSetVar[2]$, i.e., 
$(\genericSetVar[1] \setminus \genericSetVar[2]) \cup
(\genericSetVar[2] \setminus \genericSetVar[1])$.

\begin{defn}[Symmetric Uncommon Roots \symmURDefn]
Let $\setOfEvalPts \assign \{\evalPt[1], \dots, \evalPt[\numberOfPoints] \}$ be 
a set of $\numberOfPoints$ unique, fixed elements in $\finiteField$.  
For vector $\genericVec = \langle \genericVecComp[1], \dots, 
\genericVecComp[\numberOfPoints] \rangle \in 
\finiteField^{\numberOfPoints}$, let $\polyFromVec{\genericVec} 
\in \polyField$ be the unique polynomial of degree $\le 2 
\degreeOfKey{1}$ that interpolates the set of points 
$\{ (\evalPt[\ptIdx], \genericVecComp[\ptIdx]) \}_{\ptIdx \in \nats[\numberOfPoints]}$.
Then, for $\genericVec[1], 
\genericVec[2] \in \finiteField^{\numberOfPoints}$,
\begin{align*}
    \symmURDist{\genericVec[1]}{\genericVec[2]}
    & ~\defeq~ \setSize{\uniqueRoots{\polyFromVec{\genericVec[1]}} 
    \symmDiff \uniqueRoots{\polyFromVec{\genericVec[2]}}} \\
    & ~=~ \setSize{\uniqueRoots{\polyFromVec{\genericVec[1]}}} + 
    \setSize{\uniqueRoots{\polyFromVec{\genericVec[2]}}} -
    2\cdot\setSize{\uniqueRoots{\genericPoly}}
\end{align*}    
where $\genericPoly = \gcd(\polyFromVec{\genericVec[1]}, 
\polyFromVec{\genericVec[2]})$ is the greatest common 
divisor of \polyFromVec{\genericVec[1]} and 
\polyFromVec{\genericVec[2]}.
\end{defn}

Since the cardinality of symmetric difference between two sets is a
metric, \symmUR is also a metric.

\myparagraph{Computing Hamming distance from \symmURDefn}
Given two bit vectors $\cInputVec \in \{0,1\}^{\degreeOfKey{1}}$ 
and $\cInputVecAlt \in \{0,1\}^{\degreeOfKey{1}}$,
\begin{align}
	\hammingDist{\cInputVec}{\cInputVecAlt} \le \distInMaskedSets & \iff
	\setSize{\genericSetVar[\cInputVec] \cap \genericSetVar[\cInputVecAlt]} \ge
	(\degreeOfKey{1} - \distInMaskedSets) \nonumber \\
	& \iff \frac{1}{2} \left(2 \numberOfMasks -
	\setSize{\genericSetVar[\cInputVec] \symmDiff
		\genericSetVar[\cInputVecAlt]}\right) \ge 	(\degreeOfKey{1} - \distInMaskedSets)  \label{eq:diff}\\
	& \iff \symmURDist{\vecMap(\cInputVec)}{\vecMap(\cInputVecAlt)} \le \threshold \label{eq:diff_sur}
\end{align}
where $\threshold =  2\distInMaskedSets$. The equivalences follow by definition of Hamming distance, symmetric difference \eqnref{eq:diff}, and \symmUR \eqnref{eq:diff_sur}. In our context, the above relation means 
\begin{eqnarray}
\hammingDist{\TLSH[\embedKey](\cInput)}{\TLSH[\embedKey](\cInputAlt)} \le \distInMaskedSets \nonumber \\
\iff \symmURDist{\embed[\embedKey](\cInput)}{\embed[\embedKey](\cInputAlt)} \le 2\distInMaskedSets \label{eq:LSH_to_SUR}
\end{eqnarray}

\myparagraph{\symmURDefn computation from \npaFunc}
One advantage of using $\symmURDefn$ is that we can
compute the distance between two vectors contributed by mutually
untrusting parties {\it without explicitly revealing the vectors to
  each other}. The process is as follows. Let $\genericVec[1] = \vecMap(\cInputVec)$
and $\genericVec[2] = \vecMap(\cInputVecAlt)$ 
be two vectors contributed by $\client$ and \server,
respectively. They sample random polynomials $\randomPoly[1] \getsr
\polyField$ and $\randomPoly[2] \getsr \polyField$ respectively, of
degree $\degreeOfKey{1}$. 
Then, they locally compute $\randomVec[1]
\gets \langle \randomPolyEval[\evalPt][1] \rangle_{\evalPt \in
  \setOfEvalPts}$ and $\randomVec[2] \gets \langle
\randomPolyEval[\evalPt][2] \rangle_{\evalPt \in \setOfEvalPts}$.
Finally, using calls to \npaFunc,
they compute $\randomVec[2] \hadamardProd \genericVec[1]
+ \randomVec[1] \hadamardProd \genericVec[2]$. Then,
\begin{align}
  \label{eq:sur_random}
& \symmURDist{\genericVec[1]}{(\randomVec[2] \hadamardProd \genericVec[1] +
\randomVec[1] \hadamardProd \genericVec[2])} \nonumber\\
& = \symmURDist{\genericVec[2]}{(\randomVec[2] \hadamardProd \genericVec[1] +
    \randomVec[1] \hadamardProd \genericVec[2])} \nonumber\\
& = \symmURDist{\genericVec[1]}{\genericVec[2]} & \mbox{(w.h.p.)} 
\end{align}
\eqnref{eq:sur_random} is due to the fact that $\randomPoly[1]$ and
$\randomPoly[2]$ are random polynomials, and so with overwhelming
probability, they do not share common roots with 
$\polyFromVec{\genericVec[1]}$ and $\polyFromVec{\genericVec[2]}$ that interpolate
$\genericVec[1]$ and $\genericVec[2]$, respectively. Thus,
$\gcd(\polyFromVec{\genericVec[1]},\combPoly) =
\gcd(\polyFromVec{\genericVec[2]},\combPoly)
=\gcd(\polyFromVec{\genericVec[1]},\polyFromVec{\genericVec[2]})$ with high
probability, where $\combPoly \gets
\randomPoly[2] \cdot \polyFromVec{\genericVec[1]} + \polyFromVec{\genericVec[1]}
\cdot \randomPoly[1]$. We prove this in \appref{app:background}.

Crucially, the value of $\numberOfPointsSim \in (\degreeOfKey{1}, 2\degreeOfKey{1} + 1)$ can be set as a function of $\degreeOfKey{1}$ and \threshold such that if
$\symmURDist{\genericVec}{\genericVecAlt} \le \threshold$, then $\server$
learns \genericVec  (and $\cInputVec$), otherwise {\it no} information is
revealed about $\cInputVec$ to the server~\cite{Chakraborti2023:DAPSI}. Intuitively, the idea is that with
less than $\numberOfPointsSim \le 2\degreeOfKey{1} + 1$ points, it is not
possible
to uniquely define and interpolate the polynomial $\randomPoly[2] \cdot
\polyFromEmbedding[\cInputVec] + \randomPoly[1] \cdot
\polyFromEmbedding[\cInputVecAlt]$
and thus $\randomVec[1], \randomVec[2]$,  unless
$\polyFromEmbedding[\cInputVec]$ and $\polyFromEmbedding[\cInputVecAlt]$ share
$2(\degreeOfKey{1} - \threshold)$ roots.

\myparagraph{Protocol}
The protocol \protocolCommitESP realizing \testAndCommit is presented in \appref{app:tc}. Intuitively, the protocol is an adaptation of the fuzzy PSI protocol of \citet{Chakraborti2023:DAPSI}, and checks the Hamming distance between \client[1]'s input and the embeddings in the blocklist by computing the \symmURDefn distance. Specifically, 
recall that one of the outputs of \protocolEmbed is $\prpOutputEmbedding = \PRP[\langle\PRPKeyVec[1], \PRPKeyVec[2]\rangle](\embed[\embedKey](\cInput)) = \PRPKeyVec[1] \hadamardProd \embed[\embedKey](\cInput)  + \PRPKeyVec[2]$. For each element $\lanElmtU$ in $\server$'s blocklist, the protocol computes 
$\randomVec[1] \hadamardProd \embed[\embedKey](\cInput) + \randomVec[2] \hadamardProd \embed[\embedKey](\lanElmtU)$ using a series of calls to $\npaFunc$. Then, using \eqnref{eq:sur_random} the protocol computes $\symmURDist{\embed[\embedKey](\cInput)}{\embed[\embedKey](\lanElmtU)}$, which allows it to check $\hammingDist{\TLSH[\embedKey](\cInput)}{\TLSH[\embedKey](\lanElmtU)} > \distInMaskedSets$ due to \eqnref{eq:LSH_to_SUR}. We set $\numberOfPointsSim \in (\degreeOfKey{1}, 2\degreeOfKey{1} + 1)$ using the parameters suggested in \cite{Chakraborti2023:DAPSI}
such that if $\hammingDist{\TLSH[\embedKey](\cInput)}{\TLSH[\embedKey](\lanElmtU)} > \distInMaskedSets$, then \server learns no information about the Hamming distance.

If the distance check is satisfied, the protocol proceeds to computing the authentication token $\prfOutput \assign \PRP[\PRPKey[1]]^{-1}(\prpOutputEmbedding) \linearComb \PRP[\PRPKey[2]]^{-1}(\prpOutputHash)$. Observe that $\prfOutput$ is a linear combination of two vectors in $\metricSpace$, and recall from \secref{sec:framework} that $\PRP$ (and so $\PRP^{-1})$ is a random linear combination of its input. Therefore, $\authToken$ is interactively computed by \client[1] and \server from $\prpOutputEmbedding$ and $\prpOutputHash$ using calls to $\npaFunc$, ensuring that \server does not learn $\PRP[\PRPKey[1]]^{-1}(\prpOutputEmbedding)$ or $\PRP[\PRPKey[2]]^{-1}(\prpOutputHash)$. 

\begin{theorem}
There is a secure protocol realizing \testAndCommit with $\bigO{\blockListSize \metricSpaceDimGen}$ communication and compute costs, assuming that there are secure protocols realizing \npaFunc and \oleFunc.  
\end{theorem}

\begin{figure}
  \begin{subfigure}{3in}
    \centering
    \scalebox{0.8}{
    \begin{tikzpicture}[
    >=latex,
    func/.style={draw, rectangle, thick, dashed},
    arr/.style={->, thick}
]

\coordinate (Rmid) at (1.35,0);
\coordinate (Lup)  at (-1.35,0.4);
\coordinate (Ldown)at ( -1.35,0);

\node[func, minimum width=2cm, minimum height=2cm, align=center] (F1) at (0,0) {Reusable \\ Garbled Circuit};
\node at ($(F1.north)+(0,0.3)$) {$\protocolEmbed$};

\node at (-2.5,1.5) {$\client[1]$};
\node at ( 2.5,1.5) {$\server$};

\draw[arr] (-4,0.4) -- node[above, font=\footnotesize] {$\cInput, \cMask$} (Lup);
\draw[arr] (Ldown) -- node[below, font=\footnotesize, yshift=1pt] {
$
\begin{aligned}
\prpOutputEmbedding & = \PRP[\PRPKey[1]] (\embed[\embedKey](\cInput)) \\
& = \PRPKeyEmbeddingVec[1] \hadamardProd \embed[\embedKey](\cInput) + \PRPKeyEmbeddingVec[2],  \\
\prpOutputHash & = \PRP[\PRPKey[2]](\sMask \hadamardProd \hashStar) \\
& = \PRPKeyHashVec[1] \hadamardProd (\sMask \hadamardProd \hashStar) + \PRPKeyHashVec[2]
\end{aligned}
$
} (-4.5,0);
\draw[arr] (4,0) -- node[above, font=\footnotesize, xshift=4pt] {$\embedKey, \PRPKey[1] = \langle \PRPKeyEmbeddingVec[1],
		\PRPKeyEmbeddingVec[2] \rangle$,}
node[below, font=\footnotesize, xshift=4pt] {$\PRPKey[2] = \langle \PRPKeyHashVec[1],
		\PRPKeyHashVec[2] \rangle, \sMask$} (Rmid);

\node[func, minimum width=2.7cm, minimum height=7.5cm] (F2) at ($(F1.south)+(0,-5)$) {};
\node at ($(F2.north)+(0,0.3)$) {$\protocolCommit$};

\node[func, minimum width=2cm, minimum height=1.2cm] (F3) at ($(F1.south)+(0,-2.2)$) {$\npaFunc$};

\draw[arr] (-4,-3.2) -- node[above, font=\footnotesize] {$\prpOutputEmbedding, \serversRandomPolyEval[1]$} (-1, -3.2);
\draw[arr] (4, -2.9) -- node[above, font=\footnotesize] {$\lanElmt[1], \PRPKey[1], \serversRandomPolyEvalAlt[1]$} (1, -2.9);
\draw[arr] (1,-3.5) -- node[below, font=\footnotesize, xshift=3pt] {$\serversRandomPolyEvalAlt[1] \hadamardProd \embed[\embedKey](\cInput) 
		+ \serversRandomPolyEval[1] \hadamardProd \lanElmt[1]$} (4,-3.5);

\node[func, minimum width=2cm, minimum height=1.2cm] (F4) at ($(F1.south)+(0,-4.5)$) {$\npaFunc$};
\node at ($(F4.north)+(0,0.3)$) {$\ldots$};

\draw[arr] (-4,-5.8) -- node[above, font=\footnotesize] {$\prpOutputEmbedding, \serversRandomPolyEval[\blockListSize]$} (-1, -5.8);
\draw[arr] (4, -5.5) -- node[above, font=\footnotesize] {$\lanElmt[\blockListSize], \PRPKey[1], \serversRandomPolyEvalAlt[\blockListSize]$} (1, -5.5);
\draw[arr] (1,-6.1) -- node[below, font=\footnotesize, xshift=3pt] {$\serversRandomPolyEvalAlt[\blockListSize] \hadamardProd \embed[\embedKey](\cInput) + 
		\serversRandomPolyEval[\blockListSize] \hadamardProd \lanElmt[\blockListSize]$} (4,-6.1);

\node[func, minimum width=2cm, minimum height=1.2cm] (F5) at ($(F1.south)+(0,-7.5)$) {$\npaFunc$};

\node[below=8pt of F4, align=left, font=\footnotesize] {
Test: $\forall \setIdx \in \nats[\blockListSize]$ \\
$\setSize{\embed[\embedKey](\cInput)-\lanElmt[\setIdx]} > 2\distInUniverse$ \\
};

\draw[arr] (-4,-8.3) -- node[above, font=\footnotesize] {$\prpOutputHash, \sum_{\setIdx=1}^{\blockListSize} \serversRandomPolyEval[\setIdx]$} (-1, -8.3);
\draw[arr] (-1,-8.9) -- node[above, font=\footnotesize] {$\nonce$} (-4, -8.9);
\draw[arr] (4.5, -8.2) -- node[above, font=\footnotesize, yshift=-1pt, xshift=3pt] {$\sum_{\setIdx=1}^{\blockListSize} \serversRandomPolyEvalAlt[\setIdx] \hadamardProd \embed[\embedKey](\cInput)
		+ \serversRandomPolyEval[\setIdx] \hadamardProd \lanElmt[\setIdx]$, } 
node[below, font=\footnotesize, yshift=1pt] {$\sum_{\setIdx=1}^{\blockListSize} \serversRandomPolyEvalAlt[\setIdx], \PRPKey[1], \PRPKey[2], \sMask$} (1, -8.2);
\draw[arr] (1,-8.9) -- node[below, font=\footnotesize] {$\prfOutput = \embed[\embedKey](\cInput) + \sMask \hadamardProd \hashStar$} 
node[below, font=\footnotesize, yshift=-20pt, xshift=5pt] {Store $\langle \nonce, \embedKey, \sMask, \prfOutput \rangle$} (4,-8.9);

\draw[arr] (-3.5, -11) -- node[above, font=\footnotesize] {Send $\cMask, \cInput, \nonce$ to $\client[2]$} (3.5,-11);

\end{tikzpicture}
    }
    \caption{Explicit check}
    \label{fig:protocol_overview}
  \end{subfigure}
  \\
  \begin{subfigure}{3in}
    \centering
    \scalebox{0.8}{
    \begin{tikzpicture}[
    >=latex,
    func/.style={draw, rectangle, thick, dashed},
    arr/.style={->, thick}
]

\coordinate (Lmid) at (-1,0);
\coordinate (Rup)  at ( 1,0.3);
\coordinate (Rdown)at ( 1,-0.3);

\node[func, minimum width=2cm, minimum height=1.2cm] (F) at (0,0) {$\oleFunc$};
\node at ($(F.north)+(0,0.3)$) {$\protocolValidate$};

\node at (-2.5,1.5) {$\client[2]$};
\node at ( 2.5,1.5) {$\server$};

\draw[arr] (-4,0) -- node[above, font=\footnotesize] {$\embed[\embedKey](\cInput), \hashStar$} (Lmid);
\draw[arr] ( 4,0.3) -- node[above, font=\footnotesize] {$\sMask, \prfOutput$} (Rup);
\draw[arr] (Rdown) -- node[below, font=\footnotesize] {$\embed[\embedKey](\cInput) + \sMask \hadamardProd \hashStar$} (4,-0.3);

\end{tikzpicture}
    }
    \caption{Implicit check}
    \label{fig:protocol_overview_ic}
  \end{subfigure}
  \caption{High-level description of \primitiveNameShort instantiation ($\hashStar \gets \hash(\cInput,\embed[\embedKey](\cInput), \embedKey, \cMask)$)}
  \label{fig:protocol}
\end{figure}

\subsection{\ImplicitCheckPhaseHeading}
\label{sec:app:implicit-check}
After receiving 
the file $\cInput$, the mask $\cMask$ and a nonce $\nonce$, the client $\client[2]$ 
forwards $\nonce$ to \server. $\nonce$ allows the server to identify the 
authentication token $\prfOutput$ for the file, the embedding key $\embedKey$,
and the mask $\sMask$. $\client[2]$ recomputes
$\embed[\embedKey](\cInput)$ and $\hash(\cInput,
\embed[\embedKey](\cInput), \embedKey, \cMask)$
after obtaining $\embedKey$ from $\server$. 
Then, $\client[2]$ and $\server$ invoke \oleFunc,
to compute $\prfOutput$, which is output to \server (see \figref{fig:protocol_overview_ic}). 
Finally, $\client[2]$ learns 
whether $\prfOutput$ was previously stored at $\server$, which implies 
$\cInput$ was checked. The full protocol is described in \appref{app:ic}. 


%

\begin{theorem}
There is a protocol realizing \authentication with $\bigO{\metricSpaceDimGen}$ communication and compute costs, assuming that there is a secure protocol realizing \oleFunc. 
\end{theorem}

\section{Evaluation}
\label{sec:eval}

We implemented the \primitiveNameShort-based malware 
detection\footnote{\url{https://github.com/zxinyuan/Half-Moon-Cookie}} in
C++11. We benchmarked these implementations on a Ubuntu 22.04 machine
with 12 cores and 32GB of RAM as the server, and multiple machines
that are each equipped with 4 cores and 16GB of RAM as the clients.
We used the emp-toolkit~\cite{emp-toolkit} for OT and garbled circuit
operations, and CRGC for reusable garbled circuits.  We used the
NTL~\cite{ntl} library for finite field arithmetic. The operations
were performed in a 128-bit prime-order field.  We implemented Ghosh
et al.'s OLE~\cite{Ghosh2017:OLE} and Ghosh et al.'s
PSI~\cite{Ghosh2019:PSI} as building blocks for
\primitiveNameShort. Our construction works for any metrics that can
be embedded into Hamming distance; see \appref{app:protocol}.


\subsection{Dataset}
\label{sec:eval:dataset}

As discussed in \secref{sec:usecase}, our \primitiveNameShort-based
malware detection can be applied to email attachments for ransomware
checking. We used the Enron dataset~\cite{Klimt04:Enron} for
collecting the general email attachment size distribution. In
\figref{fig:enron}, we show the cumulative distribution function (CDF)
calculated from $133{,}127$ files in
\url{https://trec-legal.umiacs.umd.edu/corpora/trec/legal10/},
including files uncompressed from .zip archives.  We ignored files
with size less than \SI{500}{\byte}.  The average size of an
attachment was \SI{193.6}{\kilo\byte} in the Enron dataset, while the
median was \SI{98.8}{\kilo\byte}. We separated out the executables
(files with extensions .exe, .dll, etc.), which had a mean size of
\SI{13.7}{\kilo\byte} and a median of \SI{3.8}{\kilo\byte}.  We
conducted our experiments based on such statistics.

For the malware blocklist, we used the Ember
dataset~\cite{Joyce2025:Ember}, embedding the $16{,}356{,}790$ malware
instances to clusters of different radii.

\begin{figure}[t]
\centering
    \resizebox{0.75\columnwidth}{!}{\begin{tikzpicture}[scale=0.9]

\begin{axis}[
    xmode=log,
    log basis x=10,
    xlabel={File size (\SI{}{\kilo\byte}, log scale)},
    ylabel={CDF},
    ymin=0, ymax=1,
    width=9cm,
    height=6cm,
    legend pos=south east,
]
\addplot+[thick, mark=*, color=black]
coordinates {
    (1, 0.02)
    (5, 0.09)
    (10, 0.13)
    (100, 0.74)
    (500,0.91)
    (1000,0.94)
    (2000,0.97)
    (5000, 1.00)
};
\addlegendentry{All attachments}

\addplot+[thick, dashed, mark=square*, color=black]
coordinates {
    (1, 0.12)
    (5, 0.56)
    (10, 0.72)
    (100, 1.00)
};
\addlegendentry{Executables}


\end{axis}
\end{tikzpicture}}
    \caption{Email attachment size distribution}
    \label{fig:enron}
\end{figure}

\subsection{\ExplicitCheckPhase}
\label{sec:eval:explicit}

We start in \secref{sec:eval:explicit:hash} by comparing our
\explicitCheckPhase protocol when implemented using the common malware
hashes.  Then, we turn to measuring throughput of our
\explicitCheckPhase protocol in \secref{sec:eval:explicit:throughput}.

\subsubsection{Choosing a hash function}
\label{sec:eval:explicit:hash}

Unlike cryptographic hashes, malware hashes (e.g., locality sensitive
hashes, context triggered piecewise hashes) can be used as file
identifiers that can be compared for estimating file similarity.
TLSH~\cite{Oliver13:TLSH}, ssdeep~\cite{Kornblum06:ssdeep},
sdhash~\cite{Roussev10:sdhash} are the most common options and are
used by VirusTotal and VirusShare.  We refer readers to the original
works on fuzzy hashing for detailed evaluations of embedding accuracy
and similarity detection performance. In general, the accuracy of
these fuzzy hashes can be tuned by adjusting the threshold parameter
that determines when two digests are considered
similar. \primitiveNameShort makes black-box use of the hash function
and the cost does not depend on the specific threshold value used to
classify matches; the computational and communication costs of our
scheme remain independent of this parameter choice. Thus,
practitioners may select the threshold that best fits their
accuracy-performance tradeoff without affecting the efficiency of
\primitiveNameShort.

We first illustrate in \tblref{tab:em_cost} the utility of our
optimizations described in \secref{sec:app:embed-and-map:reusable}
that leverage reusable garbled circuits (RGCs) to implement these hash
functions in the Embed-and-Map protocol (\protocolEmbed).  This table
compares the communication volumes and response times of
\protocolEmbed, where response time is measured at \client[1],
starting when it invokes \protocolEmbed and finishing it when it
completes that invocation with a lightly loaded \server.  For each
hash function, the table compares cases when the hash function
implementation does not (\noIndicator) or does (\yesIndicator)
leverage the RGC optimization described in
\secref{sec:app:embed-and-map:reusable}.  (In the case it does not,
regular garbled circuits are used, instead.)  The size of $\cInput$
was \SI{5}{\kilo\byte}, which is about the median size of an email
attachment executable.

\begin{table}[htb]
\centering
\begin{tabular}{@{}lc*{2}{S[table-format=1.1e1]}@{}}
\toprule
\multicolumn{1}{c}{hash}
& \multicolumn{1}{c}{RGCs}
& \multicolumn{1}{c}{\parbox[m]{4em}{\centering \textbf{Comm.} (\SI{}{\mega\byte})}}
& \multicolumn{1}{c}{\parbox[m]{5em}{\centering \textbf{Response Time} (\SI{}{\second})}} \\
\midrule
\multirow{2}{*}{TLSH}
& \noIndicator  & 6.8e4 & 1.0e3 \\
& \yesIndicator & 6.0e1 & 4.8e0 \\
\midrule
\multirow{2}{*}{ssdeep}
& \noIndicator & 5.5e3 & 2.0e2 \\
& \yesIndicator & 5.3e1 & 5.4e0 \\
\midrule
\multirow{2}{*}{sdhash}
& \noIndicator & 5.8e4 & 9.7e2 \\
& \yesIndicator & 1.3e2 & 2.4e1 \\
\bottomrule
\end{tabular}
\caption{Embed-and-map (\protocolEmbed) costs with (\yesIndicator) and
  without (\noIndicator) reusable garbled circuits (RGCs), with
  $\cInput$ of size \SI{5}{\kilo\byte}}
\label{tab:em_cost}
\end{table}

As shown in \tblref{tab:em_cost}, for all three hash functions, our RGC
optimization reduces communication volume of \protocolEmbed by over two
orders of magnitude and response time by at least one order of
magnitude.  Note that the no-RGC costs shown in \tblref{tab:em_cost}
serve as a very conservative lower bound for an implementation of a
full \explicitCheckPhase using a monolithic, conventional garbled
circuit.

Comparing across hash functions, TLSH provides the fastest response time 
and ssdeep provides the lowest communication volume.  The response time for
\protocolEmbed using sdhash (implemented using RGCs) is about $4\times$
that using ssdeep, and the response time using sdhash is $5\times$
more than that using TLSH. Using ssdeep or TLSH results in similar
communication volumes, though sdhash induces $\sim 2\times$ more
communication.

While ssdeep yields a comparably efficient implementation of
\protocolEmbed, unfortunately it embeds the digest to edit distance,
which is not directly compatible with our \primitiveNameShort design.
As a result, to use ssdeep for an \explicitCheckPhase, its
edit-distance embedding is further embedded into Hamming distance
within \protocolEmbed.  We do so through edit-sensitive
parsing~\cite{Cormode07:ESP}, which embeds edit-distance to L1
distance and is then compatible with \citet{Chakraborti2023:DAPSI}.
This extra embedding step results in an additional \SI{1.23}{\second}
in client response time and \SI{4.46}{\mega\byte} in total
communication in \protocolEmbed.  Moreover, the embedding vector size
must be increased to compensate for the approximation introduced by
the edit-sensitive parsing, in order to preserve the accuracy of
ssdeep, which results in greater cost in the Test-and-Commit protocol
implementation (\protocolCommitESP).  This extra cost causes the total
\explicitCheckPhase response time using ssdeep to exceed the response
time using TLSH by a more substantial margin.  This can be seen by
comparing \tblref{tbl:explicit-check:tlsh} and
\tblref{tbl:explicit-check:ssdeep}, which show the response times as
measured by \client[1] for a full \explicitCheckPhase, for various
sizes of $\cInput$ and $\policyLan$.

\begin{table}[htb]
\centering
\begin{subtable}[t]{0.75\columnwidth}
  \setlength{\tabcolsep}{1pt}
  \FPeval{\MinNumber}{14.2}
  \FPeval{\MaxNumber}{1773.4}
  \begin{tabular}{@{}cc|XXXXX}
    & \multicolumn{1}{c}{}
    & \multicolumn{5}{c}{$\setSize{\policyLan}$} \\
    \multirow{6}{*}{\rotatebox[origin=c]{90}{size of $\cInput$ (\SI{}{\kilo\byte})}}
    & \multicolumn{1}{c}{}
    & \multicolumn{1}{c}{$1e2$}
    & \multicolumn{1}{c}{$5e2$}
    & \multicolumn{1}{c}{$1e3$}
    & \multicolumn{1}{c}{$5e3$} 
    & \multicolumn{1}{c}{$1e4$} \\ \cline{3-7} \\[-8.5pt]
    & $5e0$ &   14.2 &   44.9 &   82.7 &  391.1 &  781.1 \\
    & $1e1$ &   19.2 &   49.9 &   87.7 &  396.5 &  786.2 \\
    & $1e2$ &  110.8 &  138.3 &  177.8 &  485.1 &  876.2 \\
    & $2e2$ &  207.8 &  240.2 &  282.5 &  586.5 &  978.3 \\
    & $1e3$ & 1012.8 & 1040.2 & 1078.3 & 1387.7 & 1773.4
  \end{tabular}
  \caption{TLSH}
  \label{tbl:explicit-check:tlsh}
\end{subtable}
\\[8pt]
\begin{subtable}[t]{0.75\columnwidth}
  \setlength{\tabcolsep}{1pt}
  \FPeval{\MinNumber}{44.7}
  \FPeval{\MaxNumber}{6100.7}
  \begin{tabular}{@{}cc|XXXXX}
    & \multicolumn{1}{c}{}
    & \multicolumn{5}{c}{$\setSize{\policyLan}$} \\
    \multirow{6}{*}{\rotatebox[origin=c]{90}{size of $\cInput$ (\SI{}{\kilo\byte})}}
    & \multicolumn{1}{c}{}
    & \multicolumn{1}{c}{$1e2$}
    & \multicolumn{1}{c}{$5e2$}
    & \multicolumn{1}{c}{$1e3$}
    & \multicolumn{1}{c}{$5e3$} 
    & \multicolumn{1}{c}{$1e4$} \\ \cline{3-7} \\[-8.5pt]
    & $5e0$ &   44.7 & 169.0 & 319.4 & 1548.4 & 3112.9 \\
    & $1e1$ &   61.8 & 186.9 & 336.7 & 1566.6 & 3129.8 \\
    & $1e2$ & 337.7 & 456.0 & 606.7 & 1839.3 & 3399.6 \\
    & $2e2$ &  631.8 & 756.5 & 904.7 & 2136.7 & 3709.7 \\
    & $1e3$ & 3033.2 & 3160.7 & 3319.0 & 4547.6 & 6100.7
  \end{tabular}
  \caption{ssdeep}
  \label{tbl:explicit-check:ssdeep}
\end{subtable}
\\[8pt]
\begin{subtable}[t]{0.75\columnwidth}
  \setlength{\tabcolsep}{1pt}
  \FPeval{\MinNumber}{45.5}
  \FPeval{\MaxNumber}{6819.4}
  \begin{tabular}{@{}cc|XXXXX}
    & \multicolumn{1}{c}{}
    & \multicolumn{5}{c}{$\setSize{\policyLan}$} \\
    \multirow{6}{*}{\rotatebox[origin=c]{90}{size of $\cInput$ (\SI{}{\kilo\byte})}}
    & \multicolumn{1}{c}{}
    & \multicolumn{1}{c}{$1e2$}
    & \multicolumn{1}{c}{$5e2$}
    & \multicolumn{1}{c}{$1e3$}
    & \multicolumn{1}{c}{$5e3$} 
    & \multicolumn{1}{c}{$1e4$} \\ \cline{3-7} \\[-8.5pt]
    & $5e0$ &   45.5 & 77.3 & 114.2 & 421.7 & 812.7 \\
    & $1e1$ &  101.1 & 132.0 & 170.0 & 477.2 & 868.3 \\
    & $1e2$ &  658.1 & 683.7 & 714.9 & 1031.5 & 1420.2 \\
    & $2e2$ & 1253.7 & 1282.4 & 1324.5 & 1631.1 & 2026.3 \\
    & $1e3$ & 6056.7 & 6084.5 & 6121.2 & 6428.2 & 6819.4
  \end{tabular}
  \caption{sdhash}
  \label{tbl:explicit-check:sdhash}
\end{subtable}
\caption{\ExplicitCheckPhase response time (s) when instantiated with
  different fuzzy hashes}
\label{tbl:explicit-check}
\end{table}

We also tested using sdhash, which transforms $\cInput$ to a Bloom
filter whose size scales linearly with $\cInput$. This expansion
significantly increases the communication and computation required
during \testAndCommit, ultimately slowing overall performance.  As
such, TLSH outperforms both ssdeep and sdhash in the the overall
\explicitCheckPhase for all parameter settings we tested
(\tblref{tbl:explicit-check:tlsh}
vs.\ \tblref{tbl:explicit-check:ssdeep}
vs.\ \tblref{tbl:explicit-check:sdhash}).  Therefore, we pick TLSH for
the rest of our \primitiveNameShort protocol benchmarking below.

\subsubsection{\ExplicitCheckPhase throughput}
\label{sec:eval:explicit:throughput}

To measure the throughput of the \explicitCheckPhase in our
\primitiveNameShort design, we conducted tests wherein we loaded the
server \server with increasingly many concurrent client requests until
it saturated.  \figref{fig:EC} shows the results of these tests.
The workload on the server is shown on the x-axis, representing the
number of client requests issued concurrently to the server. The
y-axis reports the response time per client request, averaged over all
clients' invocations.  The black solid line in both plots show the
$\client[1]$ response time with input size of \SI{10}{\kilo\byte} and
$\setSize{\policyLan} = 100$.  Response time remains low and grows
slowly with few concurrent client requests, with a single-client
request completing in \SI{19.11}{\second}.  However, as the
concurrency level increases, response time reaches a noticeable
inflection point around 250 simultaneous requests.  At this load, the
average response time spikes, suggesting that the server reached a
saturation threshold beyond which queuing and resource contention
dominate system performance.  The red solid line represents the same
functionality in a monolithic garbled circuit. It takes
\SI{247}{\second} for a single client to finish the task, while the
performance significantly degrades at around 5 clients.

In \figref{fig:EC:blocklist}, we show results of \primitiveNameShort
when varying the size of the blocklist, $\setSize{\policyLan}$. We
observe that the overall throughput degrades more quickly as the
number of entries increases, because a larger blocklist directly
increases the communication cost between $\client[1]$ and $\server$,
leading to greater network overhead and earlier saturation. As a
result, the turning point, where latency begins to rise sharply,
occurs at lower concurrency levels when the $\policyLan$ is
large. When $\setSize{\policyLan} = 1{,}000$, the saturation threshold
is at around 80 clients, while it is at around 20 clients with
$\setSize{\policyLan} = 10{,}000$. In \figref{fig:EC:input}, we vary
the $\client[1]$ input size. In contrast, the impact on throughput is
substantially smaller. Due to the reusable garbled circuit technique
described in \secref{sec:app:embed-and-map:reusable}, the server-side
workload remains independent of the input size, and only the
client-side local circuit computation scales with the input
length. Therefore, while latency does increase gradually with larger
inputs, the shape of the throughput curve remains largely consistent,
and the turning point does not shift significantly.

\begin{figure}[t!]
  \centering

\begin{subfigure}[t]{0.475\columnwidth}
\begin{tikzpicture}[scale=0.85]
\begin{axis}[
legend cell align={left},
legend style={
	fill opacity=0.8,
	draw opacity=1,
	text opacity=1,
	at={(0.03,0.97)},
	anchor=north west,
	draw=lightgray204,
	font=\small,
},
scaled x ticks=manual:{}{\pgfmathparse{#1}},
tick align=outside,
tick pos=left,
x grid style={darkgray176},
xtick style={color=black},
xlabel={clients},
ymax=1100,
y grid style={darkgray176},
ylabel={avg.\ response time (\SI{}{\second})},
ylabel style={yshift=-1ex},
ytick style={color=black},
height=4.5cm,
width=4cm,
]
\addplot [semithick, black]
table {%
1 19.11
3 19.38
5 19.62
10 21.08
20 27.08
40 28.92
50 37.5
60 41.99
70 47.76
80 55.96
100 62.49
150 94.01
200 121.69
250 187.65
300 487.07
400 935.43
};
\addplot [semithick, dashed, black]
table {%
1 75.84
3 81.25
5 83.13
10 92.8
20 94.68
30 98.09
40 109.64
50 129.63
80 285.49
100 695.42
120 1374.44
};
\addplot [semithick, dotted, black]
table {%
1 786.12
3 818.9
5 846.26
10 1058.94
20 1676.73
};
\addplot [semithick, red]
table {%
1 247.04
3 253.02
5 257.2
10 350.3
20 608.78
30 1796.5
};
\end{axis}
\end{tikzpicture}
\caption{$\cInput$ is \SI{10}{\kilo\byte}.
  $\setSize{\policyLan} = 1e2$ \mbox{(---)};
  $1e3$ (--\,--); or $1e4$ ($\cdot$\,$\cdot$).}
 \label{fig:EC:blocklist}
\end{subfigure}%
\hfill
\begin{subfigure}[t]{0.475\columnwidth}
\centering
\begin{tikzpicture}[scale=0.85]
\begin{axis}[
legend cell align={left},
legend style={
	fill opacity=0.8,
	draw opacity=1,
	text opacity=1,
	at={(0.03,0.97)},
	anchor=north west,
	draw=lightgray204,
	font=\small,
},
scaled x ticks=manual:{}{\pgfmathparse{#1}},
tick align=outside,
tick pos=left,
x grid style={darkgray176},
xtick style={color=black},
xlabel={clients},
xticklabels={$0$,$0$, $200$, $400$},
ymax=1100,
y grid style={darkgray176},
ylabel={avg.\ response time (\SI{}{\second})},
ylabel style={yshift=-1ex},
ytick style={color=black},
height=4.5cm,
width=4cm,
]
\addplot [semithick, black]
table {%
1 19.11
3 19.38
5 19.62
10 21.08
20 27.08
40 28.92
50 37.5
60 41.99
70 47.76
80 55.96
100 62.49
150 94.01
200 121.69
250 187.65
300 487.07
400 935.43
};
\addplot [semithick, dashed, black]
table {%
1 109.14
3 109.41
5 109.65
10 111.11
20 117.11
40 118.95
50 127.53
60 132.02
70 137.79
80 145.99
100 152.52
150 184.04
200 211.72
250 277.68
300 577.1
400 1025.46
};
\addplot [semithick, dotted, black]
table {%
1 211.26
3 211.52
5 211.77
10 213.32
20 219.32
40 221.06
50 231.74
60 234.23
70 240
80 247.3
100 254.74
150 286.25
200 313.93
250 379.89
300 679.31
400 1127.67
};
\addplot [semithick, red]
table {%
1 247.04
3 253.02
5 257.2
10 350.3
20 608.78
30 1796.5
};
\end{axis}
\end{tikzpicture}
\caption{$\setSize{\policyLan} = 1e2$.  $\cInput$ is $1e1$ (---);
  $1e2$ (--\,--); or $2e2$ ($\cdot\,\cdot$) \SI{}{\kilo\byte}.}
\label{fig:EC:input}
\end{subfigure}

\caption{\ExplicitCheckPhase throughput.  Red line
  (\textcolor{red}{---}) is the same parameter sizes as the black line
  (---) but with the \explicitCheckPhase functionality implemented in
  a monolithic garbled circuit.}
  \label{fig:EC}
\end{figure}

\subsection{\ImplicitCheckPhase}
\label{sec:eval:implicit}

\addText{
As discussed in \secref{sec:related}, the design of 
\primitiveNameShort is closely related to the functionality 
provided by an asymmetric fuzzy PSI. In particular, a fuzzy PSI 
enables one party to learn only the existence of approximate 
matches between its set and another party’s input under a 
given distance metric.}

\addText{
We compare \implicitCheckPhase to a one-shot fuzzy PSI 
instantiation for two important reasons. First, \implicitCheckPhase 
amortizes the cost of private matching by avoiding recomputing 
proximity against the blocklist for every subsequent check 
where the same client input must be validated multiple times. 
Second, \implicitCheckPhase provides resilience against TOCTOU 
attacks. In a naive fuzzy PSI-based approach, the proximity check 
and the subsequent use of the result are decoupled; unless it is 
performed at every use. Therefore, we use asymmetric fuzzy PSI 
as a natural baseline, as it captures the same level of functionality 
and security of \implicitCheckPhase.}

We compare to \citet{Chakraborti2023:DAPSI} and
\citet{Blass25:fuzzyPSI} since they are also Hamming-distance-based
designs and support sets of asymmetric sizes. The TLSH embedding
vector is of length 35 bytes, so
\citet{Chakraborti2023:DAPSI} outperformed \citet{Blass25:fuzzyPSI} as
its cost is independent of the vector length.  We also report the
result from \citet{Rindal2021:volepsi}, the state-of-the-art
semi-honest exact PSI, which compares against $\blockedElmts$ instead
of the embeddings $\policyLan$.

As $\protocolValidateESP$ requires only a few $\oleFunc$ executions
for reconstructing the token and then performs an exact match over the
token, it significantly outperforms either a fuzzy PSI or an exact PSI
by at least two orders of magnitude in response time and three orders
of magnitude in communication volume.  Note also that the cost of
$\protocolValidateESP$ does not depend on the blocklist size.

\begin{table}[htb]
\centering
\begin{tabular}{@{}l*{2}{S[table-format=1.1e1]}@{}}
\toprule
 & \multicolumn{1}{c}{\parbox[m]{4em}{\centering \textbf{Comm.} (\SI{}{\mega\byte})}} &  \multicolumn{1}{c}{\parbox[m]{5em}{\centering \textbf{Response Time} (\SI{}{\second})}} \\
\midrule
\ImplicitCheckPhase (\protocolValidateESP) & 7.2e-3 & 1.9e-1 \\
\citet{Chakraborti2023:DAPSI} & 3.4e1 & 1.6e2 \\
\citet{Blass25:fuzzyPSI} & 2.0e3 & 2.6e3 \\
\citet{Rindal2021:volepsi} & 6.3e2 & 6.1e1 \\
\bottomrule
\end{tabular}
\caption{\ImplicitCheckPhase cost comparison on $\cInput$ of size
  \SI{100}{\kilo\byte}}
  \label{tab:ic_cost}
\end{table}

\section{Conclusion}

This paper presents \primitiveName, which enables a client to
determine if its item approximately matches any item of a server's
blocklist, without requiring the client to disclose its item to the
server or the server to disclose its blocklist to the client.  If the
client's item passes this check, the server is provided a hiding
and binding token
to that item to store on its allowlist until the blocklist is updated,
so that future clients can more efficiently confirm that the item they
received was checked against the latest blocklist.  We apply
\primitiveNameShort to malware defense, wherein a sending client
performs the (relatively expensive) blocklist check before sending an
executable, permitting receiving clients to perform only the
(relatively inexpensive) allowlist check before using the executable,
enabling them to efficiently avoid TOCTOU attacks.  Through judicious
protocol construction, \primitiveNameShort scales efficiently in the
sizes of both the item to be checked and the blocklist.

\bibliographystyle{IEEEtranSN}
\bibliography{bib/bib.bib}

\appendices

\section{One-Time Pairwise Unpredictable-Permutation}
\label{app:proof_prp}

\begin{theorem*}
	For a random selection of a key $\langle \PRPKeyVec[1], \PRPKeyVec[2]
	\rangle \in (\finiteField\setminus\{0\})^{\numberOfPointsSim} \times
	(\finiteField)^{\numberOfPointsSim}$, the function $\PRP[\langle
	\PRPKeyVec[1], \PRPKeyVec[2]\rangle](\cInputVec) = \PRPKeyVec[1] \hadamardProd
	\cInputVec + \PRPKeyVec[2]$ is a one-time pairwise-unpredictable permutation
	over \metricSpaceESP.
\end{theorem*}

\begin{proof}

	First note that $\PRP$ is a permutation. For a fixed key 
$\PRPKey \assign \langle \PRPKeyVec[1], \PRPKeyVec[2] \rangle$, 
if there are two vectors $\genericVec[1] \assign \langle 
\genericVecComp[1][1], \dots, \genericVecComp[1][\numberOfPointsSim] 
\rangle$ and $\genericVec[2] \assign \langle \genericVecComp[2][1], \dots, 
\genericVecComp[2][\numberOfPointsSim] \rangle$ such that $\genericVec[1] 
\neq \genericVec[2]$ and $\PRP[\genericPRPKey](\genericVec[1]) = 
\PRP[\genericPRPKey](\genericVec[2])$, then there is at least one $\ptIdx 
\in [1, \metricSpaceDimGen]$ where  $\genericVecComp[1][\ptIdx] \neq 
\genericVecComp[2][\ptIdx]$ but $\PRPKeyVec[1][\ptIdx] \cdot 
\genericVecComp[1][\ptIdx] + \PRPKeyVec[2][\ptIdx] = \PRPKeyVec[1][\ptIdx] 
\cdot \genericVecComp[2][\ptIdx] + \PRPKeyVec[2][\ptIdx]$. This leads to a 
contradiction that \finiteField is a field. Also, $\PRP$ is surjective because 
for any vector $\genericVec$, $\PRP[\genericPRPKey]^{-1}(\genericVec) = 
\langle \frac{\genericVecComp[1] - \PRPKeyVec[2][1]}{\PRPKeyVec[1][1]}, 
\dots, \frac{\genericVecComp[\metricSpaceDimGen] - 
\PRPKeyVec[2][\metricSpaceDimGen]}{\PRPKeyVec[1][\metricSpaceDimGen]} 
\rangle \in \metricSpaceGen$. Note, $\PRPKeyVec[1] \neq \langle 0, \dots, 0 
\rangle$.

Now, it is clear that $\prob{\PRP[\langle \PRPKeyVec[1], \PRPKeyVec[2]\rangle](\cInputVec) = \PRPKeyVec[1] \hadamardProd
\cInputVec + \PRPKeyVec[2] = \genericVec} = \frac{1}{\setSize{\finiteField}^{\numberOfPointsSim}}$ due to the random choice 
of $\PRPKeyVec[1]$ and $\PRPKeyVec[2]$ from $\metricSpace$. 
Consider another arbitrary input $\cInputVecAlt \in \metricSpaceESP$ such 
that $\setSize{\cInputVec \Delta  \cInputVecAlt} = \diffInComps$ 
where $\diffInComps \in [1, \numberOfPointsSim]$, i.e., 
 $\setSize{\cInputVec \Delta  \cInputVecAlt}$ is the 
 number of dimensions where \cInputVec and \cInputVecAlt 
 differ. Consider any one such dimension $\rootIdx \in [1, \numberOfPointsSim]$, and 
 two arbitrary vectors $\genericVec, \genericVecAlt \in \finiteField$. Then,

	\begin{align*}
	& \prob{\PRPKeyVec[1][\rootIdx] \cdot \cInputVec[\rootIdx] + \PRPKeyVec[2][\rootIdx] = 
	\genericVecComp[\rootIdx]
	 \land \PRPKeyVec[1][\rootIdx] \cdot \cInputVecAlt[\rootIdx] + \PRPKeyVec[2][\rootIdx] = 
	 \genericVecAltComp[\rootIdx]}\\
	& = \prob{\PRPKeyVec[2][\rootIdx] = 
	\genericVecComp[\rootIdx] - \PRPKeyVec[1][\rootIdx] \cdot \cInputVec[\rootIdx]
	\land \PRPKeyVec[1][\rootIdx] = \frac{\genericVecAltComp[\rootIdx] - 
	\genericVecComp[\rootIdx]}{\cInputVecAlt[\rootIdx] - \cInputVec[\rootIdx]}} \\
	& = \frac{1}{\setSize{\finiteField \setminus 0} \cdot 
\setSize{\finiteField}}
	\end{align*}

The probability 
reflects the fact that it is always possible to produce one pair $\PRPKeyVec[1][\rootIdx], 
\PRPKeyVec[2][\rootIdx] \in \{\finiteField \setminus 0\} \times 
\finiteField$ to satisfy both equations, and since $\PRPKeyVec[1]$ and $\PRPKeyVec[2]$
are selected randomly, any such pair is equally likely. 
Also, for any other dimension $\setIdx \in [1, \numberOfPointsSim], \setIdx \neq \rootIdx$, 
where \cInputVec and \cInputVecAlt differ, $\prob{\PRPKeyVec[1][\rootIdx] \cdot 
\cInputVec[\rootIdx] + \PRPKeyVec[2][\rootIdx] = 
	\genericVecComp[\rootIdx]
	\land \PRPKeyVec[1][\rootIdx] \cdot \cInputVecAlt[\rootIdx] + \PRPKeyVec[2][\rootIdx] = 
	\genericVecAltComp[\rootIdx]}$ is independent of $\prob{\PRPKeyVec[1][\setIdx] \cdot 
	\cInputVec[\setIdx] + \PRPKeyVec[2][\setIdx] = 
	\genericVecComp[\setIdx]
	\land \PRPKeyVec[1][\setIdx] \cdot \cInputVecAlt[\setIdx] + \PRPKeyVec[2][\setIdx] = 
	\genericVecAltComp[\setIdx]} $ because each component of \PRPKeyVec[1]
	and \PRPKeyVec[2] are selected independently. Thus, we have 
	
\begin{align*}
	& \prob{\PRP[\genericPRPKey](\cInputVec) = \genericVec \land 
	\PRP[\genericPRPKey](\cInputVecAlt) = \genericVecAlt} \\
	& = \frac{1}{(\setSize{\finiteField \setminus 
	0} \cdot 
		\setSize{\finiteField})^{\setSize{\cInputVec \Delta \cInputVecAlt}}} \cdot \left(\frac{1}{\setSize{\finiteField}}\right)^{\numberOfPointsSim - \diffInComps} = \frac{1}{\setSize{\finiteField 
		\setminus 
			0}^{\diffInComps} \cdot 
		\setSize{\finiteField}^{\numberOfPointsSim}} \end{align*}

By applying Bayes' theorem, we get the result

\begin{align*}
\cprob{}{\PRP[\genericPRPKey](\cInputVec) = \genericVec}{\PRP[\genericPRPKey](\cInputVecAlt) = \genericVecAlt} = \frac{1}{\setSize{\finiteField 
		\setminus 
			0}^{\diffInComps}}
\end{align*}


\end{proof}

\begin{lemma}[\cite{Kissner2005:PSI}]
	\label{lemma:kissener_uniformly_random}
	Given two polynomials $\genericPoly[1]$ and $\genericPoly[2]$, and two
	uniformly random polynomials $\randomPoly[1]$ and $\randomPoly[2]$ of
	fixed degrees $\degreeOfPoly(\randomPoly[1]) \ge \degreeOfPoly
	(\genericPoly[2])$ and $\degreeOfPoly(\randomPoly[2]) \ge 
	\degreeOfPoly(\genericPoly[1])$, the polynomial \randomPoly[3] satisfying the following 
	is a random polynomial. 
	\[
	\randomPoly[1] \cdot \genericPoly[1] + \randomPoly[2] \cdot \genericPoly[2] 
	= \gcd(\genericPoly[1], \genericPoly[2]) \cdot \randomPoly[3]
	\]
\end{lemma}

\begin{prop}
	\label{prop:symmURFromRand}
	Let $\genericVec[1], \genericVec[2] \in \finiteField^{\numberOfPoints}$.
	Let $\randomPoly[1], \randomPoly[2] \in \polyField$ be random
	polynomials of fixed degrees $\degreeOfPoly(\randomPoly[1]) 
	\ge \degreeOfPoly(\polyFromVec{\genericVec[1]})$ and
	$\degreeOfPoly(\randomPoly[2]) \ge \degreeOfPoly(\polyFromVec
	{\genericVec[2]})$. Then,
	{\small
	\begin{align*}
		\prob{\symmURDist{\genericVec[1]}{\genericVec[2]} \neq
			\setSize{\uniqueRoots{\polyFromVec{\genericVec[1]}}} + 
			\setSize{\uniqueRoots{\polyFromVec{\genericVec[2]}}} -
			2\cdot\setSize{\uniqueRoots{\genericPoly}}}
		& \le \frac{1}{\setSize{\finiteField}}\\
		\mathllap{\mbox{for $\genericPoly = \gcd(\polyFromVec
				{\genericVec[1]}, \randomPoly[1] \polyFromVec{\genericVec[2]})$}} \\
		\prob{\symmURDist{\genericVec[1]}{\genericVec[2]} \neq
			\setSize{\uniqueRoots{\polyFromVec{\genericVec[1]}}} + 
			\setSize{\uniqueRoots{\polyFromVec{\genericVec[2]}}} -
			2\cdot\setSize{\uniqueRoots{\genericPoly}}}
		& \le \frac{1}{\setSize{\finiteField}}\\
		\mathllap{\mbox{for $\genericPoly = \gcd(\polyFromVec
				{\genericVec[1]} \randomPoly[2], \polyFromVec{\genericVec[2]})$}} \\
		\prob{\symmURDist{\genericVec[1]}{\genericVec[2]} \neq
			\setSize{\uniqueRoots{\polyFromVec{\genericVec[1]}}} + 
			\setSize{\uniqueRoots{\polyFromVec{\genericVec[2]}}} -
			2\cdot\setSize{\uniqueRoots{\genericPoly}}}
		& \le \frac{1}{\setSize{\finiteField}} \\
		\mathllap{\mbox{for $\genericPoly = \gcd(\polyFromVec{\genericVec[1]} 
				\randomPoly[2], \randomPoly[1]\polyFromVec{\genericVec[2]})$}} \\
	\end{align*}
	}
\end{prop}

\begin{proof}
The proposition follows from the result that given some fixed
polynomial \genericPoly and a random polynomial
$\randomPoly$, we have $\prob{\gcd(\genericPoly, \randomPoly) \neq 1} \le
1/\setSize{\finiteField}$~\cite{Ghosh2019communication}.

Due to \lemmaref{lemma:kissener_uniformly_random}, we can compute $\symmURDist{\genericVec[1]}{\genericVec[2]}$ from $\randomPoly[1] \cdot \polyFromVec{\genericVec[1]}
+ \randomPoly[2] \cdot \polyFromVec{\genericVec[2]}$, knowing either
$\polyFromVec{\genericVec[1]}$ or $\polyFromVec{\genericVec[2]}$, 
and the degrees of $\polyFromVec{\genericVec[1]}$,
$\polyFromVec{\genericVec[2]}$. This is because
$\gcd(\polyFromVec{\genericVec[1]}, \randomPoly[1] \cdot
\polyFromVec{\genericVec[1]} + \randomPoly[2] \cdot
\polyFromVec{\genericVec[2]}) = \gcd(\polyFromVec{\genericVec[1]}, 
\polyFromVec{\genericVec[2]})$ with overwhelming probability.
\end{proof}

\section{Background}
\label{app:background}

\subsection{CRGC Background}
\label{app:CRGC}

CRGC~\cite{CRGC} compiles a Boolean circuit $\circuit$ with garbler
input $\garblerInput$ into an obfuscated circuit $\circuitAlt$ and
obfuscated input $\garblerInputAlt$ such that $\circuit(\garblerInput, 
\evaluatorInput) = \circuitAlt(\garblerInputAlt, \evaluatorInput)$ for any
evaluator input $\evaluatorInput$. 
\addText{
In other words, $\circuitAlt$ does
not leak garbler's input to the evaluator while preserving the original 
circuit functionality. In this appendix, we summarize the obfuscation
techniques and argue about the input privacy provided by CRGC.}
In short, obfuscation proceeds via bit flipping
and selective gate transformations. Gates that admit indistinguishability
are placed in the reusable section, while any remaining gates can be 
handled by a non-reusable, normal garbled circuit.

\begin{figure}[h]
\centering
\includegraphics[width=0.5\textwidth]{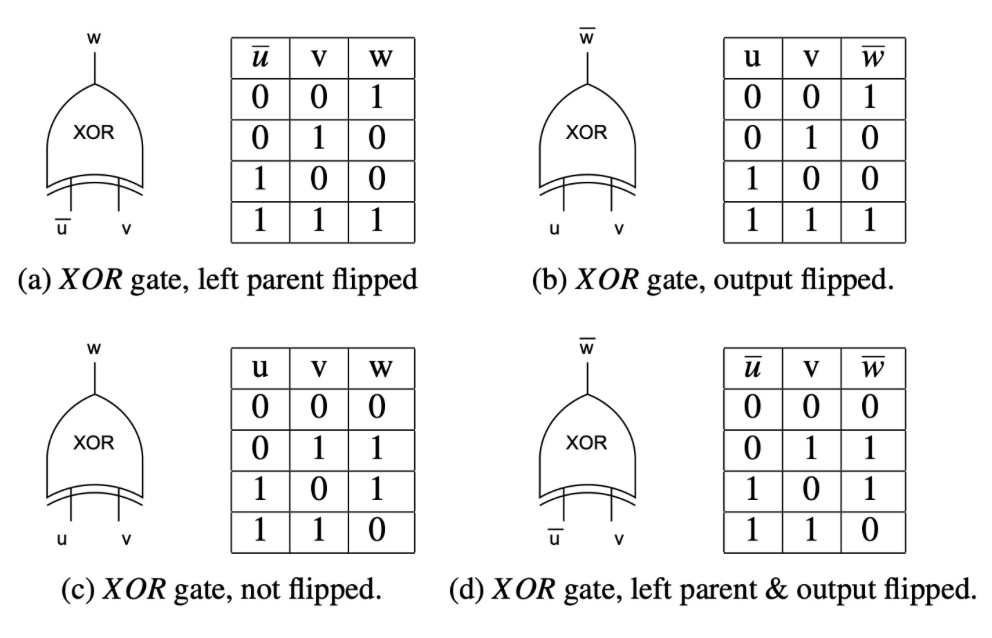}
\caption{Truth Table for XOR Gates}
\label{fig:RGCXOR}
\end{figure}

We start with bit flipping. In CRGC, bit flipping transforms a garbled 
circuit to be reusable, which is the foundation of Lemma 2. Bit flipping 
refers to applying a one-time pad over $\garblerInput$ and all wires 
in the circuit $\circuit$ to obtain $\circuitAlt$ and $\garblerInputAlt$. 
Evaluator input $\evaluatorInput$ and final output wires do not get 
flipped. A flipped wire needs to be modified so that the truth table 
of its child gates maintains the integrity of $\circuit$. 
\figref{fig:RGCXOR} shows such an example of XOR gates. Note that 
the truth table of (a) is identical to (b), but the garbler's input in (a) is 
flipped while the gate output wire in (b) is flipped. The evaluator 
cannot distinguish between the two; thus, $\garblerInput$ is hidden 
from the evaluator. Similarly, the evaluator cannot tell (c) apart from 
(d). Concretely, CRGC Lemma 2 shows that flipping inputs/outputs of 
balanced gates yields pairs of indistinguishable truth tables, so the 
evaluator has advantage 0 in telling the garbler’s original input bit. 
Hence each balanced gate provides indistinguishability obfuscation 
over the garbler's input.

However, bit flipping alone provides privacy only for balanced gates, 
whose truth tables have symmetric output distributions (i.e. have the 
same number of 0's and 1's). Imbalanced gates do not enjoy this 
symmetry: even under input flips, the evaluator can still distinguish 
which inputs yield the gate’s unique outputs. Thus, leakage originating 
from imbalanced gates may propagate backward toward garbler 
inputs. To prevent such leakage, CRGC must ensure that every path 
from a potentially revealing (imbalanced) gate to any garbler input 
contains \emph{at least two balanced gates} through intermediary gate 
transforms. This structural requirement forms the basis of leakage 
prediction in CRGC frameworks: the library analyzes all paths from 
imbalanced gates and classifies a garbler input as achieving 
indistinguishability if and only if every such path contains at least two 
balanced gates. 
\addText{
In \lemmaref{lem:CRGCt}, we argue indistinguishability obfuscation 
of transformed by CRGC bit flipping and intermediary gate obfuscation.}

\addText{
\begin{lemma}
Let $\circuitAlt$ be the transformed circuit obtained from $\circuit$ by 
Bit Flipping together with the fixed-gate and intermediary-gate 
transformations. Suppose a gate $\gate$ lies in the reusable section and, 
after these transformations, $\gateAlt$ is represented by a balanced truth 
table (XOR/XNOR), or more generally by a passive gate (i.e. a gate that does 
not introduce new information about the garbler’s input beyond what is 
already encoded in its input wire labels) modified so as to provide 
indistinguishability obfuscation. Then $\gate$ does not leak the garbler's 
wire value $\wireLabel$ to the evaluator under the standard 
cryptographic assumptions of garbling.
\label{lem:CRGCt}
\end{lemma}
}

\begin{sproof}
\addText{
Bit Flipping applies a one-time-pad-style mask to garbler wires and all 
internal wires, and the paper states that under Bit Flipping all balanced 
gates (XOR/XNOR) achieve indistinguishability obfuscation. In particular, 
different garbler inputs can yield identical balanced-gate truth tables 
after flipping, so the evaluator cannot infer the underlying garbler bit 
by inspection.}

\addText{
The paper’s Lemma 2 formalizes this for balanced reusable gates: 
for XOR/XNOR,
$$\truthTableOut = \wireLabelAlt \oplus \truthTableIn \ or \ \truthTableOut = \neg(\wireLabelAlt \oplus \truthTableIn)$$
with each wire label masked as $\wireLabelAlt = \wireLabel \oplus \wireLabelMask$ 
where $\wireLabelMask \getsr \{0, 1\}$. Since the evaluator does not know 
the hidden masks on garbler-side/internal wires, it cannot distinguish 
the cases $\wireLabel=0$ and $\wireLabel=1$ by inspecting the truth table 
of $\gate$. Hence $\gate$ does not leak $\wireLabel$.}

\addText{
Moreover, the fixed-gate and intermediary-gate transformations are 
specifically used to turn certain level-1 (i.e. gates where one of the 
input wire corresponds to a garbler's input wire) or passive gates into 
balanced/passive forms that provide indistinguishability obfuscation 
without breaking correctness. Therefore, once a gate has been 
transformed into such a reusable balanced/passive form, the same 
indistinguishability argument applies.}
\end{sproof}

To integrate CRGC into \embedAndMap, $\server$ first decomposes 
$\embed$ into reusable and non-reusable components. This decomposition 
is guided by whether a subcircuit already satisfies the CRGC non-leakage 
condition, or can be transformed to do so. To make a circuit compatible with 
CRGC requirements, $\server$ may insert additional XOR layers parameterized 
by the embedding key $\embedKey$. The purpose of these inserted balanced 
gates is to ensure that every path from an output wire back to any garbler input 
wire contains at least two balanced gates, thereby meeting the structural 
privacy guarantees required by CRGC. The rest parts that are incompatible
with such transforms will be executed via a standard Yao's garbled circuit.
The resulting combination of reusable and non-reusable sections is referred to 
as a PRGC. In \lemmaref{lem:CRGCf}, we argue such construction of PRGC 
provides the same level of garbler input privacy as a standard Yao's garbled 
circuit. We prove security of \embedAndMap instantiation in \appref{app:proofs_pake}.

\addText{
\begin{lemma}
Let $\circuitAlt$ be a PRGC obtained from $\circuit$ such that:
\begin{enumerate}
\item
Every gate in the reusable section satisfies the reusable-section privacy condition, 
and in particular every reusable gate is covered by \lemmaref{lem:CRGCt} above.
\item
Every gate not satisfying that condition is placed in a non-reusable section evaluated 
with Yao’s garbled-circuit protocol.
\end{enumerate}
Then $\circuitAlt$ provides the same level of input privacy against a semi-honest evaluator 
as Yao’s garbled-circuit protocol.
\label{lem:CRGCf}
\end{lemma}
}

\begin{sproof}
\addText{
In \citet{CRGC}, the appendix reduces PRGC input privacy to showing that, 
in the evaluator’s view, the wire labels in reusable sections are distributed 
independently of the garbler’s actual input, namely that the relevant wire 
values satisfy equation (3) in \citet{CRGC}. Lemmas 1-3 states that all gates
in a reusable section are handled by the balanced/passive 
indistinguishability argument, so the reusable section contributes no 
information about the garbler input beyond what is already hidden by 
the masked wire-label distribution.}

\addText{
For every gate outside the reusable section, the construction places it in 
a non-reusable section. Lemma 4 then applies: such a section is evaluated 
using OT plus Yao garbling, and therefore does not leak $\wireLabel$ 
under the security assumptions of Yao’s garbled-circuit protocol. The outputs 
of each such non-reusable section are either final outputs, which are allowed 
by functionality, or they re-enter the reusable section at gates that again satisfy 
equation (3).}

\addText{
Therefore, the evaluator’s full view decomposes into:
\begin{itemize}
\item reusable-section views that are indistinguishable under equation (3), and
\item non-reusable-section views protected by Yao security.
\end{itemize}}

\addText{
Composing these two kinds of sections yields exactly the PRGC privacy claim 
stated by the paper: PRGCs do not leak garbler inputs and guarantee the same 
level of input privacy as Yao’s garbled-circuit protocol.}
\end{sproof}

\subsection{Hamming-distance-aware private queries}
We will use the two-party Hamming-distance-aware query protocol by 
\citet{Chakraborti2023:DAPSI}. 
Given two bit vectors $\cInputVec \in
\{0,1\}^{\degreeOfKey{1}}$ and $\cInputVecAlt \in
\{0,1\}^{\degreeOfKey{1}}$, the protocol checks whether
$\hammingDist{\cInputVec}{\cInputVecAlt} \le \distInMaskedSets$. If
$\hammingDist{\cInputVec}{\cInputVecAlt} > \distInMaskedSets$, then
neither of the parties learn any information about $\cInputVec$ and
$\cInputVecAlt$. Here, we will discuss the high level steps, which is
enough to understand its use in a \primitiveNameShort
construction. The protocol has two steps:

\smallskip
{\it \underline{Mapping step}:}
The parties first locally map their respective inputs, $\cInputVec$ and
$\cInputVecAlt$, to the metric space $(\finiteField^{\numberOfPointsSim}, 
\symmURDefn)$ using the injective map $\vecMap:\{0,1\}^{\degreeOfKey{1}} 
\rightarrow \finiteField^{\numberOfPointsSim}$.  $\vecMap$ takes a 
set $\setOfEvalPts$ as a parameter, which both parties first agree 
on. $\genericVec \assign \vecMap(\cInputVec)$ and $\genericVecAlt 
\assign \vecMap(\cInputVecAlt)$ are the mapped embeddings.

\smallskip
{\it \underline{Matching step}:}
To check whether $\hammingDist{\cInputVec}{\cInputVecAlt} \le
\distInMaskedSets$, the parties interactively and privately check
that $\symmURDist{\genericVec}{\genericVecAlt} \le \threshold = 
2\distInMaskedSets$ using 
\eqnsref{eq:hamming_distance_diff}{eq:diff_sur}. 
As pointed out in \citet{Chakraborti2023:DAPSI},
$\hammingDist{\cInputVec}{\cInputVecAlt} \le
\distInMaskedSets$ is equivalent to $\symmURDist{\genericVec}{\genericVecAlt} \le \threshold$
since $\vecMap$ is injective.
For this, the protocol leverages the fact that \symmURDefn can be computed using 
\npaFunc due to \eqnref{eq:sur_random}. 

Crucially, the protocol sets $\numberOfPointsSim \in (\degreeOfKey{1}, 2\degreeOfKey{1} + 1)$ as a
function of $\degreeOfKey{1}$ and \threshold such that if
$\symmURDist{\genericVec}{\genericVecAlt} \le \threshold$, then the server
learns \genericVec  (and $\cInputVec$), otherwise {\it no} information is
revealed about $\cInputVec$ to the server. Intuitively, the idea is that with
less than $\numberOfPointsSim \le 2\degreeOfKey{1} + 1$ points, it is not
possible
to uniquely define and interpolate the polynomial $\randomPoly[2] \cdot
\polyFromEmbedding[\cInputVec] + \randomPoly[1] \cdot
\polyFromEmbedding[\cInputVecAlt]$
and thus $\randomVec[1], \randomVec[2]$,  unless
$\polyFromEmbedding[\cInputVec]$ and $\polyFromEmbedding[\cInputVecAlt]$ share
$2(\degreeOfKey{1} - \threshold)$ roots.

\section{\PrimitiveName for Hamming Distance}
\label{app:protocol}

\subsection{Embed-and-Map}
\label{app:em}

The protocol \protocolEmbed (\figref{fig:esp_em}) realizing \embedAndMap
perform two tasks: embedding and permutation. \server selects keys for both
embedding and permutation, produces a garbled circuit with the keys as its
inputs, and sends the circuit to $\client[1]$ (evaluator). The circuit first 
embeds $\client[1]$'s input into $\metricSpaceESP$ (\stepref{esp:em:embed}).
The embedding has to be executed in the garbled circuit as the malicious client 
is not trusted to honestly embed its inputs.  Otherwise, a malicious $\client[1]$ 
can set the $\falseAcceptRate{\blockedElmts}{\policyLan}{\threshold}$ to 1 through 
iterating through the input universe. 
After that, the circuit computes the one-time pairwise-independent permutation 
\PRP on the mapped input (\stepref{esp:em:map}). In addition, the circuit also
produces a collision-resistant mapping on $\client[1]$'s input, randomized by 
both $\client[1]$ and $\server$. We explain the necessity of the two blinding
factors $\cMask$ and $\sMask$ in the next functionality.

\begin{oframed}
{\small

	\smallskip\noindent\textbf{Parameters:} A keyed one-time pairwise-unpredictable permutation 
	scheme	$\PRP: {\PRPKeySpace}^{\numberOfPointsSim} \times \metricSpaceESP \rightarrow
	\metricSpaceESP$; a non-programmable random oracle $\hash:
	\{0,1\}^{\ast} \rightarrow \metricSpaceESP$; an embedding function 
	$\embed:\embedKeySpace\times \universe
	\rightarrow \metricSpaceESP$.
	
	\smallskip\noindent\textbf{Inputs:} $\client[1]$'s input is
	$\cInput \in \universe$ and $\cMask \in \metricSpace$. \server's input is a $\PRPKey[1],
	\PRPKey[2] \in {\PRPKeySpace}^{\numberOfPointsSim}$, $\embedKey \in \embedKeySpace$, and $\sMask \in \metricSpace$.
	
	\smallskip\noindent\textbf{Protocol:}
	\server generates a garbled circuit and $\client[1]$ executes this circuit
	as the evaluator. The circuit does the following:
	
	\begin{enumerate}
		\item \label{esp:em:embed} Compute $\cInputVec \assign \embedESP(\cInput)$ 
		
		\item \label{esp:em:map} Output $\prpOutputEmbedding \assign
		\PRP[\PRPKey[1]] (\cInputVec)$ and
		$\prpOutputHash \assign
		\PRP[\PRPKey[2]](\sMask \hadamardProd 
		\hash(\cInput, \embedESP(\cInput), \embedKey, \cMask)$ to $\client[1]$.
		\label{em:output}
	\end{enumerate}
	
}
\end{oframed}
\captionof{figure}{\protocolEmbed: Embed-and-Map protocol  \label{fig:esp_em}}

\subsection{Test-and-Commit}
\label{app:tc}
The protocol \protocolCommit (\figref{fig:esp_tc}) realizes \testAndCommit, where
\server first locally computes its input set \policyLan by mapping each element
in \blockedElmts to a vector in $\finiteField^{\numberOfPointsSim}$ with $\embed$. 
Then,  the protocol checks $\symmURDist{\cInputVec}{\lElmtVecComp{\setIdx}} \le
\threshold$ for all $\lElmtVecComp{\setIdx} \in \policyLan$ 
(\stepsref{esp:tc:lan_S}{esp:tc:interpolate}).
\server can learn the result from \npaFunc (\stepref{esp:tc:compute_comb_poly}) 
iff \server is able to reverse $\client[1]$'s input $\prpOutputEmbeddingVec$ with 
$\PRPKey[1]$, which indicates that the output from \protocolEmbed and the 
input to \protocolCommit are consistent. 

During commit phase, $\client[1]$ and \server interactively compute the summation of 
$\commitVec[\setIdx] \gets \serversRandomPolyEvalVec[\setIdx] \hadamardProd
\cInputVec + \clientsRandomPolyEvalVec[\setIdx] \hadamardProd
\lElmtVecComp{\setIdx}$ and $\sMask \hadamardProd \hash(\cInput,
\embed[\embedKey](\cInput), \embedKey, \cMask)$ using a series of calls to \oleFunc.
They start by eliminating $\clientsRandomPolyEvalVec[]$ ($\client[1]$'s input, in 
\stepsref{esp:tc:remove_client_random}{esp:tc:client_comb}) and
$\serversRandomPolyEvalVec[],
\lanElmtVec$ (\server's inputs, in
\stepsref{esp:tc:server_unblind_1}{esp:tc:output}), without
revealing either of $\cInputVec$ and 
$\sMask \hadamardProd \hash(\cInput, \embed[\embedKey](\cInput), \embedKey, \cMask)$ individually to
\server. $\client[1]$ and \server proceed with computing the {\it authentication token}
$\prfOutput \assign \cInputVec
+ \sMask \hadamardProd \hash(\cInput, \embed[\embedKey](\cInput), \embedKey, \cMask)$ (\stepref{esp:tc:output_1}). 
\server locally adds $\langle \embedKey, \sMask, \prfOutput \rangle$ to the stored allowlist from the \explicitCheckPhase if 
$\blocked{\blockedElmts}{\policyLan}{\threshold}(\embedESP(\cInput)) =
\fals$ (\stepref{esp:tc:output}).
$\cMask$ and $\sMask$ therefore guarantee that i) $\cInput$ is hidden 
from \server and \server cannot brute force over the input universe on
the allowlist entries; 
ii) $\client[1]$ cannot predict/generate the authentication token; iii) the 
checked input matches the input stored on the \server allowlist.

\begin{oframed}
	{\small	
		
		\smallskip\noindent\textbf{Parameters:} 
		A set of elements $\setOfEvalPts \assign \{\evalPt[1], \dots,
		\evalPt[\numberOfPointsSim] \} \subset \finiteField$.
		
		\smallskip\noindent\textbf{Inputs:} $\client[1]$'s input is 
		$\prpOutputEmbeddingVec \gets \PRP[\PRPKey[1]] (\cInputVec) = \langle
		\prpOutputEmbeddingComp{1}, \dots, \prpOutputEmbeddingComp{\numberOfPointsSim}
		\rangle$, and $\prpOutputHashVec \gets \PRP[\PRPKey[2]](
		\sMask \hadamardProd \hash(\cInput,
		\embedESP(\cInput), \embedKey)) = \langle \prpOutputHashComp{1},
		\dots,\prpOutputHashComp{\numberOfPointsSim} \rangle$
		obtained from \protocolEmbed. 
		\server's input is a set $\policyLan \assign \{\embed[\embedKey](\lanElmtU)~|~ \lanElmtU 
		\in 
		\blockedElmts \} \subset \metricSpaceESP$  of size
		$\blockListSize$ derived from \blockedElmts;
		a distance threshold $\threshold$; and keys $\PRPKey[1], \PRPKey[2] \in
		{\PRPKeySpace}^{\numberOfPointsSim}$ where $\PRPKey[1] = \langle \PRPKeyEmbeddingVec[1],
		\PRPKeyEmbeddingVec[2] \rangle$ and $\PRPKey[2] = \langle \PRPKeyHashVec[1],
		\PRPKeyHashVec[2] \rangle$.

		{\small \smallskip\noindent\textbf{Test Phase}}
		\begin{enumerate}[nosep, leftmargin=1.7em,labelwidth=*,align=left]
			
			{\small \item \label{esp:tc:lan_S} For $\lElmtVecComp{\setIdx} =  \langle  \lElmtVecComp{\setIdx}{1},
				\dots, \lElmtVecComp{\setIdx}{\numberOfPointsSim}  \rangle  \in \policyLan$,
				\server computes  $\polyFromEmbedding[\lElmtVecComp{\setIdx}]$ by
				interpolating
				the
				set of points $\{ (\evalPt[1], \lElmtVecComp{\setIdx}{1}), \dots,
				(\evalPt[\numberOfPointsSim], \lElmtVecComp{\setIdx}{\numberOfPointsSim} ) \}$.
				
				\item \label{esp:tc:rand_poly_C} $\client[1]$ selects $\blockListSize$ random
				polynomials $\clientsRandomPoly[1], \dots, \clientsRandomPoly[\blockListSize]$
				of degree $\numberOfMasks$, and computes
				$\clientsRandomPolyEvalVecDefn{\setIdx} = \langle
				\clientsRandomPolyEvalVec[\setIdx][1], \dots,
				\clientsRandomPolyEvalVec[\setIdx][\numberOfPointsSim] \rangle$ for $\setIdx
				\in \nats[\blockListSize]$.
				
				\item \label{esp:tc:rand_poly_S} \server selects $\blockListSize$ random
				polynomials $\serversRandomPoly[1], \dots, \serversRandomPoly[\blockListSize]$
				of degree $\numberOfMasks$, and computes
				$\serversRandomPolyEvalVecDefn{\setIdx} = \langle
				\serversRandomPolyEvalVec[\setIdx][1], \dots,
				\serversRandomPolyEvalVec[\setIdx][\numberOfPointsSim] \rangle$ for $\setIdx
				\in \nats[\blockListSize]$.

				\item \label{esp:tc:compute_comb_poly} For $\setIdx \in [1, \blockListSize]$,
				\server and $\client[1]$ call \npaFunc:
				
				\begin{enumerate}[nosep, leftmargin=1em,labelwidth=*,align=left]
					\item $\client[1]$'s inputs are $\clientsRandomPolyEvalVec[\setIdx]$ and
					$\PRP[\PRPKey[1]](\embedESP(\cInput))$.

					\item \server's inputs are $\serversRandomPolyEvalVec[\setIdx]$ and
					$\PRPKeyEmbeddingVec[1]\hadamardProd \lElmtVecComp{\setIdx}$.
					
					\item \server learns 
					\vspace{-0.1cm}
					\begin{align*}
						& \left(\serversRandomPolyEvalVec[\setIdx] \hadamardProd 
						\PRP[\PRPKey[1]](\embedESP(\cInput))\right) + \left(
						\clientsRandomPolyEvalVec[\setIdx] \hadamardProd
						(\PRPKeyEmbeddingVec[1]\hadamardProd \lElmtVecComp{\setIdx})\right)\\
						& = \left(\serversRandomPolyEvalVec[\setIdx] \hadamardProd 
						(\PRPKeyEmbeddingVec[1] \hadamardProd \embedESP(\cInput) +
						\PRPKeyEmbeddingVec[2])\right) + \left(\clientsRandomPolyEvalVec[\setIdx]
						\hadamardProd (\PRPKeyEmbeddingVec[1]\hadamardProd
						\lElmtVecComp{\setIdx})\right)
					\end{align*}
				\end{enumerate}
				
				\item \label{esp:tc:compute_vec} Using $\PRPKeyEmbeddingVec[1]$,
				$\serversRandomPolyEvalVec[\setIdx]$ and $ \PRPKeyEmbeddingVec[2]$,  \server
				computes
				\vspace{-0.1cm}
				\begin{align*}
					\commitVec[\setIdx] & \gets \serversRandomPolyEvalVec[\setIdx] \hadamardProd
					\embedESP(\cInput) + \clientsRandomPolyEvalVec[\setIdx] \hadamardProd
					\lElmtVecComp{\setIdx} \\
					& = \left\langle \serversRandomPoly[\setIdx][\evalPt] \cdot
					\polyFromEmbedding[\cInputVec][\evalPt]  +
					\clientsRandomPoly[\setIdx][\evalPt]
					\cdot \polyFromEmbedding[\lElmtVecComp{\setIdx}][\evalPt] \right\rangle_{\evalPt \in
						\setOfEvalPts}.
				\end{align*}

				\item \label{esp:tc:interpolate} For $\setIdx \in [1, \blockListSize]$, \server tries to compute
				$\gcd(\polyFromEmbedding[\cInputVec],
				\polyFromEmbedding[\lElmtVecComp{\setIdx}])$
				from $\commitVec[\setIdx]$ by interpolating
				the points \cite{Chakraborti2023:DAPSI}. If the interpolation is successful,
				\server sets
				$\blocked{\policyLan}{\threshold}(\cInputVec) = \tru$. Otherwise, it sets
				$\blocked{\policyLan}{\threshold}(\cInputVec) = \fals$.}

%
		\end{enumerate}
		
		\smallskip\noindent\textbf{Commit Phase}
		
		\begin{enumerate}[nosep, leftmargin=1.7em,labelwidth=*,align=left,start=7]
			
			{\small \item \server computes $\randomPolyComb \assign \sum\limits_{\setIdx
					\in \nats[\blockListSize]} \serversRandomPoly[\setIdx]$, and
				$\randomPolyCombEvalVec \gets \langle \randomPolyComb[\evalPt]
				\rangle_{\evalPt
					\in \setOfEvalPts}$.
				

				\item \label{esp:tc:remove_client_random} \server samples
				$\tmpVarAltVec[\setIdx] = \langle \tmpVarAltVec[\setIdx][1], \dots,
				\tmpVarAltVec[\setIdx][\numberOfPointsSim] \rangle \getsr \metricSpaceESP$.
				For $\setIdx \in \nats[\blockListSize], \ptIdx \in \nats[\numberOfPointsSim]$,
				$\client[1]$ and \server call \oleFunc where $\client[1]$'s input is $-
				\clientsRandomPolyEvalVec[\setIdx][\ptIdx]$ and \server's inputs are
				$\lElmtVecComp{\setIdx}{\ptIdx}$
				and $\tmpVarAltVec[\setIdx][\ptIdx]$. After all the calls,  $\client[1]$ learns
				$\tmpVarAltVec[\setIdx] - \clientsRandomPolyEvalVec[\setIdx] \hadamardProd
				\lElmtVecComp{\setIdx}$ for $\setIdx \in \nats[\blockListSize]$.

				\item \label{esp:tc:prp_inv_hash}  \server samples $\tmpVarVec[\setIdx] =
				\langle \tmpVarVec[\setIdx][1], \dots, \tmpVarVec[\setIdx][\numberOfPointsSim]
				\rangle \getsr \metricSpaceESP$. For
				$\setIdx \in \nats[\blockListSize], \ptIdx \in \nats[\numberOfPointsSim]$,
				\server and $\client[1]$
				call \oleFunc, where $\client[1]$'s input is $\prpOutputHashComp{\ptIdx}$, and
				\server's inputs are
				$\frac{\randomPolyComb[\evalPt[\ptIdx]]}{\PRPKeyEmbeddingVec[1][\ptIdx]'}$.
				After all calls, $\client[1]$ learns 
				$(\PRPKeyHashVec[1])^{-1} \hadamardProd \randomPolyCombEvalVec \hadamardProd
				\prpOutputHashVec + \tmpVarVec[\setIdx]$ for $\setIdx \in
				\nats[\blockListSize]$.
				
				\item \label{esp:tc:client_comb} From the step above, $\client[1]$ computes and
				sends to \server
				
				\[
				\genericVec' \gets \sum\limits_{\setIdx \in \nats[\blockListSize]}
				\tmpVarAltVec[\setIdx] + \tmpVarVec[\setIdx]-
				\clientsRandomPolyEvalVec[\setIdx] \hadamardProd \lElmtVecComp{\setIdx} +
				(\PRPKeyHashVec[1])^{-1} \hadamardProd \randomPolyCombEvalVec \hadamardProd
				\prpOutputHashVec.
				\]}

			\item \label{esp:tc:server_unblind_1} \server computes $\genericVec
			\gets (\sum\limits_{\setIdx \in \nats[\blockListSize]} \tmpVarAltVec[\setIdx]
			+ \tmpVarVec[\setIdx]) + \blockListSize \cdot ((\PRPKeyHashVec[1])^{-1}
			\hadamardProd \randomPolyCombEvalVec \hadamardProd \PRPKeyHashVec[2])$.

			\item \label{esp:tc:output_1} \server computes 
			
			\vspace{-0.5cm}
			\begin{align*}
				& \authToken \assign \frac{1}{\blockListSize} \cdot
				\randomPolyCombEvalVec^{-1}
				\hadamardProd \left(\sum\limits_{\setIdx = 1}^{\blockListSize}
				\combVector_\setIdx
				+ \genericVec' - \genericVec \right) \\
				& =  \frac{1}{\blockListSize} \cdot \randomPolyCombEvalVec^{-1}
				\hadamardProd \\
				& \left(\blockListSize \cdot \randomPolyCombEvalVec \hadamardProd
				(\embedESP(\cInput) + (\PRPKeyHashVec[1])^{-1} \hadamardProd \prpOutputHashVec -
				(\PRPKeyHashVec[1])^{-1} \hadamardProd \PRPKeyHashVec[2])\right) \\
				& = \embedESP(\cInput) + \sMask \hadamardProd
				\hash(\cInput, \embedESP(\cInput), \embedKey, \cMask).
			\end{align*}
			
			\item \label{esp:tc:output} If $\blocked{\policyLan}{\threshold}
			(\cInputVec) = \fals$, \server stores $\langle \nonce, \prfOutput, \embedKey, 
			\sMask \rangle$ and outputs $1$ and $\nonce \gets \{0,1\}^{\ast}$ to $\client[1]$. $\client[1]$ keeps
			$\nonce$, $\cInput$, and $\cMask$. Otherwise, \server aborts
			and $\client[1]$ receives a $0$. 
			
		\end{enumerate}
	}
\end{oframed}
\captionof{figure}{\protocolCommitESP: Test-and-Commit protocol}\label{fig:esp_tc}

\subsection{\ImplicitCheckPhaseHeading}
\label{app:ic}
The protocol \protocolValidate (\figref{fig:esp_tc}) realizes \authentication. 
$\client[2]$'s inputs are the re-created outputs of \protocolEmbed, without $\PRP$, with 
\embedKey obtained from \server (\figref{fig:frameworkSH}). 
\server may reveal $\embedKey$ to $\client[2]$ because $\embedKey$ is randomly 
chosen for each $\client[1]$ and an authentication token $\prfOutput$
for that $\client[2]$ will only be stored at \server if \explicitCheckPhase goes through.
Then, $\client[2]$ and \server evaluate \oleFunc with the reconstructed inputs
 and \server's blinding factor $\sMask$ (\stepref{esp:auth:ole}). $\server$ will obtain
a token $\prfOutputAlt$ to be compared to the allowlist (\stepref{esp:auth:prf}).
Same as $\prfOutput$, $\prfOutputAlt$ hides $\cInputAlt$ from \server.
$\client[2]$ learns only a 0/1 output, indicating whether the \implicitCheckPhase
is successful or not.

\begin{oframed}
{\small
		
		\smallskip\noindent\textbf{Inputs:} $\client[2]$'s input is $\cICInput[1], \cICInput[2] \in \metricSpace$.
		\server's input is $\sMask, \prfOutput \in \metricSpace$.
		
		\smallskip\noindent\textbf{Protocol}
		\begin{enumerate}[nosep, leftmargin=1.5em,labelwidth=*,align=left]
			
			\item \label{esp:auth:startSS} $\client[2]$ computes 
			$\langle \cICInputOneComp{1}, \dots, \cICInputOneComp{\numberOfPointsSim}
			\rangle \assign \cICInput[1]$
			and
			$\langle \cICInputTwoComp{1}, \dots, \cICInputTwoComp{\numberOfPointsSim}
			\rangle \assign \cICInput[2]$.

			\item \label{esp:auth:serverSS} $\server$ computes 
			$\langle \sMaskComp{1}, \dots, \sMaskComp{\numberOfPointsSim}
			\rangle \assign \sMask$.
			
			\item \label{esp:auth:ole} For $\ptIdx \in [1,
			\numberOfPointsSim]$, $\client[2]$ and \server call \oleFunc.
			
			\begin{enumerate}
				\item $\client[2]$'s input are $\cICInputOneComp{\ptIdx}$ and 
				$\cICInputTwoComp{\ptIdx}$.
				\item \server's input are $\sMaskComp{\ptIdx}$.
				\item After all calls, $\server$ learns $\prfOutputAlt \assign \cICInput[1] + \sMask \hadamardProd 
				\cICInput[2]$.
			\end{enumerate}
			\item \label{esp:auth:prf} $\server$ checks whether $\prfOutputAlt \equals \prfOutput$.
			If so, $\server$ outputs $1$ to $\client[2]$; otherwise, $\client[2]$ receives a $0$. 
		\end{enumerate}

}
\end{oframed}
	\captionof{figure}{\protocolValidateESP: \ImplicitCheckPhase protocol \label{fig:esp_authenticate}}
\label{fig:protocolIC}	

\section{Proof of Framework Security}
\label{app:proof_framework}

\addText{
In this appendix, we formally prove the framework of 
\figref{fig:frameworkSH} securely realizes the \primitiveNameShort 
design against the adversaries described in 
\secref{sec:framework:threat-model}, under the assumption that the 
constituent functionalities $\embedAndMap, \testAndCommit, 
\authentication$ are secure.  (The proofs of security for these 
functionalities can be found in \appref{app:proofs_pake}.)
A real/ideal proof is compatible in principle; we chose a game-based 
formulation~\cite{BellareRogaway} for clarity and modularity in the 
current presentation, as it lets us cleanly state each adversary’s goals, 
and isolate which primitive properties imply which system-level 
guarantees.
}

We reduce the security of our framework to typical
properties of types of function families, which we recount below.

\begin{defn}[One-time pairwise unpredictability]
  \label{dfn:prp}
  Let $\genericFnFamily: \genericFnKeyspace \times \genericFnDomain
  \rightarrow \genericFnDomain$ be a family of permutations, and
  let \Adv{} be an algorithm that returns a pair of elements from
  $\genericFnDomain$.  Let
  
  \smallskip
  \indent\begin{minipage}[t]{0.5\textwidth}
  \begin{tabbing}
    ****\=****\=\kill
    Experiment $\ExptPRP{\genericFnFamily}(\Adv{})$ \\
    \> $\genericFnKey \getsr \genericFnKeyspace$ \\
    \> $\cInput \leftarrow \Adv{}()$\\
    \> $\genericFieldElmt \gets \genericFnFamily[\genericFnKey](\cInput)$ \\
    \> $(\cInputAlt, \genericFieldElmtAlt) \leftarrow \Adv{}(\genericFieldElmt)$ \\
    \> if $\genericFnFamily[\genericFnKey](\cInputAlt) = \genericFieldElmtAlt$ \\
    \> \> then return $1$ \\
    \> \> else return $0$
  \end{tabbing}
  \end{minipage}
  
   \smallskip
  \noindent Then,
  \begin{align*}
    \Advantage{\PRPsecdef}{\genericFnFamily}{\Adv{}} 
   & = \prob{\ExptPRP(\Adv{}) = 1} \\
    \Advantage{\PRPsecdef}{\genericFnFamily}{\timeBound} 
    & = \max_{\Adv{}} \Advantage
    {\PRPsecdef}{\genericFnFamily}{\Adv{}}
  \end{align*}
  where the maximum is taken over all algorithms \Adv{} that run in
  time at most \timeBound.
\end{defn}

\subsection{Security against \client}
\label{sec:proof_framework:client} 

\subsubsection{In threat model \ref{model:client-mal}}

Assuming \embedAndMap, \testAndCommit, and \authentication 
are secure, then no transcripts are leaked to $\client[1]$ and the 
only information $\client[1]$ learns are the outputs. Since 
\policyLan is not involved in the computation of any outputs, it is 
trivial that no information about \blockedElmts will be leaked 
to $\client[1]$ than the decision ($\cInput \in \blockedElmts$ or not)
implies in one execution of the \explicitCheckPhase phase.

In this threat model, $\client[1]$ is malicious and aims to have a
token stored at \server on some $\cInput \in \blockedElmts$
during \explicitCheckPhase and verify at \implicitCheckPhase. 
As such, the goal of such a $\client[1]$ is to cause the following experiment
to return $1$, where $\client[1]$ is denoted as a triple of algorithms
$(\Adv{1}, \Adv{2}, \Adv{3})$.  For each call to the functionalities,
each input (output) pairs will be denoted in parenthesis, with the
$\client[1]$ input (output) as the first item and the \server input
(output) as the second item. 

\smallskip

\begin{minipage}[t]{\columnwidth}
	\begin{tabbing}
		****\=****\=\kill
		Experiment $\ExptCC[\PRP, \hash]
		(\Adv{1}, \Adv{2}, \Adv{3})$ \\
		\> $\langle \Advstate{1}, \cInput, \cMask \rangle \gets \Adv{1}^{\hash(\cdot)}()$ \\
		\> $\PRPKey[1] \getsr {\PRPKeySpace}^{\numberOfPointsSim}$,
		$\PRPKey[2] \getsr {\PRPKeySpace}^{\numberOfPointsSim}$, \\
		\> $\embedKey \getsr \embedKeySpace$,
		$\sMask \getsr \metricSpace$ \\
		\> $((\prpOutputEmbedding, \prpOutputHash), \cdot) \gets \embedAndMap^{\hash(\cdot)}((\cInput, \cMask), (\embedKey, \PRPKey[1], \PRPKey[2], \sMask))$ \\
		\> $\langle \Advstate{2}, \prpOutputEmbeddingAlt, \prpOutputHashAlt \rangle  \gets \Adv{2}^{\hash(\cdot)}(\Advstate{1}, \prpOutputEmbedding, \prpOutputHash)$ \\
		\> $(\bit, \prfOutput) \gets \testAndCommit^{\hash(\cdot)}((\prpOutputEmbeddingAlt, \prpOutputHashAlt), (\PRPKey[1], \PRPKey[2]))$ \\
		\> if $\prfOutput = \bot$ \\
		\> \> then return 0 \\	
		\> $\langle \cInputAlt, \cMaskAlt \rangle \gets 
		\Adv{3}^{\hash(\cdot), \authentication^{\hash(\cdot)}(\cdot, \sMask)}
		(\Advstate{2}, \embedKey)$ \\ \\
		\> $\prfOutputAlt \gets \embed[\embedKey](\cInputAlt) \linearComb \sMask \hadamardProd \hash(\cInputAlt, \embed[\embedKey](\cInputAlt), \embedKey, \cMaskAlt)$ \\
		\> if \= $\cInputAlt \in \blockedElmts \wedge \prfOutputAlt = \prfOutput$ \\
		****\=****\=\kill
		\> \> then return 1 \\ 
		\> \> else return 0
	\end{tabbing}
\end{minipage}

\noindent Then,
\begin{align*}
	\Advantage{\DCOPRFsecdef}{\PRP, \hash}{\Adv{1}, \Adv{2}, \Adv{3}}
	& = \prob{\ExptCC[\PRP, \hash]
		(\Adv{1}, \Adv{2}, \Adv{3}) = 1} \\
	\Advantage{\DCOPRFsecdef}{\PRP, \hash}{\timeBound, \hashOracleQueries, \forwardOracleQueries}
	& = \max_{\Adv{1}, \Adv{2}, \Adv{3}} \Advantage{\DCOPRFsecdef}
	{\PRP, \hash}{\Adv{1}, \Adv{2}, \Adv{3}}
\end{align*}
where the maximum is taken over all adversaries $(\Adv{1},
\Adv{2}, \Adv{3})$ running in total time \timeBound, with
at most \hashOracleQueries queries to $\hash$, 
and making at most \forwardOracleQueries oracle queries to
$\authentication(\cdot, \sMask)$ from \Adv{3}.
\bigskip

The key distinctions between the prescribed algorithm in
\figref{fig:frameworkSH} and \ExptCC[\PRP, \hash] are that 
(i) the (malicious) $\client[1]$ is permitted to \textit{compute} 
$\cInputAlt, \cMaskAlt$ (in \Adv{3}) and $\prpOutputEmbeddingAlt, 
\prpOutputHashAlt$ (in \Adv{2}), whereas a correct $\client[1]$ simply sets 
$\cInputAlt \gets \cInput$, $\cMaskAlt \gets \cMask$, $\prpOutputEmbeddingAlt \gets 
\prpOutputEmbedding$, and $\prpOutputHashAlt \gets 
\prpOutputHash$; and (ii) the malicious $\client[1]$ is granted oracle access to 
$\authentication(\cdot,\sMask)$ to model its repeated attempts at
\implicitCheckPhase. Note that we omit $\nonce$ in the experiments.
An adversarial $\client[1]$ may attempt to impersonate an honest 
$\client[2]$ by running the 
\implicitCheckPhase with $\server$ on 
multiple guessed nonces before relaying any inputs to the real $\client[2]$.
\addText{
It may also collude with multiple $\client[2]$ to obtain information
stored for other inputs.}
However, both $\embedKey$ and $\sMask$ are freshly and independently 
sampled in every \explicitCheckPhase. Consequently, varying the nonce 
$\nonce$ does not provide $\client[1]$ with any additional leverage to 
probe other entries in the allowlist. The only substantive attack surface 
is therefore the ability to guess an unknown authentication token 
$\prfOutput$ by issuing queries to the non-programmable random 
oracle $\hash$. We show that the advantage of a malicious $\client[1]$ 
in compromising a single \primitiveNameShort session is negligible; hence, 
the probability of successfully attacking any aggregate of sessions is also 
negligible.

\begin{prop}
	\label{prop:framework}
	If $\hash$ is a random oracle, and $\embedAndMap, \testAndCommit, 
   	\authentication$ are secure and output correctly, then
	\begin{align*}
		& \Advantage{\DCOPRFsecdef}{\PRP, \hash}{\timeBound, \hashOracleQueries, \forwardOracleQueries} \\
		& \le \falseAcceptRate{\blockedElmts}{\policyLan}{\threshold}
		+ \frac{2\cdot\setSize{\finiteField}}{\setSize{\PRPKeySpace}}
		+ \frac{\hashOracleQueries+1}{\setSize{\metricSpace}}	
	\end{align*}
	for $\timeBoundAlt = \timeBound + \bigO{1}$.
\end{prop}

\begin{proof}
	
	Let $(\Adv{1}, \Adv{2}, \Adv{3})$ be a
	\primitiveNameShort adversary that runs in time \timeBound and makes
	at most \hashOracleQueries queries to $\hash$ oracle and 
	\forwardOracleQueries queries to $\authentication$ oracles.
	
	There are 2 scenarios for $\cInputAlt$ from $\Adv{3}$: (i) $\cInputAlt
 	= \cInput$, which means that (malicious) $\client[1]$ has a valid authentication 
	token for a blocked input stored at \server; (ii) $\cInputAlt \ne 
	\cInput$, which indicates that $\client[1]$ is submitting inconsistent
	inputs to \explicitCheckPhase and \implicitCheckPhase.
 
	We start with (i) $\cInputAlt = \cInput$. First note that
	{\small
	\begin{align*}
		\mathrlap{\cprob{\Bigg}{\ExptCC[\PRP, \hash](\Adv{1}, \Adv{2}, \Adv{3}) = 1}{\cInputAlt = \cInput}} \\
		& = \cprob{\Bigg}{\ExptCC[\PRP, \hash](\Adv{1}, \Adv{2}, \Adv{3}) = 1}{\begin{array}{@{}r@{}} \neg\blocked{\policyLan}{\threshold}(\embed[\embedKey](\cInput))\\ \wedge~\cInput\in\blockedElmts\end{array}} \\
		& \hspace{2em} \times \cprob{\big}{\neg\blocked{\policyLan}{\threshold}(\embed[\embedKey](\cInput))}{\cInput\in\blockedElmts} \times \prob{\cInput\in\blockedElmts}\\
		& \hspace{1em} + \cprob{\Bigg}{\ExptCC[\PRP, \hash](\Adv{1}, \Adv{2}, \Adv{3}) = 1}{\begin{array}{@{}r@{}} \blocked{\policyLan}{\threshold}(\embed[\embedKey](\cInput))\\ \wedge~\cInput\in\blockedElmts\end{array}} \\
		& \hspace{2em} \times \cprob{\big}{\blocked{\policyLan}{\threshold}(\embed[\embedKey](\cInput))}{\cInput\in\blockedElmts} \times \prob{\cInput\in\blockedElmts}\\
		& \le \cprob{\big}{\neg\blocked{\policyLan}{\threshold}(\embed[\embedKey](\cInput))}{\cInput\in\blockedElmts} \\
		& \hspace{1em} + \cprob{\Bigg}{\ExptCC[\PRP, \hash](\Adv{1}, \Adv{2}, \Adv{3}) = 1}{\begin{array}{@{}r@{}} \blocked{\policyLan}{\threshold}(\embed[\embedKey](\cInput))\\ \wedge~\cInput\in\blockedElmts\end{array}}
	\end{align*}
	}

The first term is just $\falseAcceptRate{\blockedElmts}{\policyLan}
{\threshold}$, and so for the rest of the proof, we focus on bounding 
the second term. Recall that $\ExptCC[\PRP, \hash] (\Adv{1}, 
\Adv{2}, \Adv{3}) = 1$ implies $\neg\blocked{\policyLan}{\threshold}
(\PRP[\PRPKey[1]]^{-1}(\prpOutputEmbeddingAlt))$, which implies 
$\PRP[\PRPKey[1]]^{-1}(\prpOutputEmbeddingAlt) \neq 
\embed[\embedKey](\cInput)$ because $\blocked{\policyLan}
{\threshold}(\embedESP(\cInput))$. $\Adv{2}$ wins by producing
$\prpOutputEmbeddingAlt$ and $\prpOutputHashAlt$ such that 
$\PRP[\PRPKey[1]]^{-1}(\prpOutputEmbeddingAlt) \linearComb
\PRP[\PRPKey[2]]^{-1}(\prpOutputHashAlt) = \embed[\embedKey]
(\cInput) \linearComb \sMask 
\hadamardProd \hash(\cInput, \embed[\embedKey]
(\cInput), \embedKey) \wedge \prpOutputEmbeddingAlt \ne
\prpOutputEmbedding \wedge \prpOutputHashAlt \ne \prpOutputHash$.
Since $\embedKey$ is kept secret from $\client[1]$ during the
\explicitCheckPhase, $\client[1]$ has to produce 
$\prpOutputEmbeddingAlt$ and $\prpOutputHashAlt$ by querying $\PRP$. 

First, note that $\Adv{1}$ gets one invocation to $\PRP[\PRPKey[1]]
(\cdot)$ and one invocation to $\PRP[\PRPKey[2]](\cdot)$ (both via 
\embedAndMap) to produce $\prpOutputEmbedding$ and 
$\prpOutputHash$, respectively, though $\Adv{2}$ generates 
$\prpOutputEmbeddingAlt$ and $\prpOutputHashAlt$ that are 
passed to $\PRP[\PRPKey[1]]^{-1}(\cdot)$ and 
$\PRP[\PRPKey[2]]^{-1}(\cdot)$. Thus, \Adv{2}'s advantage in 
producing $\PRP[\PRPKey[1]]^{-1}(\prpOutputEmbeddingAlt)$ from
$\prpOutputEmbedding, \prpOutputHash$ is bounded by
$\Advantage{\PRPsecdef}{\PRP}{\timeBoundAlt} = \frac{2\cdot\setSize{\genericFnRange}}{\setSize
{\PRPKeySpace}}$ since $\PRP$ is an unpredictable function. The keys 
\PRPKey[1] and \PRPKey[2] are chosen independently, and if 
$\PRPKey[2] \neq \PRPKey[1]$, \prpOutputHash does not provide any 
information about $\PRP[\PRPKey[1]]^{-1}(\prpOutputEmbeddingAlt)$. 
Due to the random key selection, the probability that $\PRPKey[2] = 
\PRPKey[1]$ is $\frac{\setSize{\finiteField}}{\setSize{\PRPKeySpace}}$.

We now consider the case where $\cInputAlt \neq \cInput$. In this scenario, 
a malicious $\client[1]$ first stores an authentication token $\prfOutput$ 
corresponding to a benign input $\cInput$ at the $\server$, and later attempts 
to find a colliding preimage $\cInputAlt$, where $\cInputAlt \in \blockedElmts$. 
More precisely, the security requirement is that 
$\prfOutput$ is binding in the sense that it is computationally infeasible
to come up with a pair $\cInputAlt, \cMaskAlt$ such that $\cInputAlt \ne
\cInput \wedge \cMaskAlt \ne \cMask \wedge  \embed[\embedKey](\cInputAlt) 
\linearComb \sMask \hadamardProd \hash(\cInputAlt, \embed[\embedKey]
(\cInputAlt), \embedKey, \cMaskAlt) = \prfOutput$.
That is, no efficient adversary should be able to produce 
a different input–mask pair that maps to the same authentication token.
There are two cases:
\begin{enumerate}

\item When $\embed[\embedKey](\cInputAlt) \ne \embed[\embedKey](\cInput)$, 
the probability is $\frac{1}{\setSize{\metricSpace}}$ since $\prfOutput$ is not
known and producing such a pair of $\cInputAlt, \cMaskAlt$ is equivalent to 
guessing $\prfOutput$ over the space $\metricSpace$. The probability does
not enhance with the number of oracle queries given to either $\hash$ or 
$\authentication$.

\item When $\embed[\embedKey](\cInputAlt) = \embed[\embedKey](\cInput)$, 
it means that there is a collision over the embedding function. As $\embedKey$
is given to an adversarial $\client[1]$, it can enumerate over $\universe$ to
find such a collision with probability 1. Then finding a collision over $\prfOutput$
is equivalent to finding the pair $\cInputAlt, \cMaskAlt$ such that $\hash(\cInputAlt, 
\embed[\embedKey](\cInputAlt), \embedKey, \cMaskAlt) = \hash(\cInput, \embed
[\embedKey](\cInput), \embedKey, \cMask)$. As $\hash$ is a random oracle,
such a probability would be $\frac{\hashOracleQueries}{\setSize{\metricSpace}}$.

\end{enumerate}

Summing up all cases above, gives the result: 

\begin{align*}
	\mathrlap{\Advantage{\DCOPRFsecdef}{\PRP, \hash}{\Adv{1}, \Adv{2}, \Adv{3}}}\\
	& = \cprob{\Bigg}{\ExptCC[\PRP, \hash](\Adv{1}, \Adv{2}, \Adv{3}) = 1}{\cInputAlt = \cInput}\\
	& + \cprob{\Bigg}{\ExptCC[\PRP, \hash](\Adv{1}, \Adv{2}, \Adv{3}) = 1}{\cInputAlt \ne \cInput}\\
	& \le \falseAcceptRate{\blockedElmts}{\policyLan}{\threshold}
		+ \frac{2\cdot\setSize{\finiteField}}{\setSize{\PRPKeySpace}}		
		+ \frac{\hashOracleQueries+1}{\setSize{\metricSpace}}
\end{align*}

\end{proof}

\subsection{Security against \server}
\label{sec:proof_framework:server}

We prove security for threat model \ref{model:server-hbc} in this
section. Assuming \embedAndMap, \testAndCommit, and 
\authentication are secure, the method available to \server to 
attack the framework in \figref{fig:frameworkSH} is to guess the 
input $\cInput$ of $\client[1]$. As such, the goal of such a \server is to 
cause the following experiment to return $1$, where \server is 
denoted as algorithm $(\AdvS{1}, \AdvS{2})$. In short,
we require that \server cannot distinguish $\client[1]$'s input
$\cInput$ when $\cInput$ is not blocked by the framework.

\smallskip

\begin{minipage}[t]{\columnwidth}
	\begin{tabbing}
	****\=\kill
	Experiment $\ExptS[\hash](\AdvS{1}, \AdvS{2})$ \\
	\> $\langle \Advstate{1}, \cInput[0], \cInput[1], \sMask, \PRPKey[1], \PRPKey[2], \embedKey \rangle \gets \AdvS{1}^{\hash(\cdot)}()$\\
	\> if \= $\blocked{\policyLan}{\threshold}(\embed[\embedKey](\cInput[0]))
		\vee \blocked{\policyLan}{\threshold}(\embed[\embedKey](\cInput[1]))$ \\
	\> \> return 0 \\
	\> $\cMask \getsr \metricSpace$, 
	     $\bit \getsr \bitRange$ \\
	\> $((\prpOutputEmbedding, \prpOutputHash), \cdot) \gets \embedAndMap^{\hash(\cdot)}((\cInput[\bit], \cMask), (\embedKey, \PRPKey[1], \PRPKey[2], \sMask))$ \\
	\> $(1, \prfOutput) \gets \testAndCommit^{\hash(\cdot)}((\prpOutputEmbedding, \prpOutputHash), (\PRPKey[1], \PRPKey[2]))$ \\
	\> $\cICInput[1] \gets \embed[\embedKey](\cInput[\bit]), \cICInput[2] \gets \hash(\cInput[\bit], \embed[\embedKey](\cInput[\bit]), \embedKey, \cMask)$ \\
	\> $\bitAlt \gets \AdvS{2}^{\hash(\cdot), \authentication((\cICInput[1], \cICInput[2]), \cdot)}(\Advstate{1}, \prfOutput)$ \\
	\> if \= $\bitAlt = \bit$ \\
	\> \> then return 1 \\
        \> \> else return 0
	\end{tabbing}
\end{minipage}

\smallskip

\noindent Then,
\begin{align*}
	\Advantage{\Ssecdef}{}{\AdvS{1}, \AdvS{2}}
	& = \prob{\ExptS[\hash](\AdvS{1}, \AdvS{2}) = 1} \\
	\Advantage{\Ssecdef}{}{\timeBound, 
	\forwardOracleQueries}
	& = \max_{\AdvS{1}, \AdvS{2}} \Advantage{\Ssecdef}
	{}{\AdvS{1}, \AdvS{2}}
\end{align*}
where the maximum is taken over all adversaries $(\AdvS{1}, \AdvS{2})$
running in total time \timeBound, with at most \hashOracleQueries
queries to $\hash$, and at most \forwardOracleQueries oracle queries
to $\authentication((\cICInput[1], \cICInput[2]), \cdot)$ from
\AdvS{2}.

\bigskip

To win the experiment $\ExptS[\hash]$, $\AdvS{2}$ needs to
successfully guess $\bit$, meaning that it distinguishes the input
from the protocol transcripts. $\AdvS{1}$ selects $\cInput[0]$ 
and $\cInput[1]$ from \universe according to an application-specific
distribution that need not be uniform. But both have to be 
not blocked since the functionality we design implies that the server 
learns $\{\lanElmt \in \blockList: \distance\left(\embed[\embedKey]
(\cInput)\right) \le \threshold\}$ when $\blocked{\policyLan}{\threshold}
(\embed[\embedKey](\cInput))$. 

\begin{prop}
\label{prop:serverAdv}
If $\hash$ is a random oracle and $\embedAndMap, \testAndCommit,
\authentication$ are secure and output correctly, then
\begin{align*}
\Advantage{\Ssecdef}{\hash}{\timeBound, \hashOracleQueries, 
\forwardOracleQueries}
& \le \frac{1}{2} + \frac{\hashOracleQueries}{\setSize{\metricSpace}}
\end{align*}
for $\timeBoundAlt = \timeBound + \bigO{1}$.
\end{prop}

\begin{proof}

First recall that $(\AdvS{1},
\AdvS{2})$ receives no output from \embedAndMap and so 
does not learn $\prpOutputEmbedding$ or $\prpOutputHash$ 
individually. \testAndCommit outputs to server only $\prfOutput 
= \PRP[\PRPKey[1]]^{-1}(\prpOutputEmbedding) \linearComb
\PRP[\PRPKey[2]]^{-1}(\prpOutputHash) = \embed[\embedKey]
(\cInput) \linearComb \sMask \hadamardProd \hash(\cInput, \embed
[\embedKey](\cInput), \embedKey, \cMask)$. 
$\hash$ is a random oracle that maps $\{0,1\}^{\ast}$ into $\metricSpaceESP$,
where $\finiteField$ is a finite field. $\prfOutput \in \finiteField^{\numberOfPointsSim}$ 
is a vector of field elements, so the addition is a linear combination
in the finite field, which is a secure one-time-pad with perfect secrecy.
$\prfOutput$ carries no information about $\cInput$ or $\embed[\embedKey](\cInput)$
since $\cMask$ is kept secret from \server.
No matter the properties of $\hash$, the $(\AdvS{1}, 
\AdvS{2})$ adversary can win the above experiment with
probability at least $\frac{\hashOracleQueries}{\setSize{\metricSpace}}$, simply by
invoking $\hash(\cInput, \embed[\embedKey](\cInputAlt),
\embedKey, \cMaskAlt)$ and
guessing elements $\cInput$ of \universe and the secret $\cMask \in \metricSpace$.
\propref{prop:serverAdv} says that if $\hash$ is a random oracle, then
this is the best that $(\AdvS{1}, \AdvS{2})$ can do.
As such, the probability the adversary 
succeeds is as stated in the proposition.

\addText{
The experiment $\ExptS$ accounts for \server maliciously trying to learn $\cInput$.
However, we do not claim this because doing so would require $\embedAndMap, 
\testAndCommit, \authentication$ to be maliciously secure and they are not in our 
implementation due to efficiency reasons. In other words, through instantiating
$\embedAndMap, \testAndCommit, \authentication$ with maliciously secure
primitives, \primitiveNameShort can be secure against a deviating \server
trying to learn $\client[1]$ and $\client[2]$'s inputs.
}

\end{proof}

\section{Full Proofs of \secref{sec:app}}
\label{app:proofs_pake}

\addText{
In this appendix, we prove our instantiation of the framework.
}

\begin{theorem}[\cite{Chakraborti2023:DAPSI}]
	\label{thm:dapsi}
	Let \finiteField be a finite field of order \fieldOrder. Fix polynomials $\genericPoly[1]$ and $\genericPoly[2]$ of degree $\degreeOfKey{1}$, and an arbitrary set of points $\setOfEvalPts \gets \{\evalPt[1], \dots, \evalPt[\numberOfPointsSim]\}$. Let $\simRandomPoly$ and $\simRandomPolyAlt$ be random polynomials of degree $\geq \degreeOfKey{1}$. Also, let $\genericPolyAlt \gets \prod\limits_{x_i \getsr \finiteField} (x - x_i)$. Then, when $\degreeOfPoly(\gcd(\genericPoly[1], \genericPoly[2])) < \degreeOfKey{1} - \frac{\threshold}{2}$ and $\numberOfPointsSim = \numberOfPointsSimDefn$,

	\begin{multline}
		\sum\limits_{Y \gets \finiteField^{\numberOfPointsSim}} | \prob{  \langle \simRandomPolyEval[\evalPt[\ptIdx]] \cdot \genericPolyEval[\evalPt[\ptIdx]][1] + \simRandomPolyAltEval[\evalPt[\ptIdx]] \cdot \genericPolyEval[\evalPt[\ptIdx]][2] \rangle_{\ptIdx = 1}^{\numberOfPointsSim} = Y} \\ 
		-   \prob{  \langle \simRandomPolyEval[\evalPt[\ptIdx]] \cdot \genericPolyEval[\evalPt[\ptIdx]][1] + \simRandomPolyAltEval[\evalPt[\ptIdx]] \cdot \genericPolyAltEval[\evalPt[\ptIdx]] \rangle_{\ptIdx = 1}^{\numberOfPointsSim} = Y} | \leq \frac{1}{\setSize{\finiteField}}
	\end{multline}	
\end{theorem}

\begin{theorem*}
	Assuming that there is a protocol that securely realizes a garbled
	circuit, \protocolEmbed securely realizes \embedAndMap against a
	semi-honest server $\server$, and a malicious client $\client[1]$.
\end{theorem*}

\begin{proof}
	The entire functionality is implemented in a garbled circuit 
construction that is generated by the semi-honest \server. The 
party $\client[1]$ playing the role of the evaluator is malicious, however 
it does not get outputs from the circuit unless it correctly evaluates 
it. Assuming that there is a construction realizing \otFunc with a
semi-honest garbler and a malicious evaluator, \protocolEmbedESP 
can be proved to be secure in the \otFunc-hybrid model. 
\end{proof}

\begin{theorem*}
	Assuming that there are protocols that securely realize \oleFunc and \npaFunc, 
	\protocolCommit securely realized \testAndCommit against
	a semi-honest server $\server$, and a malicious client $\client[1]$ in the $(\oleFunc, \npaFunc)$-hybrid model. 

\end{theorem*}

\begin{proof}
	We will show that there is a polynomial time simulator that simulates $\client[1]$'s and \server's views.
	
	\smallskip\noindent   
	{\bf When \server is corrupt:} The simulator obtains $\PRPKey[1] \assign \langle \PRPKeyVec[1], \PRPKeyVec[2] \rangle$  and $\PRPKey[2] \assign 
	\langle \PRPKeyHashVec[1], \PRPKeyHashVec[2] \rangle$,  and $\serversRandomPoly[1], \dots, \serversRandomPoly[\blockListSize]$  from \server's  random tape.
	Let $\PRPKeyVec[1] \assign \langle \PRPKeyVec[1][1], \dots, \PRPKeyVec[1][\numberOfPointsSim] \rangle$ and 
	$\PRPKeyVec[2] \assign \langle \PRPKeyVec[2][1], \dots, \PRPKeyVec[2][\numberOfPointsSim] \rangle$. Also, let 
	$\PRPKeyHashVec[1] \assign \langle \PRPKeyHashVec[1][1], \dots, \PRPKeyHashVec[1][\numberOfPointsSim] \rangle$ and 
	$\PRPKeyHashVec[2] \assign \langle \PRPKeyHashVec[2][1], \dots, \PRPKeyHashVec[2][\numberOfPointsSim] \rangle$.
	
	 For $\setIdx \in [1, \blockListSize]$, the simulator obtains $\lanElmt_\setIdx \in \policyLan$ from $\server$'s input tape; Recall \server is assumed to be semi-honest. 
	 Let $\lanElmt_\setIdx = \langle \lElmtVecComp{\setIdx}{1}, \dots, \lElmtVecComp{\setIdx}{\numberOfPointsSim} \rangle$ and let $\polyFromEmbedding[\lanElmt_\setIdx]$ be the polynomial obtained by interpolating the points $\{ (\evalPt[1], \lElmtVecComp{\setIdx}{1}), \dots, (\evalPt[\numberOfPointsSim], \lElmtVecComp{\setIdx}{\numberOfPointsSim}) \}$.

	First consider the case where \server's output tape has $\gcd(\polyFromEmbedding[\cInput], \polyFromEmbedding[\lanElmt_\setIdx])$ for $\setIdx \in [1, \blockListSize]$. There are two cases here 
	
	\begin{enumerate}
		\item \label{case:gcd_not_found} \textit{For $\rootIdx \neq \setIdx$:} When \server's output tape does not have $\gcd(\polyFromEmbedding[\cInput],  \polyFromEmbedding[\lanElmt_\rootIdx])$, then the degree of $\gcd(\polyFromEmbedding[\cInput],  \polyFromEmbedding[\lanElmt_\rootIdx]) < \degreeOfKey{1} - \threshold/2$. In this case, 
		the simulator sets $\simOp{\genericSetVar[\cInput]} \assign \{ \genericRand[1], \dots, \genericRand[\degreeOfKey{1}] \}$, where $\genericRand[1], \dots, \genericRand[\degreeOfKey{1}] \getsr \finiteField$, and the polynomial $\simPolyFromEmbedding[\cInput] \assign \prod\limits_{\genericSetElmt \in \simOp{\genericSetVar[\cInput]}} (x - \genericSetElmt)$. The simulator simulates \npaFunc returning evaluations of 
		$\langle \PRPKeyVec[1][\ptIdx] \cdot \left(   \serversRandomPoly[\rootIdx][\evalPt[\ptIdx]] \cdot \simPolyFromEmbedding[\cInput][\evalPt[\ptIdx]]  + \simClientsRandomPoly[\rootIdx][\evalPt[\ptIdx]]  \cdot \polyFromEmbedding[\lanElmt_\rootIdx][\evalPt[\ptIdx]]  \right) \\
+ \serversRandomPoly[\rootIdx][\evalPt[\ptIdx]] \cdot \PRPKeyVec[2][\ptIdx] \rangle_{\ptIdx = 1}^{\numberOfPointsSim}$ 
		to \server, where $\simClientsRandomPoly[\rootIdx]$ is a random polynomial of degree $\degreeOfKey{1}$. The simulation is indistinguishable from the 
		real protocol execution: $\simClientsRandomPoly[\rootIdx]$ is sampled from the same distribution as $\clientsRandomPoly[\rootIdx]$, and 
		from  \thmref{thm:dapsi}, $\left\langle \serversRandomPoly[\rootIdx][\evalPt[\ptIdx]] \cdot \simPolyFromEmbedding[\cInput][\evalPt[\ptIdx]]  + \simClientsRandomPoly[\rootIdx][\evalPt[\ptIdx]]  \cdot \polyFromEmbedding[\lanElmt_\rootIdx][\evalPt[\ptIdx]]  \right\rangle_{\ptIdx = 1}^{\numberOfPointsSim}$ is indistinguishable from $\left\langle \serversRandomPoly[\rootIdx][\evalPt[\ptIdx]] \cdot \polyFromEmbedding[\cInput] [\evalPt[\ptIdx]] + \clientsRandomPoly[\rootIdx][\evalPt[\ptIdx]] \cdot \polyFromEmbedding[\lanElmt_\rootIdx][\evalPt[\ptIdx]]  \right\rangle_{\ptIdx = 1}^{\numberOfPointsSim} $.

		\item \label{case:gcd_found} \textit{For $\rootIdx = \setIdx$:} When \server's output tape has $\gcd(\polyFromEmbedding[\cInput], \polyFromEmbedding[\lanElmt_\rootIdx])$, let $\genericSetVar[\cInput \cap \lanElmt_\rootIdx]$ be the set of the roots of $\gcd(\polyFromEmbedding[\cInput], \polyFromEmbedding[\lanElmt_\rootIdx])$ and $\setSize{\genericSetVar[\cInput \cap \lanElmt_\rootIdx]} = d$. Then, 
		the simulator computes a set $\simOp{\genericSetVar[\cInput]} \assign \{ \genericRand[1], \dots, \genericRand[\degreeOfKey{1} - d] \} \cup \genericSetVar[\cInput \cap \lanElmt_\rootIdx]$ where $\genericRand[1], \dots, \genericRand[\degreeOfKey{1} - d] \notin \genericSetVar[\cInput \cap \lanElmt_\rootIdx]$. The simulator sets $\simPolyFromEmbedding[\cInput] \assign \prod\limits_{\genericSetElmt \in \simOp{\genericSetVar[\cInput]}} (x - \genericSetElmt)$. The simulator simulates \npaFunc returning evaluations of 
		$\langle \PRPKeyVec[1][\ptIdx] \cdot \left(   \serversRandomPoly[\rootIdx][\evalPt[\ptIdx]] \cdot \simPolyFromEmbedding[\cInput][\evalPt[\ptIdx]]  + \simClientsRandomPoly[\rootIdx][\evalPt[\ptIdx]]  \cdot \polyFromEmbedding[\lanElmt_\rootIdx][\evalPt[\ptIdx]]  \right) \\
+ \serversRandomPoly[\rootIdx][\evalPt[\ptIdx]] \cdot \PRPKeyVec[2][\ptIdx] \rangle_{\ptIdx = 1}^{\numberOfPointsSim}$ 
		to \server, where $\simClientsRandomPoly[\rootIdx]$ is a random polynomial of degree $\degreeOfKey{1}$. Let $\serversRandomPoly[\rootIdx] \cdot \polyFromEmbedding[\cInput]  + \clientsRandomPoly[\rootIdx] \cdot \polyFromEmbedding[\lanElmt_\rootIdx]  = \gcd(\polyFromEmbedding[\cInput], \polyFromEmbedding[\lanElmt_\rootIdx]) \cdot R(x)$ and 
		$\serversRandomPoly[\rootIdx] \cdot \simPolyFromEmbedding[\cInput]  + \simClientsRandomPoly[\rootIdx]  \cdot \polyFromEmbedding[\lanElmt_\rootIdx] = \gcd(\polyFromEmbedding[\cInput], \polyFromEmbedding[\lanElmt_\rootIdx]) \cdot \simOp{R}(x)$ 
		Since $\serversRandomPoly[\rootIdx]$, $\clientsRandomPoly[\rootIdx]$ and $\simClientsRandomPoly[\rootIdx]$ are all random polynomials, 
		due to \lemmaref{lemma:kissener_uniformly_random}, $R(x)$ and  $\simOp{R}(x)$ are both random polynomials of the same degree, and are therefore indistinguishable. Thus,  $\serversRandomPoly[\rootIdx] \cdot \polyFromEmbedding[\cInput]  + \clientsRandomPoly[\rootIdx] \cdot \polyFromEmbedding[\lanElmt_\rootIdx]$ is indistinguishable from 	$\serversRandomPoly[\rootIdx] \cdot \simPolyFromEmbedding[\cInput]  + \simClientsRandomPoly[\rootIdx]  \cdot \polyFromEmbedding[\lanElmt_\rootIdx]$, and consequently the simulation is indistinguishable 
		from the real execution. 
		
	\end{enumerate}
	
	In the case where there is no $\lanElmt_\setIdx$ where \server's output tape has $\gcd({\polyFromEmbedding[\cInput], \polyFromEmbedding[\lanElmt_\setIdx]})$, the simulator follows case \eqnref{case:gcd_not_found} for all $\setIdx \in [1, \blockListSize]$.

	In the commit phase, the simulator implements the random oracle, and sets $\langle \simOp{\hashComp{1}}, \dots, \simOp{\hashComp{\numberOfPointsSim}} \rangle \getsr 
	\metricSpaceESP$. The simulator computes $\langle  \simOp{\prpOutputHashComp{1}}, \dots, \simOp{\prpOutputHashComp{\numberOfPointsSim}}   \rangle \assign \PRP[\PRPKey[2]](\langle \simOp{\hashComp{1}}, \dots, \simOp{\hashComp{\numberOfPointsSim}} \rangle)$, where the simulator gets $\PRPKey[2]$ from the random tape of \server. The simulator simulates the calls to \oleFunc in \stepsref{esp:tc:remove_client_random}{esp:tc:prp_inv_hash} with inputs 
	$\simClientsRandomPoly[\setIdx][\evalPt[\ptIdx]]$ and $\simOp{\prpOutputHashComp{\ptIdx}}$ respectively. The simulator does not need to simulate an output to \server. 
	
	Then, for \stepref{esp:tc:client_comb} the simulator returns 
	$\simOp{\genericVec'} \assign 
	\langle \sum\limits_{\setIdx = 1}^{\blockListSize} \tmpVar_{\setIdx, \ptIdx} + \tmpVar'_{\setIdx, \ptIdx} -  
	\simClientsRandomPoly[\setIdx] \cdot \polyFromEmbedding[\lanElmt_\setIdx][\evalPt[\ptIdx]]+\frac{\randomPolyComb[\evalPt[\ptIdx]]}{\PRPKeyHashVec[1][\ptIdx]} \cdot \simOp{\prpOutputHashComp{\ptIdx}}  \rangle_{\ptIdx = 1}^{\numberOfPointsSim}$ to \server. $\simOp{\genericVec'}$ is indistinguishable from $\genericVec'$ as $\simClientsRandomPoly[\setIdx]$ is identically distributed 
	to $\clientsRandomPoly[\setIdx]$, and $\langle \simOp{\hashComp{1}}, \dots, \simOp{\hashComp{\numberOfPointsSim}} \rangle$ is indistinguishable from $\langle \hashComp{1}, \dots, \hashComp{\numberOfPointsSim} \rangle$ assuming that $\hash(\cdot)$ is a random oracle.

	\smallskip\noindent
	{\bf When $\client[1]$ is corrupt:} 
	We will prove the protocol secure in the (\oleFunc, \npaFunc)-hybrid model. The simulator does the following
	
	\begin{enumerate}
		\item During \stepref{esp:tc:compute_comb_poly} of \protocolCommitESP, the simulator obtains \genericAdv's inputs
		$\langle \clientsRandomPoly[\setIdx][\evalPt[\ptIdx]] \rangle_{\ptIdx = 1}^{\numberOfPointsSim}$ and $\prpOutputEmbeddingAlt$. The simulator simulates 
		\npaFunc with $\client[1]$'s inputs. There is no output to \genericAdv from this step.  
		
		\item In \stepref{esp:tc:remove_client_random}, the simulator gets 
		$\clientsRandomPoly[\setIdx][\evalPt[\ptIdx]]^{\ast}$ from \genericAdv's input to \oleFunc in  for $\setIdx \in [1, \blockListSize], \ptIdx \in [1, \numberOfPointsSim]$. If $\clientsRandomPoly[\setIdx][\evalPt[\ptIdx]]^{\ast} = \clientsRandomPoly[\setIdx][\evalPt[\ptIdx]]$ (obtained earlier), then the simulator sends $\clientsRandomPoly[\setIdx][\evalPt[\ptIdx]]$ to \oleFunc, obtains $out_{\setIdx, \ptIdx}$ from \oleFunc, and simulates \oleFunc returning $out_{\setIdx, \ptIdx}$ to $\client[1]$. 
		Otherwise, the simulator sets $\simOp{out_{\setIdx, \ptIdx}} \getsr \finiteField$ and simulates \oleFunc returning $\simOp{out_{\setIdx, \ptIdx}}$ to $\client[1]$. The simulator outputs whatever \genericAdv outputs. 
		
		\item In \stepref{esp:tc:prp_inv_hash}, the simulator gets $ \prpOutputHashComp{1}', \dots, \prpOutputHashComp{\numberOfPointsSim}'$ as inputs to \oleFunc.  The simulator simulates returning the output of \oleFunc to $\client[1]$, and outputs whatever \genericAdv outputs.

			\item The simulator gets $\genericVec'$ from \genericAdv in \stepref{esp:tc:client_comb}. The simulator checks whether $\genericVec'$ has been correctly computed as the sum of the values returned to $\client[1]$ in the previous steps. If so, the simulator sets $\prpOutputHash'' \assign \prpOutputHash'$. Otherwise, for each $\rootIdx \in [1, \numberOfPointsSim]$ where $\genericVec'_{\rootIdx}$ has not been correctly computed, $\prpOutputHashComp{\rootIdx}'' \getsr \finiteField$.

		\item The simulator sends to \testAndCommit the inputs $\prpOutputEmbeddingAlt$ and $\prpOutputHash''$. The functionality returns $\prfOutput$ or $\bot$. If the output of \testAndCommit is $\bot$, then send 0 to $\client[1]$. Otherwise, send 1 to $\client[1]$. Output whatever $\genericAdv$ outputs.

%
		
	\end{enumerate}

	The simulation is indistinguishable from the real execution of the protocol. In \stepref{esp:tc:rand_poly_C}, 
	$out_{\setIdx, \ptIdx}$ and $\simOp{out_{\setIdx, \ptIdx}}$ are indistinguishable. This is because in the real execution of $\protocolCommitESP$, $out_{\setIdx, \ptIdx} = \tmpVar'_{\setIdx, \ptIdx} - 
	\clientsRandomPoly[\setIdx][\evalPt[\ptIdx]] \cdot \polyFromEmbedding[\lanElmt_\setIdx][\evalPt[\ptIdx]] \getsr \finiteField$, when $\tmpVar'_{\setIdx, \ptIdx} \getsr \finiteField$, which is indistinguishable from $\simOp{out_{\setIdx, \ptIdx}} \getsr \finiteField$. In \stepref{esp:tc:rand_poly_S},, the simulator simply simulates 
	\oleFunc using \genericAdv's input and thus follows the protocol exactly.

	In \stepref{esp:tc:compute_vec},  first note that the ideal functionality outputs $\bot$ when $\symmURDist{\embed[\embedKey](\cInput)}{\lanElmt\ } \le \threshold$, which equivalent to checking that the degree of $\gcd(\polyFromEmbedding[\cInput], \polyFromEmbedding[\lanElmt]) \ge \degreeOfKey{1} - 
	\frac{\threshold}{2}$. On the other hand, the simulation outputs $\bot$ when the degree of $\combPoly \assign \gcd(\polyFromEmbedding[\cInput], \serversRandomPoly[\setIdx] \cdot  \polyFromEmbedding[\cInput]  + \clientsRandomPoly[\setIdx]   \cdot \polyFromEmbedding[\lanElmt]) \ge \degreeOfKey{1} - 
	\frac{\threshold}{2}$. From 
	\propref{prop:symmURFromRand}, we have 
	
	\[
	\prob{\symmURDist{\embedESP(\cInput)}{\lanElmt\ } \neq 
	2\cdot \degreeOfKey{1} - \degreeOfPoly(\combPoly)} \le 1/\setSize{\finiteField}
	\]
	
%

	Thus, with overwhelming probability, the output of the simulation to $\client[1]$ and the real protocol execution are identical.
	
	Next, note that when $\genericVec'$ has been correctly computed, the simulator provides $\prpOutputEmbeddingAlt$ and $\prpOutputHashAlt$ to 
\testAndCommit, which returns $\prfOutput \assign \PRP[\PRPKey[1]]^{-1}\left( \prpOutputEmbeddingAlt \right) + \PRP[\PRPKey[2]]^{-1}\left( \prpOutputHashAlt\right)$. In the real execution, $\prfOutput$ is identically computed. When \genericAdv provides inconsistent value of $\genericVec'$, either due to the fact that $\genericVec'$ is not correctly computed as a sum of the results returned before to $\client[1]$, or 
	because $\clientsRandomPoly[\setIdx][\evalPt[\ptIdx]]^{\ast} \neq \clientsRandomPoly[\setIdx][\evalPt[\ptIdx]]$  for some $\ptIdx, \setIdx$, 
	then the simulator sends to \testAndCommit 
$\prpOutputHash'' \neq \prpOutputHashAlt$. Let 
$\prpOutputHash''$ and $\prpOutputHashAlt$ differ in the $\rootIdx$-th component. Then, $\PRP[\PRPKey[1]]^{-1}(\prpOutputEmbeddingAlt) + \PRP[\PRPKey[2]]^{-1}(\prpOutputHashAlt)$ and $\PRP[\PRPKey[1]]^{-1}(\prpOutputEmbeddingAlt) + \PRP[\PRPKey[2]]^{-1}(\prpOutputHash'')$ also differ in the $\rootIdx$-component, and the value of the $\rootIdx$-th component of $\PRP[\PRPKey[2]]^{-1}(\prpOutputHash'')$ is uniformly random to $\client[1]$. Thus, the value $\prfOutput$ in the real protocol is indistinguishable from the output of the simulation. 


\end{proof}

\begin{theorem*}
	Assuming that there is a protocol to securely realize \oleFunc,
	\protocolValidate securely realizes \authentication against a
	semi-honest server \server and a malicious 
	client $\client[1]$.
\end{theorem*}

\begin{proof}
	We will show that there is a polynomial time simulator that indistinguishably simulates \server's and $\client[1]$'s views of the protocol in the ideal world simulation.
	
	\smallskip\textbf{When \server is corrupt:}
	It is straightforward to simulate \server's view when interacting 
  	with \oleFunc. Specifically, for each $\ptIdx \in [1, \numberOfPointsSim]$, 
	the simulator simulates \oleFunc with \server's input $\sMask_\ptIdx$ 
	(which it gets from \server's input tape), and outputs a value 
	$\genericRand[\ptIdx] \getsr \finiteField$. The simulation is 
	indistinguishable from the protocol execution if \oleFunc can be 
	simulated since $\server$'s output, $\prfOutput$, in the real protocol 
	execution is identically distributed to the simulated output 
	assuming that $\hash(\cdot)$ is a random oracle.

	
	\textbf{When $\client[1]$ is corrupt:}
	In the protocol, $\client[1]$ receives 1 from \server when the output 
	of \oleFunc is equal to $\prfOutput$, otherwise it receives 0. Since 
	$\prfOutput = \embed[\embedKey](\cInput) + \sMask \times \hash
	(\cInput, \embed[\embedKey](\cInput), \embedKey, \cMask)$ is 
	uniformly distributed in $\metricSpace$, the probability that for 
	some $\cICInput[1] \neq \embed[\embedKey](\cInputAlt) $ and 
	$\cICInput[2] \neq \hash(\cInput, \embed[\embedKey](\cInput), 
	\embedKey, \cMask)$, $\cICInput[1] + \sMask \times \cICInput[2] = 
	\prfOutput$ is $\le 1/\setSize{\finiteField}$. Therefore, with 
	overwhelming probability, when $\client[1]$ provides incorrect inputs 
	to \oleFunc in the real protocol execution, the output to $\client[1]$ 
	in the real protocol execution is 0. 

	In the ideal world, the simulator obtains the inputs $\cICInput[1]$ and 
	$\cICInput[2]$ to \oleFunc, and sends these values to \authentication. 
	The output to $\client[1]$ is identically distributed in the real world and 
	the ideal world. Thus, $\client[1]$'s view in the simulation is 
	indistinguishable assuming that there is a simulator that can simulate 
	\oleFunc against a malicious sender.

\end{proof}

\end{document}